\documentclass[authoryear,review]{article}
\usepackage{appendix}
\usepackage{setspace}
\usepackage[font=small,  margin=2cm]{caption}
\usepackage[tbtags]{amsmath}
\usepackage{amsthm,amssymb,amsfonts}
\usepackage[numbers]{natbib}
\usepackage{multirow}
\usepackage{subcaption}
\usepackage{lscape}
\usepackage{makecell}
\usepackage{exscale}
\usepackage{booktabs}
\usepackage{comment}
\usepackage{array}
\usepackage{fullpage}
\usepackage{url}
\usepackage{algorithm}
\usepackage{algpseudocode}
\usepackage{bm}
\usepackage{smile}
\usepackage{mathtools}
\usepackage{wrapfig}
\usepackage{lipsum}
\usepackage{mathrsfs}
\usepackage{relsize}
\usepackage{dsfont}
\usepackage{multirow}
\usepackage{apalike}
\usepackage{chngcntr}
\usepackage[usenames,dvipsnames,svgnames,table]{xcolor}
\usepackage[colorlinks=true,linkcolor=blue,urlcolor=blue,citecolor=blue]{hyperref}

\numberwithin{equation}{section}
\numberwithin{theorem}{section}
\numberwithin{corollary}{section}
\counterwithout{asmp}{section}
\numberwithin{definition}{section}

\begin{document}

\title{\LARGE Projected Estimation for Large-dimensional Matrix  Factor Models}

	\author{Long Yu\thanks{ School of Management, Fudan University, Shanghai, China;  Email:{\tt fduyulong@163.com}.},~~Yong He\thanks{ Institute for Financial Studies, Shandong University, Jinan, China; Email:{\tt heyong@sdu.edu.cn}.},~~Xin-bing Kong \thanks{Nanjing Audit University, Nanjing,  China.; Email:{\tt xinbingkong@126.com}},~~ Xinsheng Zhang\thanks{ School of Management, Fudan University, Shanghai, China; Email:{\tt xszhang@fudan.edu.cn}.}}	
	\date{}	
	\maketitle
	In this study, we propose a projection estimation method for large-dimensional matrix factor models with cross-sectionally spiked eigenvalues. By projecting the observation matrix onto the row or column factor space, we simplify factor analysis for matrix series to that for a lower-dimensional tensor. This method also reduces the magnitudes of the idiosyncratic error components, thereby increasing the signal-to-noise ratio, because the projection matrix linearly filters the idiosyncratic error matrix. We theoretically prove that the projected estimators of the factor loading matrices achieve faster convergence rates than existing estimators under similar conditions. Asymptotic distributions of the projected estimators are also presented. A novel iterative procedure is given to specify the pair of row and column factor numbers. Extensive numerical studies verify the empirical performance of the projection method. Two real examples in finance and macroeconomics reveal factor patterns across rows and columns, which coincides with financial, economic, or geographical interpretations.

\vspace{2em}


\section{Introduction}\label{sec1}


Time-series models with factor structures are widely used in fields such as finance, macroeconomics, and machine learning (e.g.  \cite{chamberlain1983arbitrage,fama1993common,stock2002forecasting,fan2015risks}). With the increasing complexity of data structure, the factor models for time series have experienced three stages: factor models for multivariate time series of fixed dimension (e.g. \cite{ross1977capital}),  large-dimensional vector factor models (e.g. \cite{forni2000generalized,bai2002determining,bai2003inferential,ahn2013eigenvalue,fan2013large,kong2019factor,He2020Large}), and  matrix factor models (e.g. \cite{wang2019factor,Chen2020StatisticalIF}). The matrix factor model  not only reveals the serial dynamics for a panel of variables, but also explores the spatial correlations among entries of the observation matrix in a parsimonious way.


\cite{wang2019factor} was the first to introduce a factor model for matrix time series. For ease of presentation, let $\Xb_t$ be a $p_1\times p_2$ matrix of variables observed at $t$. \cite{wang2019factor} factorized $\Xb_t$ as
 \begin{equation}\label{mod1.1}
 (\Xb_t)_{p_1\times p_2}=(\Rb)_{p_1\times k_1}(\Fb_t)_{k_1\times k_2}(\Cb^\top)_{k_2\times p_2}+(\Eb_t)_{p_1\times p_2},\quad t=1,\ldots,T,
 \end{equation}
where $\Rb$ is the $p_1\times k_1$ row factor loading matrix exploiting the variations of $\Xb_t$ across the rows, $\Cb$ is the $p_2\times k_2$ column factor loading matrix reflecting the differences across the columns of $\Xb_t$, $\Fb_t$ is the common factor matrix for all cells in $\Xb_t$, and $\Eb_t$ is the idiosyncratic component of a matrix form. Model (\ref{mod1.1}) is particularly suited to modeling well-structured tables of macroeconomic indicators, financial characteristics, and frames of pictures. For example, Figure \ref{fig1} shows a time list of tables recording the macroeconomic variables across a number of countries. In this example, the dynamics of the panels might be driven by a much lower-dimensional matrix series of composite indices by taking the interrelationship among countries and macroeconomic variables into consideration. The cross-country (column-sectional) and cross-variable (row-sectional) exposures on these latent composite factors can be summarized in $\Rb$ and $\Cb$, respectively. In \cite{wang2019factor}, two interesting interpretations of model (\ref{mod1.1}) on integrating the column and row interactions were explicitly illustrated. The two-step hierarchical interpretation shows that (\ref{mod1.1}) reduces the number of parameters for the vector factor modeling by stacking the columns of $\Xb_t$ from $(p_1p_2+1)k_1k_2$ to $p_1k_1+p_2k_2+k_1k_2$, resulting in a much more parsimonious model.

 \begin{figure}[hbpt]
 		\centering
 		 \scalebox{0.5}{\includegraphics{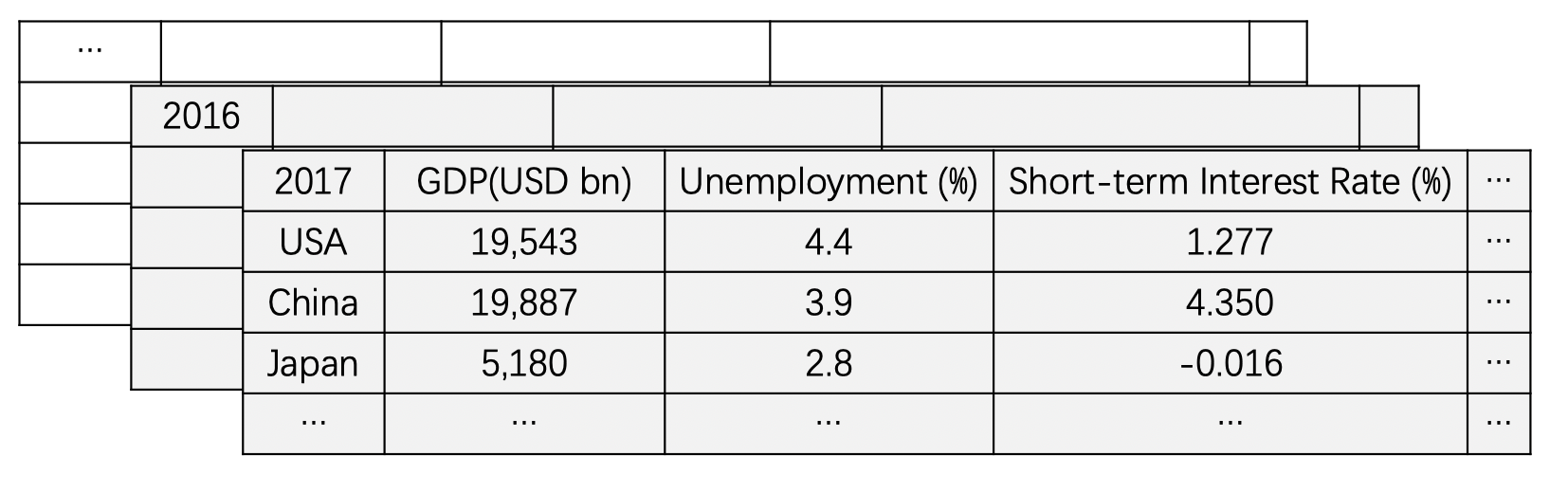}}
 	\caption{A real example of matrix-variate observations consisting of macroeconomic variables for a number of countries.}
 	\label{fig1}
 \end{figure}


Model (\ref{mod1.1}) was later extended to the constrained version by \cite{chen2019constrained} and the threshold matrix factor model in \cite{liu2020threshold}. \cite{Chen2020Modeling} applied model (\ref{mod1.1}) to  the dynamic transport network in applications to international trade flow. Noticeably, all these works are along the line of \cite{lam2011estimation} and \cite{lam2012factor} by implementing an eigen-analysis of the auto-cross-covariance matrix, which relies heavily on the serial correlations of the factors. Other works follow the other line (e.g. \cite{bai2003inferential,fan2013large}) that the matrix factor series influences all the series and hence leads to spiked eigenvalues along the column and row dimensions. For example, \cite{virta2017independent} constructed independent components from low-rank spiked observation matrices, but they assumed a noiseless model. \cite{Chen2020StatisticalIF} extended the approach used by \cite{lettau2020factors} to the matrix factor model and proposed estimators of the factor loading matrices and factor matrices in the model (\ref{mod1.1}). Their estimation is mainly based on the eigen-decomposition of an aggregate of the sample mean matrix and the column (or row) covariance matrices, which are assumed to be pervasive along the two cross-sectional dimensions. However, their spectral method handles the rows or columns of the matrix individually, which does not take full advantage of the joint low rank structure across both the rows and columns. Therefore, we believe that  the efficiency of the estimated row and column factor spaces can still be improved, with the hope to achieve faster convergence rates. A mathematically rigorous comparison of the convergence rates with \cite{Chen2020StatisticalIF} is provided in Section \ref{sec3}.


In this study, we follow the second line by assuming that the factors are pervasive along the two cross-sectional dimensions and adopt similar assumptions on model (\ref{mod1.1}) as in \cite{Chen2020StatisticalIF}.
We propose a projection estimation method for (\ref{mod1.1}). Ideally, if the normalized column loading matrix $\Cb/\sqrt{p_2}$ is known and orthonormal, then the transformed data matrix $\Xb_t\Cb/p_2=\Rb\Fb_t+\Eb_t\Cb/p_2$ simply consists of $k_2$ linear combinations of the columns of $\Rb$ plus $k_2$ linear combinations of the columns of $\Eb_t$. This condition amounts to projecting the columns of $\Xb_t$ onto the product space composed of the columns of $\Rb$ and the $k_2$ factors explaining the columns of $\Xb_t$, plus an error matrix $\Eb_t\Cb/p_2$. An advantage of the projection is that the original error matrix $\Eb_t$ is linearly filtered and the resulting entries are of order $O_p(p_2^{-1/2})$ under Assumptions C and D in Section \ref{sec3}. In summary, the projection simultaneously achieves dimension reduction and denoising, which is in the spirit of constructing principal portfolios in finance to reduce the idiosyncratic risk. As the $k_2$ columns of the transformed data matrix $\Xb_t\Cb/p_2$ all lay in the column space of $\Rb$ asymptotically under a mild condition, in a second step, a simple principal component analysis on the projected data matrix yields an estimator of $\Rb$. For the example illustrated in Figure \ref{fig1}, the projection procedure amounts to first summarizing the $k_2$ factors behind the macroeconomic indicators ($k_2$ columns of $\Fb_t$), and then recovering $\Rb$ from the variations of $\Rb\Fb_t$ across countries. Through the above two steps, the low-rank structure along the two cross-sectional dimension is exploited in a succeeding manner. The preceding argument is heuristic. In practice, $\Cb$ is unknown and has to be initially estimated. Indeed, the estimator of $\Cb$ proposed by \cite{Chen2020StatisticalIF} serves as a good projection matrix, and we describe it in Section \ref{sec2}. By simply applying the procedure to $\Xb_t^\top$, we can estimate $\Cb$ in the same manner.

In this study, we theoretically prove that our projection estimators improve
the convergence rates of those in \cite{Chen2020StatisticalIF}. Figure \ref{fig2} in the simulation study clearly illustrates the improvement.
We also propose an iterative algorithm to consistently determine the pair of column and row factor numbers, which performs impressively well in  numerical studies.


The remainder of this paper is organized as follows. Section \ref{sec2} introduces the model setup and our projection approach. In section \ref{sec3}, we present the technical assumptions and asymptotic results, including the convergence rates and limiting distributions.  Section \ref{sec4} is devoted to numerical studies. Two real data examples are provided in section \ref{sec5}. Section \ref{sec6} concludes and discusses possible future works.


To end this section, we introduce some notations used throughout the study. For a matrix $\Xb_t$ observed at time $t$, $x_{t,ij}$ denotes its $ij$-th entry,  $\bx_{t,i\cdot}$ $(\bx_{t,\cdot j}$) denotes its $i$-th row ($j$-th column) and let $\text{Vec}({\Xb_t})$ be the vector obtained by stacking the columns of $\Xb_t$. For a matrix $\Ab$, $\|\Ab\|$ and $\|\Ab\|_F$ represent the spectral norm and Frobenious norm, respectively. $\|\Ab\|_{\max}$ is the maximum of $|\mathrm{A}_{ij}|$'s.   $\lambda_j(\Ab)$  is the $j$-th eigenvalue of $\Ab$ if $\Ab$ is symmetric. The notations $\stackrel{p}{\rightarrow},\stackrel{d}{\rightarrow}$ and $\stackrel{a.s.}{\rightarrow}$ represent convergence in probability, in distribution and almost surely, respectively. The $o_p$ is for convergence to zero in probability and $O_p$ is for stochastic boundedness. For two random series $X_n$ and $Y_n$, $X_n\lesssim Y_n$ means that $X_n=O_p(Y_n)$, and $X_n\gtrsim Y_n$ means that $Y_n=O_p(X_n)$. The notation $X_n\asymp Y_n$ means that $X_n\lesssim Y_n$ and $X_n\gtrsim Y_n$. For two random vectors $\bX,\bY$, $\bX\perp\bY$ means that $\bX$ and $\bY$ are independent. $[n]$ denotes the set $\{1,\ldots,n\}$. $\otimes$ denotes the Kronecker product. The constant $c$ may not be  identical in different lines.

\section{Matrix factor model and projected estimators}\label{sec2}
\subsection{Matrix factor model}

The model (\ref{mod1.1}) factorizes each matrix as a low-rank common component plus an idiosyncratic component, which can be regarded as an extension of the vector factor model to the matrix regime. It provides a new framework and interpretation for the analysis of 3D tensor data. The loading matrices $\Rb$ and $\Cb$ in model (\ref{mod1.1}) are not separately identifiable. In the current paper, only the loading spaces are of interest, and thus we assume without loss of generality that
\begin{equation}\label{IDC}
\|p_1^{-1}\Rb^\top\Rb-\Ib_{k_1}\|\rightarrow 0,\quad\text{and}\quad \|p_2^{-1}\Cb^\top\Cb-\Ib_{k_2}\|\rightarrow 0.
\end{equation}
If this is not the case, then two matrices $\Qb_1$ and $\Qb_2$ will always exist with orthogonal columns, such that
\[
\Rb=\Qb_1\Wb_1 \quad\text{and}\quad \Cb=\Qb_2\Wb_2,
\]
where $\Wb_1$ and $\Wb_2$ are $k_1\times k_1$ and $k_2\times k_2$ full rank matrices, respectively. Therefore, $\Rb$ (or $\Cb$) lies in the same column space as $\Qb_1$ (or $\Qb_2$), and $\Xb_t$ can be rewritten as
\[
\Xb_t=(\sqrt{p_1}\Qb_1)\tilde\Fb_t(\sqrt{p_2}\Qb_2)^\top+\Eb_t,\quad\text{where}\quad \tilde\Fb_t=\frac{1}{\sqrt{p_1p_2}}\Wb_1\Fb_t\Wb_2^\top.
\]
Then, $\Xb_t$ becomes a matrix factor model with row and column loading matrices satisfying (\ref{IDC}).
Assumption (\ref{IDC}) is not only an identifiability condition for the factor loading spaces, but also a strong factor condition assuming pervasive factors along the row and column dimensions.

\subsection{Projected estimation}
In this section, we introduce our projection estimation approach. As a simple heuristic argument,  $\Cb$ is assumed to be known and satisfies the orthogonal condition $\Cb^\top\Cb/p_2=I_{k_2}$. As stated in the introduction, we project the data matrix to a lower dimensional space by setting
\begin{equation}\label{equ3.1.1}
\Yb_t=\frac{1}{p_2}\Xb_t\Cb=\frac{1}{p_2}\Rb\Fb_t\Cb^\top\Cb+\frac{1}{p_2}\Eb_t\Cb:=\Rb\Fb_t+\tilde\Eb_t:=\Rb(\bbf_{t, \cdot 1},..., \bbf_{t, \cdot k_2})+(\tilde{\be}_{t, \cdot 1},...,\tilde{\be}_{t, \cdot k_2}).
\end{equation}
After transformation, $\Yb_t$ is a $p_1\times k_2$ matrix-valued observation, lying in a much lower column space than $\Xb_t$. $\Fb_t$ and $\tilde\Eb_t$ can be regarded as factors and errors for $\Yb_t$. When $k_2=1$, it is exactly a vector factor model. As a result, the projection achieves dimension reduction of the error matrix and decrease of the noise levels. For the $i$th row of $\tilde\Eb_t$, denoted as $\tilde\be_{t,i\cdot}$, $\mathbb{E}\|\tilde\be_{t,i\cdot}\|^2\le cp_2^{-1}$  as long as the original errors $\{\be_{t,ij}\}^{p_2}_{j=1}$ are weakly dependent column-wise. When $p_2$ is large enough, $\Yb_t$ can be treated as a nearly noise-free factor model with $O(p_1)$ loading parameters to be estimated.

Given $\Yb_t$, we define
\[
\Mb_1=\frac{1}{Tp_1}\sum_{t=1}^{T}\Yb_t\Yb_t^\top,
\]
and then the row factor loading matrix $\Rb$ can be estimated by the leading $k_1$ eigenvectors of $\Mb_1$. Heuristically, under some mild conditions,
\begin{eqnarray}\label{M1p}
\Mb_1=\frac{1}{Tp_1}\sum^T_{t=1}\sum^{k_2}_{j=1}\yb_{t,\cdot j}\yb_{t, \cdot j}^\top\approx\frac{1}{p_1}\Rb\left(\frac{1}{T}\sum^T_{t=1}\sum^{k_2}_{j=1}\bbf_{t,\cdot j}\bbf_{t, \cdot j}^\top\right)\Rb^\top+\frac{1}{Tp_1}\sum^T_{t=1}\sum^{k_2}_{j=1}\tilde{\be}_{t, \cdot j}\tilde{\be}_{t, \cdot j}^\top.
\end{eqnarray}
By looking at the $k_2$ columns of $\Yb_t$ as observations within time unit $t$, (\ref{equ3.1.1}) and (\ref{M1p}) demonstrate that $\{\Yb_t\}_{t=1}^T$ is effectively a vector factor model of length $Tk_2$ with asymptotically vanishing idiosyncratic entries.

One problem with the above ideal argument is that the projection matrix $\Cb$ is unavailable in practice. A natural solution is to replace it with a consistent initial estimator $\hat \Cb$. The column factor loading matrix $\Cb$ can be similarly estimated by projecting $\Xb_t$ onto the space of $\Cb$ with transformation matrix $\Rb$ or its estimator $\hat{\Rb}$. The choices of $\hat\Rb$ and $\hat\Cb$ will be discussed later. We summarize the projection procedure in Algorithm \ref{alg1} by starting from $\hat\Rb$ and $\hat\Cb$, which results in an estimated $\Mb_1$ denoted by $\tilde{\Mb}_1$.

\begin{algorithm}[H]
	\caption{Projected method for estimating matrix factor spaces}\label{alg1}
	{\bf Input:} Data matrices $\{\Xb_t\}_{t\le T}$, the pair of row and column factor numbers $k_1$ and $k_2$\\
	{\bf Output:} Factor loading matrices $\tilde\Rb$ and $\tilde\Cb$
	\begin{algorithmic}[1]
		\State  obtain the   initial estimators $\hat \Rb$ and $\hat\Cb$;
		
		\State project the data matrices to lower dimensions by defining $\hat\Yb_t=p_2^{-1}\Xb_t\hat\Cb$ and $\hat\Zb_t=p_1^{-1}\Xb_t^\top\hat\Rb$;
		
		\State given $\hat \Yb_t$ and $\hat\Zb_t$, define $\tilde\Mb_1=(Tp_1)^{-1}\sum_{t=1}^{T}\hat\Yb_t\hat\Yb_t^\top$ and $\tilde\Mb_2=(Tp_2)^{-1}\sum_{t=1}^{T}\hat\Zb_t\hat\Zb_t^\top$, and estimate the loading spaces by  the leading $k_i$ eigenvectors of $\tilde\Mb_i$, denoted as $\tilde\Qb_i, i=1,2$;
		\State the row and column loading matrices are finally given by $\tilde\Rb=\sqrt{p_1}\tilde\Qb_1$ and $\tilde\Cb=\sqrt{p_2}\tilde\Qb_2$.
	\end{algorithmic}
\end{algorithm}

The projection method can be implemented recursively by plugging in the newly estimated $\tilde\Rb$ and $\tilde\Cb$ to replace $\hat{\Rb}$ and $\hat{\Cb}$ in \textbf{Step 2} and iterating \textbf{Steps 2-4}. Theoretical analysis of the recursive solution is challenging. The simulation results in Section \ref{sec4} show that the  projection estimators with a single iteration perform sufficiently well compared with the recursive method. Actually, with $T\asymp p_1\asymp p_2$ and $\hat\Cb$ (or $\hat\Rb$) chosen suitably, we can prove that  the projected estimator $\tilde\Rb$ (or $\tilde{\Cb}$) converges to $\Rb$ (or $\Cb$) after rotation with rate $O_p(1/\sqrt{Tp_2})$ (or {$O_p(1/\sqrt{Tp_1})$)} in terms of  the averaged squared errors, which is the optimal rate even when the loading matrix $\Cb$ (or $\Rb$) is known in advance.

\subsection{Initial projection matrices $\hat\Rb$ and $\hat\Cb$}\label{sec2.3}
The columns of $\Xb_t$ in model (\ref{mod1.1}) can be written in the form of a vector factor model as
\begin{equation}\label{equ2.2}
\bx_{t,\cdot j}=\Rb\Fb_t\bC_{j\cdot}^\top+\be_{t,\cdot j}:=\Rb\overline{\bbf}_{t, \cdot j}+\be_{t,\cdot j},\quad t=1,\ldots,T, \ \ j=1,\ldots,p_2,
\end{equation}
where $\overline{\bbf}_{t, \cdot j}=\Fb_t\bC_{j\cdot}^\top$. Therefore, to estimate $\Rb$, a natural approach is to regard each column as an individual vector observation and apply the conventional PCA method for vector time series. Specifically, we define the scaled column sample covariance matrix as
\[
\hat\Mb_1=\frac{1}{Tp_1p_2}\sum_{t=1}^{T}\sum_{j=1}^{p_2}\bx_{t,\cdot j}\bx_{t,\cdot j}^\top=\frac{1}{Tp_1p_2}\sum_{t=1}^{T}\Xb_t\Xb_t^\top.
\]
When the columns of $\Cb$ are orthogonal and under other mild conditions, approximately
\begin{eqnarray}\label{hatM1}
\hat\Mb_1&\approx & \frac{1}{p_1}\Rb \left(\frac{1}{Tp_2}\sum^T_{t=1}\sum^{p_2}_{j=1}\overline{\bbf}_{t, \cdot j}\overline{\bbf}_{t, \cdot j}^\top\right)\Rb^\top +\frac{1}{p_1}\frac{1}{Tp_2}\sum^T_{t=1}\sum^{p_2}_{j=1}\be_{t,\cdot j}\be_{t, \cdot j}^\top \nonumber\\
&=&\frac{1}{p_1}\Rb \left(\frac{1}{T}\sum^T_{t=1}\sum^{k_2}_{j=1}\bbf_{t, \cdot j}\bbf_{t, \cdot j}^\top\right)\Rb^\top +\frac{1}{Tp_1}\bigg(\sum^T_{t=1}\sum^{p_2}_{j=1}\be_{t,\cdot j}p_2^{-1/2}\be_{t, \cdot j}^\top p_2^{-1/2}\bigg).
\end{eqnarray}
The term $T^{-1}\sum^T_{t=1}\sum^{k_2}_{j=1}\bbf_{t, \cdot j}\bbf_{t, \cdot j}^\top$ typically converges to a symmetric positive definite matrix while the error terms are asymptotically negligible under certain conditions. Consequently, only the leading $k_1$ eigenvalues of $\hat \Mb_1$ are spiky.  Motivated by Davis-Kahan's $\sin(\Theta)$ theorem (\cite{davis1970rotation} and \cite{Yu2015useful}), we find that leading $k_1$ eigenvectors of $\hat\Mb_1$ lie in the same column space of $\Rb$ asymptotically. Therefore, we propose to use the leading $k_1$ eigenvectors of $\hat\Mb_1$ as an estimator of $\Qb_1$, denoted as $\hat\Qb_1$. The row loading matrix is then estimated by $\hat\Rb=\sqrt{p_1}\hat \Qb_1$. The column loading matrix $\Cb$ can be estimated by parallel steps applied to $\{\Xb_t^\top\}_{t\le T}$.

We note that the above initial estimator is  the $\alpha$-PCA solution in \cite{Chen2020StatisticalIF} with $\alpha=0$. A comparison of (\ref{M1p}) and (\ref{hatM1}) shows that $\Mb_1$ and $\hat{\Mb}_1$ have approximately the same covariance matrix of the common components. However, $\hat{\Mb}_1$ accumulates more error terms than $\Mb_1$, thus implying a higher signal-to-noise ratio for $\Yb_t$ than that for $\Xb_t$. This explains the gain of efficiency of our projection estimation method over the initial estimator, or more generally $\alpha$-PCA procedure. Other choices of initial estimates of $\Rb$ and $\Cb$ are admissible as long as two sufficient conditions (\ref{c1}) and (\ref{c2}) in the following section are fulfilled. For simplicity, we  only demonstrate theoretically that the above initial estimators work.

\section{Theoretical Results}\label{sec3}
In this section, we present theoretical results on the convergence rates and asymptotic distributions of the projected estimators. The estimation of factors and common components are also considered. The numbers of factors are treated as given initially, and then we propose an iterative algorithm to consistently estimate the numbers of factors.

\subsection{Technical assumptions}
The matrix factor models are specifically designed for 3D tensor data. The correlation structure for complex high-order tensor data make  the theoretical analysis challenging. Throughout this study, we make the following assumptions on  the correlations across time, row, and column.

\vspace{1em}
\noindent\textbf{Assumption A. Alpha mixing.} The vectorized factor $\text{Vec}(\Fb_t)$ and noise $\text{Vec}(\Eb_t)$ are $\alpha$-mixing. A vector process $\{\bz_t,t=0,\pm1,\pm2,\ldots\}$ is $\alpha$-mixing if, for some $\gamma\ge 2$, the mixing coefficients satisfy the condition that
\[
\sum_{h=0}^\infty\alpha(h)^{1-2/\gamma}<\infty,
\]
where $\alpha(h)=\sup_t\sup_{A\in\mathcal{F}_{-\infty}^t,B\in\mathcal{F}_{t+h}^\infty}|P(A\cap B)-P(A)P(B)|$ and $\mathcal{F}_\tau^s$ is the $\sigma$-filed generated by $\{\bz_t:\tau\le t\le s\}$.

\noindent\textbf{Assumption B. Factor matrix.} The factor matrix satisfies $\mathbb{E}(\Fb_t)={\bf 0}$,  $\mathbb{E}\|\Fb_t\|^4\le c<\infty$ for some constant $c>0$ and
\begin{equation}\label{equ:covariance}
\frac{1}{T}\sum_{t=1}^T\Fb_t\Fb_t^\top\overset{p}{\rightarrow}\bSigma_1 \text{ and }\frac{1}{T}\sum_{t=1}^T\Fb_t^\top\Fb_t\overset{p}{\rightarrow}\bSigma_2,
\end{equation}
where $\bSigma_i$ is $k_i\times k_i$ positive definite matrix with distinct eigenvalues and spectral decomposition $\bSigma_i=\bGamma_i\bLambda_i\bGamma_i^\top$, $i=1,2$. The factor numbers $k_1$ and $k_2$ are fixed as $\min\{T,p_1,p_2\}\rightarrow\infty$. \\

\noindent\textbf{Assumption C. Loading matrix.} Positive constants $\bar r$ and $\bar c$ exists such that $\|\Rb\|_{\max}\le \bar r$, $\|\Cb\|_{\max}\le \bar c$. As $\min\{p_1,p_2\}\rightarrow\infty$, $\|p_1^{-1}\Rb^\top\Rb-\Ib_{k_1}\|\rightarrow 0$ and $	 \|p_2^{-1}\Cb^\top\Cb-\Ib_{k_2}\|\rightarrow 0$.

The $\alpha$-mixing condition in Assumption A allows weak temporal correlations  for both the factors and noises. In Assumption B, the factor matrix is centralized with bounded fourth moment. The condition in (\ref{equ:covariance}) of Assumption B is easily fulfilled under the $\alpha$-mixing assumption, by Corollary 16.2.4 in \cite{athreya2006measure}. The eigenvalues of $\bSigma_i$'s are assumed to be distinct such that the corresponding eigenvectors are identifiable. We assume strong factor conditions in Assumption C, which means that the row and column factors are pervasive along both dimensions. This result is an extension of the pervasive assumption in \cite{stock2002forecasting} to the matrix regime. For identifiability, we assume $\|p_1^{-1}\Rb^\top\Rb-\Ib_{k_1}\|\rightarrow 0$ and $	 \|p_2^{-1}\Cb^\top\Cb-\Ib_{k_2}\|\rightarrow 0$ as $\min\{p_1,p_2\}\rightarrow\infty$. Assumptions A,  B, and C are standard and common in the literature and similar assumptions are adopted by \cite{Chen2020StatisticalIF}, except that the factor matrix is not centralized in their setting.

\noindent\textbf{Assumption D. Weak correlation of noise $\Eb_t$ across column, row, and time.} A positive constant $c<\infty$ {exists} such that
\begin{enumerate}
	\item $\mathbb{E}e_{t,ij}=0$, $\mathbb{E}(e_{t,ij}^8)\le c$.
	\item for any $t\in[T]$, $i\in[p_1]$, $j\in[p_2]$,
	\[
	 (1).\sum_{s=1}^T\sum_{l=1}^{p_1}\sum_{h=1}^{p_2}|\mathbb{E}e_{t,ij}e_{s,lh}|\le c,\quad (2).\sum_{l=1}^{p_1}\sum_{h=1}^{p_2}|\mathbb{E}e_{t,lj}e_{t,ih}|\le c.
	\]
	\item  for any $t\in[T],i,l_1\in [p_1],j,h_1\in[p_2]$,
	\begin{small}
		\[
		\begin{split}
		 (1).&\sum_{s=1}^{T}\sum_{l_2=1}^{p_1}\sum_{h=1}^{p_2}\bigg|\text{Cov}(e_{t,ij}e_{t,l_1j}, e_{s,ih}e_{s,l_2h})\bigg|\le c, \quad\sum_{s=1}^T\sum_{l=1}^{p_1}\sum_{h_2=1}^{p_2}\bigg|\text{Cov}(e_{t,ij}e_{t,ih_1}, e_{s,lj}e_{s,lh_2})\bigg|\le c,\\
		 (2).&\sum_{s=1}^T\sum_{l_2=1}^{p_1}\sum_{h_2=1}^{p_2}\bigg|\text{Cov}(e_{t,ij}e_{t,l_1h_1}, e_{s,ij}e_{s,l_2h_2})\bigg|\le c,\quad \sum_{s=1}^T\sum_{l_2=1}^{p_1}\sum_{h_2=1}^{p_2}\bigg|\text{Cov}(e_{t,l_1j}e_{t,ih_1}, e_{s,l_2j}e_{s,ih_2})\bigg|\le c.
		\end{split}
		\]
	\end{small}
\end{enumerate}

Assumption D is essentially an extension of  Assumption C in \cite{bai2003inferential} to the matrix regime. Similar conditions are adopted by \cite{Chen2020StatisticalIF}. Assumption D.2 (1) allows weak correlation of the noises across time, row and column. It can be a sufficient condition to the Assumptions D.2, E and G.3 in \cite{Chen2020StatisticalIF}. Assumption D.2 (2) further controls the column-wise and row-wise correlation of the noises. Assumption D.3 (1) is similar to the Assumption G.1 in \cite{Chen2020StatisticalIF}, where they require that
\[
\mathbb{E}\bigg\|\frac{1}{\sqrt{Tp_1p_2}}\sum_{t=1}^T\sum_{l=1}^{p_1}\sum_{j=1}^{p_2}\Rb_{l\cdot}\big(e_{t,ij}e_{t,lj}-\mathbb{E}e_{t,ij}e_{t,lj}\big)\bigg\|^2\le c.
\]
Therefore,  the correlation of noises up to the second moment is controlled.
Assumption D.3 (2) is similar to D.3 (1), but on different combinations of noise pairs.  Suppose that the $\{e_{t,ij}\}$'s are located in a 3D space indexed by time, row  and column, Assumption D is satisfied if $e_{t,ij}\perp e_{s,lh}$ as long as the index distance between them is larger than some bandwidth, or the correlation decays sufficiently fast as the distance increases.

\noindent\textbf{Assumption E. Weak dependence between factor $\Fb_t$ and noise $\Eb_t$.} Constant $c>0$ exists such that
\begin{enumerate}
	\item
	for any deterministic vectors $\bv$ and $\bw$ satisfying $\|\bv\|=1$ and $\|\bw\|=1$ with suitable dimensions,
	\[ \mathbb{E}\bigg\|\frac{1}{\sqrt{T}}\sum_{t=1}^T(\Fb_{t}\bv^\top\Eb_t\bw)\bigg\|^2\le c;
	\]
	\item for any $i,l_1\in[p_1]$ and $j,h_1\in[p_2]$,
	\begin{small}
		\[
		\begin{split}
		 &(1).\Big|\sum_{h=1}^{p_2}\mathbb{E}(\bar\bzeta_{ij}\otimes\bar\bzeta_{ih})\Big|_{\max}\le c,\quad \Big|\sum_{l=1}^{p_1}\mathbb{E}(\bar\bzeta_{ij}\otimes\bar\bzeta_{lj})\Big|_{\max}\le c,\\
		&(2). \Big|\sum_{l=1}^{p_1}\sum_{h_2=1}^{p_2}\text{Cov}(\bar\bzeta_{ij}\otimes\bar\bzeta_{ih_1},\bar\bzeta_{lj}\otimes \bar\bzeta_{lh_2})\Big|_{\max}\le c, \Big|\sum_{l_2=1}^{p_1}\sum_{h=1}^{p_2}\text{Cov}(\bar\bzeta_{ij}\otimes\bar\bzeta_{l_1j},\bar\bzeta_{ih}\otimes \bar\bzeta_{l_2h})\Big|_{\max}\le c,
		\end{split}
		\]
where $\bar\bzeta_{ij}=\text{Vec}(\sum_{t=1}^T\Fb_te_{t,ij}/\sqrt{T})$.
	\end{small}
\end{enumerate}

Assumption E.1 is summarized from the Assumptions F and G.2 in \cite{Chen2020StatisticalIF}. Indeed, we can view $\bv^\top\Eb_t\bw$ as a random variable with  mean zero and bounded variance since the noise is weakly correlated across row and column. Therefore, Assumption E.1 simply implies that $\mathbb{E}(\Fb_t\bv^\top\Eb_t\bw)\approx {\bf 0}$ and the temporal correlations of the series $\{\Fb_t\bv^\top\Eb_t\bw\}$ are also weak.  Assumption E.2 controls  higher-order correlations between the factor and noise series, where $\bar \bzeta_{ij}$ can be simply viewed as random vectors with fixed dimension and bounded marginal variances (under Assumption E.1). Assumption E is satisfied if the noise series is independent across time and independent of the factor series, given the Assumptions A to D.

\subsection{Asymptotics on  projection estimators}
We first present the following conditions on the convergence rates of the initial estimators $\hat{\Rb}$ and $\hat{\Cb}$ to guarantee the projection procedure works.

\vspace{1em}
\noindent {\bf (Sufficient Condition)} There exist $k_1\times k_1$ matrices $\hat\Hb_1$  satisfying $\hat\Hb_1\hat\Hb_1^\top\overset{p}{\rightarrow}\Ib_{k_1}$ and
\begin{equation}\label{c1}
\begin{aligned}
&\text{(a). }\frac{1}{p_1}\|\hat\Rb-\Rb\hat\Hb_1\|_F^2=O_p(w_1),\quad\text{(b). }\frac{1}{p_2}\bigg\|\frac{1}{Tp_1}\sum_{s=1}^{T}\Eb_s^\top(\hat\Rb-\Rb\hat\Hb_1)\Fb_s\bigg\|_F^2=O_p(w_2),
\end{aligned}
\end{equation}
where $w_1,w_2\rightarrow 0$ as $T,p_1$ and $p_2$ go to infinity simultaneously. There exist $k_2\times k_2$ matrices $\hat\Hb_2$ satisfying $\hat\Hb_2\hat\Hb_2^\top\overset{p}{\rightarrow}\Ib_{k_2}$ and
\begin{equation}\label{c2}
\begin{aligned}
&\text{(a). }\frac{1}{p_2}\|\hat\Cb-\Cb\hat\Hb_2\|_F^2=O_p(m_1),\quad\text{(b). }\frac{1}{p_1}\bigg\|\frac{1}{Tp_2}\sum_{s=1}^{T}\Eb_s(\hat\Cb-\Cb\hat\Hb_2)\Fb_s^\top\bigg\|_F^2=O_p(m_2),
\end{aligned}
\end{equation}
where $m_1,m_2\rightarrow 0$ as $T,p_1$ and $p_2$ go to infinity simultaneously.

Now, we state a theorem on the convergence rates of our projection estimators.

\begin{theorem}[Consistency of the projected estimators]\label{thm1}
{Under Assumptions A to E and sufficient conditions (\ref{c1}) and (\ref{c2}),  matrices $\tilde\Hb_1$ and $\tilde\Hb_2$ exist, satisfying $\tilde\Hb_1^\top\tilde\Hb_1/p_1\stackrel{p}{\rightarrow}\Ib_{k_1}$ and $\tilde\Hb_2^\top\tilde\Hb_2/p_2\stackrel{p}{\rightarrow}\Ib_{k_2}$, such that
	\[
\begin{aligned}
 &\frac{1}{p_1}\|\tilde\Rb-\Rb\tilde\Hb_1\|_F^2=O_p(\tilde w_1),\quad \|\tilde\bR_{i\cdot}-\bR_{i\cdot}\tilde\Hb_1\|^2=O_p(\tilde w_1),\quad i=1,\ldots,p_1,\\
 & \frac{1}{p_2}\|\tilde\Cb-\Cb\tilde\Hb_2\|_F^2=O_p(\tilde m_1),\quad \|\tilde\bC_{j\cdot}-\bC_{j\cdot}\tilde\Hb_2\|^2=O_p(\tilde m_1), \quad j=1,\ldots,p_2,
\end{aligned}
	\]
	as $T,p_1$ and $p_2$ go to infinity simultaneously, where $\bR_{i\cdot}$, $\tilde\bR_{i\cdot}$, $\bC_{j\cdot}$ and $\tilde\bC_{j\cdot}$ are the $i$-th /$j$-th row of $\Rb$, $\tilde\Rb$, $\Cb$ and $\tilde\Cb$, respectively, and
	\[
	\begin{aligned}
	\tilde w_1=&\frac{1}{Tp_2}+\frac{1}{p_1^2p_2^2}+m_1^2 \bigg(\frac{1}{p_1^2}+\frac{1}{Tp_1}\bigg)+m_2,\quad
	\tilde m_1= \frac{1}{Tp_1}+\frac{1}{p_1^2p_2^2}+w_1^2 \bigg(\frac{1}{p_2^2}+\frac{1}{Tp_2}\bigg)+w_2.
	\end{aligned}
	\]}
\end{theorem}

The estimation error bounds for $\tilde{\Rb}$ and $\tilde{\Cb}$ show clear dependence on the accuracy of the initial estimates. Actually, Theorem \ref{thm3} verifies that the proposed initial estimates in section \ref{sec2.3} satisfy the sufficient conditions (\ref{c1}) and (\ref{c2}) with
	\begin{equation}\label{equ3.3}
w_1=\frac{1}{p_1^2}+\frac{1}{Tp_2},\quad w_2=\frac{1}{Tp_1^2}+\frac{1}{T^2p_2^2},\quad 	 m_1=\frac{1}{p_2^2}+\frac{1}{Tp_1},\quad m_2=\frac{1}{Tp_2^2}+\frac{1}{T^2p_1^2}.
\end{equation}

A corollary follows directly.
 \begin{corollary}\label{cor1}
 {	Under Assumptions A to E, and based on assumed conditions (\ref{c1})--(\ref{equ3.3}), it holds that in Theorem \ref{thm1},
 	\[
 		\tilde w_1=\frac{1}{Tp_2}+\frac{1}{p_1^2p_2^2}+\frac{1}{T^2p_1^2},\quad
 	\tilde m_1=\frac{1}{Tp_1}+\frac{1}{p_1^2p_2^2}+\frac{1}{T^2p_2^2}.
 	\]}
 	\end{corollary}

The convergence rates for the estimators of $\Rb$ and $\Cb$ in \cite{Chen2020StatisticalIF} are $O_p\{(Tp_2)^{-1}+p_1^{-1}\}$ and $O_p\{(Tp_1)^{-1}+p_2^{-1}\}$, respectively. Corollary \ref{cor1} demonstrates that our projected estimators of the row and column factor spaces perform no worse than \cite{Chen2020StatisticalIF}'s estimator, and achieve faster convergence rates than \cite{Chen2020StatisticalIF}'s estimators, when $Tp_2>p_1$ for estimating the row factor loading matrix $\Rb$ and $Tp_1>p_2$ for estimating the column factor loading matrix $\Cb$.

If we take each column (or row) as individual observation and look at $\{\Xb_t\}$ as a vector time series of length $Tp_2$ and dimension $p_1$ as in (\ref{equ2.2}), the theorems in \cite{bai2003inferential} and \cite{fan2013large} indicate a convergence rate of $(Tp_2)^{-1}+ p_1^{-2}$ for estimating $\Rb$ (or $(Tp_1)^{-1}+p_2^{-2}$ for estimating $\Cb$). Indeed, under our assumptions, we can improve the results in \cite{Chen2020StatisticalIF} and show that the $\alpha$-PCA estimator with $\alpha=0$, i.e. the initial estimator, achieves the rates conceivable from \cite{bai2003inferential} and \cite{fan2013large}. Recall that, as (\ref{equ3.1.1})--(\ref{M1p}) show, our first-step projected matrix series $\{\Yb_t\}$ can be interpreted as a series of $p_1$ dimensional vectors of length $Tk_2$ with asymptotically negligible error entries. A comparison of (\ref{equ3.1.1})--(\ref{M1p}) with (\ref{equ2.2})--(\ref{hatM1}) demonstrates that the projection estimators benefit from a smaller noise level in the sense of the spectral norm. Indeed, the proof in the supplementary material shows that the idiosyncratic risk components for $\Mb_1$ and $\hat\Mb_1$ are of orders ${1}/{\sqrt{Tp_2}}+{1}/{(Tp_1)}+{1}/{(p_1p_2)}$ and ${1}/{\sqrt{Tp_2}}+{1}/{p_1}$, respectively. The convergence rates in Corollary \ref{cor1} also {imply} that our projection estimators converge faster than  the PCA estimators by vectorizing the columns of $\Xb_t$.

To further study the entry-wise asymptotic distributions of the estimated loadings, we need the following assumptions.

\noindent{\bf Assumption F} { \it For $i\le p_1$,
 	\[
 	 \frac{1}{\sqrt{Tp_2}}\sum_{t=1}^{T}\Fb_t\Cb^
 	 \top\be_{t,i\cdot}\stackrel{d}{\rightarrow}\mathcal{N}({\bf 0},\Vb_{1i}),\quad\text{where}\quad \Vb_{1i}={\lim_{T,p_1,p_2\rightarrow\infty}}\frac{1}{Tp_2}\sum_{t=1}^{T}\mathbb{E}\Fb_t\Cb^\top\text{cov}(\be_{t,i\cdot})\Cb\Fb_t^\top.
 	\]
 	For $j\le p_2$,
 	 	\[
 	 \frac{1}{\sqrt{Tp_1}}\sum_{t=1}^{T}\Fb_t^\top\Rb^\top\be_{t,\cdot j}\stackrel{d}{\rightarrow}\mathcal{N}({\bf 0},\Vb_{2j}),\quad\text{where}\quad \Vb_{2j}={\lim_{T,p_1,p_2\rightarrow\infty}}\frac{1}{Tp_1}\sum_{t=1}^{T}\mathbb{E}\Fb_t^\top\Rb^\top\text{cov}(\be_{t,\cdot j})\Rb\Fb_t.
 	\]
 	$\Vb_{1i}$ and $\Vb_{2j}$ are positive definite matrices whose eigenvalues are bounded away from 0 and infinity.}

Assumption F can be verified by martingale central limit theorem. It is easily fulfilled under the proposed $\alpha$-mixing condition and weak correlation assumptions.  One can refer to Chapter 16 of \cite{athreya2006measure} for more details. Similar assumptions are imposed for vector or matrix factor models, as in the work of  \cite{bai2003inferential} and \cite{Chen2020StatisticalIF}. The following Theorem \ref{thm2} shows the asymptotic distributions of the projected estimators of the loading matrices.

\begin{theorem}[Asymptotic normality of  projection estimators]\label{thm2}
	Under Assumptions A to F, if the initial estimators $\hat \Rb$ and $\hat\Cb$ are proposed as in Section \ref{sec2.3},
\begin{enumerate}
	\item for $i\le p_1$,
	\[
	\left\{\begin{aligned}
	 &\sqrt{Tp_2}(\tilde\bR_i-\tilde\Hb_1^\top\bR_i)\stackrel{d}{\rightarrow}\mathcal{N}({\bf 0}, \bLambda_1^{-1}\bGamma_1^\top\Vb_{1i}\bGamma_1\bLambda_1^{-1}),&\text{if}\quad &Tp_2=o(\min\{T^2p_1^2,p_2^2p_1^2\}),\\
	 &\tilde\bR_i-\tilde\Hb_1^\top\bR_i=O_p\big(\frac{1}{Tp_1}+\frac{1}{p_2p_1}\big),&\text{if}\quad &Tp_2\gtrsim \min\{T^2p_1^2,p_2^2p_1^2\};
	\end{aligned}
	\right.
	\]
	\item for $j\le p_2$,
	\[
	\left\{\begin{aligned}
	 &\sqrt{Tp_1}(\tilde\bC_j-\tilde\Hb_2^\top\bC_j)\stackrel{d}{\rightarrow}\mathcal{N}({\bf 0}, \bLambda_2^{-1}\bGamma_2^\top\Vb_{2j}\bGamma_2\bLambda_2^{-1}),&\text{if}\quad &Tp_1=o(\min\{T^2p_2^2,p_1^2p_2^2\}),\\
	 &\tilde\bC_j-\tilde\Hb_2^\top\bC_j=O_p\big(\frac{1}{Tp_2}+\frac{1}{p_1p_2}\big),&\text{if}\quad &Tp_1\gtrsim \min\{T^2p_2^2,p_1^2p_2^2\}.
	\end{aligned}
	\right.
	\]
\end{enumerate}
\end{theorem}

\subsection{Theorems on initial estimators}\label{sec3.2}
As claimed, the initial estimators are $\alpha$-PCA solutions in \cite{Chen2020StatisticalIF} with $\alpha=0$. However, based on the argument below Corollary \ref{cor1}, the convergence rate in their paper is slower than the expected one  by a factor of $p_1^{-1}$ or $p_2^{-1}$. Under our assumptions, an improved rate is accessible and summarized in the following theorem.

\begin{theorem}\label{thm3}
	Under Assumptions A to E, (\ref{equ3.3}) holds for the initial estimators.
\end{theorem}

For the initial estimator $\hat\Rb$, the rate $w_1$  matches the typical rate $O_p(T^{-1}+p_1^{-2})$ of the vector factor model when $p_2=1$, as shown by Theorem 2 in \cite{bai2003inferential}.  The initial estimators are also asymptotically normally distributed as shown in the next theorem.

\begin{theorem}[Asymptotic normality of the initial estimators]\label{thm4}
Under Assumptions A to F, as $T,p_1,p_2\rightarrow\infty$,
	\begin{enumerate}
		\item for $i\le p_1$,
		\[
		\left\{\begin{aligned}
		 &\sqrt{Tp_2}(\hat\bR_i-\hat\Hb_1^\top\bR_i)\stackrel{d}{\rightarrow}\mathcal{N}({\bf 0}, \bLambda_1^{-1}\bGamma_1^\top\Vb_{1i}\bGamma_1\bLambda_1^{-1}),&\text{if}\quad &Tp_2=o(p_1^2),\\
		 &\hat\bR_i-\hat\Hb_1^\top\bR_i=O_p(p_1^{-1}),&\text{if}\quad &Tp_2\gtrsim p_1^2;
		\end{aligned}
		\right.
\]
		\item for $j\le p_2$,
\[
\left\{\begin{aligned}
&\sqrt{Tp_1}(\hat\bC_j-\hat\Hb_2^\top\bC_j)\stackrel{d}{\rightarrow}\mathcal{N}({\bf 0}, \bLambda_2^{-1}\bGamma_2^\top\Vb_{2j}\bGamma_2\bLambda_2^{-1}),&\text{if}\quad &Tp_1=o(p_2^2),\\
&\hat\bC_j-\hat\Hb_2^\top\bC_j=O_p(p_2^{-1}),&\text{if}\quad &Tp_1\gtrsim p_2^2,
\end{aligned}
\right.
\]
where $\hat\bR_i$ and $\hat\bC_j$ are the $i$-th and $j$-th row vectors of $\hat\Rb$ and $\hat\Cb$, respectively.
	\end{enumerate}
\end{theorem}

Compared with Theorem \ref{thm2}, $\hat\Rb$ and $\tilde\Rb$ share the same asymptotic covariance matrix when $p_1$ is sufficiently large. However,  the normality of $\hat\Rb$ requires a more stringent condition that $p_1^2\gg Tp_2$, while for the projected estimator we only require $p_1^2\gg \max\{p_2/T, T/p_2\}$. A similar conclusion holds for $\hat\Cb$ and $\tilde\Cb$.

\subsection{Estimating  factor matrix and common components}
As long as the loading matrices are determined, the factor matrix $\Fb_t$ can be estimated easily by
\[
\tilde\Fb_t=\frac{1}{p_1p_2}\tilde\Rb^\top\Xb_t\tilde\Cb.
\]
The common component matrix is then given by
\[
\tilde{\bf S}_t=\tilde\Rb\tilde\Fb_t\tilde\Cb^\top.
\]
The next theorem provides the consistency of the estimated factors and common components.
\begin{theorem}\label{thm5}
	Under Assumptions A to E, as $\min\{T,p_1,p_2\}\rightarrow\infty$, for any $t\in [T],i\in[p_1]$ and $j\in [p_2]$,
	\[
	\begin{split}
	 (1).&\|\tilde\Fb_t-\tilde\Hb_1^{-1}\Fb_t\tilde\Hb_2^{-1}\|\le O_p\bigg(\frac{1}{\sqrt{T}\times \min\{p_1,p_2\}}+\frac{1}{\sqrt{p_1p_2}}\bigg),\\
	(2).&| \mathrm{\tilde S}_{t,ij}-\mathrm{S}_{t,ij}|=O_p\bigg(\frac{1}{\sqrt{Tp_1}}+\frac{1}{\sqrt{Tp_2}}+\frac{1}{\sqrt{p_1p_2}}\bigg).
	\end{split}
	\]
\end{theorem}
\begin{remark}
	The convergence rates in Theorem \ref{thm5} are  the same as those in \cite{Chen2020StatisticalIF} when $p_1\asymp p_2$, although the estimated loadings by the projection method are generally  more accurate. This reason is that the estimation error of $\tilde\Fb_t$ mainly comes from the error term $(p_1p_2)^{-1}\Rb^\top\Eb_t\Cb$. Even if the loadings $\Rb$ and $\Cb$ are known, the best convergence rate for estimating $\Fb_t$ is still of the rate $(p_1p_2)^{-1/2}$ under the spectral norm. This error further affects the estimation of the common components.
One can easily verify the asymptotic normality of $\tilde\Fb_t$ by imposing certain conditions on $(p_1p_2)^{-1}\Rb^\top\Eb_t\Cb$.

\end{remark}

\subsection{Determining the pair of row and column factor numbers $k_1,k_2$}
The dimensions $k_1$ and $k_2$ of the common factor matrix need to be determined before  the  procedures can be applied. In this study, we specify the numbers of row and column  factors by borrowing the eigenvalue-ratio statistics discussed in \cite{lam2012factor} and \cite{ahn2013eigenvalue}. In detail,  $\hat\Rb$ and $\hat\Cb$ are selected as the initial projection matrices, and then $k_1$ is estimated by
\begin{equation}\label{equ2}
\hat k_1=\arg\max_{j\le k_{\max}}\frac{\lambda_j(\tilde\Mb_1)}{\lambda_{j+1}(\tilde\Mb_1)},
\end{equation}
where $k_{\max}$ is a predetermined upper bound for $k_1$. \cite{Chen2020StatisticalIF} proposed a similar criterion using $\hat\Mb_1$. We use $\tilde\Mb_1$ rather than $\hat\Mb_1$ because $\tilde\Mb_1$ is usually more accurate  for approximating the column covariance matrix of the common components.

When the common factors are sufficiently strong, the leading $k_1$ eigenvalues of $\tilde\Mb_1$ are well separated from the others. Thus, the eigenvalue ratios in equation (\ref{equ2}) are asymptotically maximized exactly at $j=k_1$. To avoid vanishing denominators, we can add an asymptotically negligible term, such as $c\delta$ for some small constant $c$ and $\delta= \max\{1/\sqrt{Tp_2},1/\sqrt{Tp_1},p_1^{-1}\}$, to the denominator of equation (\ref{equ2}).  One problem to calculate $\tilde\Mb_1$ is that $\hat\Cb$ must be predetermined, which means $k_2$ must be given first. Empirically, $k_1$ and $k_2$ are both unknown. To address this difficulty, we suggest using the following iterative Algorithm \ref{alg2} to determine the paired numbers of factors.

\begin{algorithm}[H]
	\caption{Iterative algorithm to specify numbers of factors}\label{alg2}
	{\bf Input:} Data matrices $\{\Xb_t\}_{t\le T}$, maximum number $k_{\max}$, maximum iterative step $m$\\
	{\bf Output:} Numbers of row and column factors $\hat k_1$ and $\hat k_2$
	\begin{algorithmic}[1]
		\State  initialization: $\hat k_1^{(0)}=k_{\max},\hat k_2^{(0)}=k_{\max}$;
		
		\State for $t=1,\ldots,m$, given $\hat k_2^{(t-1)}$, estimate $\hat\Cb^{(t)}$ by the initial estimator, and calculate $\tilde\Mb_1^{(t)}$ using $\hat\Cb^{(t)}$,  then $\hat k_1^{(t)}$ is given by equation (\ref{equ2});
		
		\State given  $\hat k_1^{(t)}$, estimate $\hat\Rb^{(t)}$ by the initial estimator, and calculate $\tilde\Mb_2^{(t)}$ using $\hat\Rb^{(t)}$,  then $\hat k_2^{(t)}$ is given by a parallel ``ER" approach by replacing $\tilde\Mb_1$ with $\tilde \Mb_2^{(t)}$ in equation (\ref{equ2});
		
		\State repeat Steps 2 and 3 until $\hat k_1^{(t)}=\hat k_1^{(t-1)}$ and $\hat k_2^{(t)}=\hat k_2^{(t-1)}$, or reach the maximum iterative step.
	\end{algorithmic}
\end{algorithm}
\begin{remark}
	The term $c\delta$ only works as a lower bound of the denominator in our technical proofs. As one reviewer pointed out, adding such a term may affect the finite sample performance. We  compared the empirical performances of the iterative algorithm with $c=0$, $c=10^{-4}$ and $c=1$ in the simulation study. The numerical results are not much sensitive to the
term $c\delta$. In practice, we suggest setting $c$ sufficiently small in case of underestimation.
	\end{remark}

The consistency of the iterative algorithm is guaranteed by the following theorem.
\begin{theorem}[Specifying the numbers of row and column factors]\label{thm6}
	Under Assumptions A to E,  when $\min\{k_1,k_2\}>0$, $\min\{T,p_1,p_2\}\rightarrow\infty$ and $ k_{\max}$ is a predetermined constant no smaller than $\max\{k_1,k_2\}$, if $\hat k_2^{(t-1)}\in[k_2,k_{\max}]$ for some $t$ in the iterative algorithm \ref{alg2},
	\[
	\text{Pr}(\hat k_1^{(t)}=k_1)\rightarrow 1;
	\]
	and if $\hat k_1^{(t)}\in[k_1,k_{\max}]$ for some $t$ in  Algorithm \ref{alg2}, then
	\[
	\text{Pr}(\hat k_2^{(t)}=k_2)\rightarrow 1.
	\]
\end{theorem}
Theorem \ref{thm6} indicates that as long as we start with some $k_1^{(0)}$ and $k_2^{(0)}$ larger than the true $k_1$ and $k_2$, the iterative algorithm can consistently estimate the numbers of factors. The algorithm is computationally very fast because it has a large probability to stop within finite steps.

The finite sample performance of the eigenvalue-ratio method usually depends on the maximized ratio at the true number of factors. A larger ratio implies a better separation of the spiked eigenvalues, which leads to better estimation of the number of factors. The maximized ratio of the above algorithm is shown to be $\min\{\sqrt{Tp_2},\sqrt{Tp_1},p_1\}$ for $k_1$ and $\min\{\sqrt{Tp_1},\sqrt{Tp_2},p_2\}$ for $k_2$ in the proof. However, if we vectorize the data matrices and apply the eigenvalue-ratio approach in \cite{ahn2013eigenvalue}, the maximized ratio for the total number of factors will be of the rate $\min\{\sqrt{T}, p_1p_2\}$. Therefore, the eigenvalue-ratio method for the vectorized model may perform better when $T$ is large but $p_1$ or $p_2$ is small, which is conceivable as we have a pair of factor numbers to be estimated. The estimation of $k_2$ brings new errors to $\hat k_1$ in the iterative algorithm. Actually, we show in the proof that when $k_2$ is given, the maximized eigenvalue-ratio for estimating $k_1$ is of the order $\min\{\sqrt{Tp_2}, Tp_1, p_1p_2\}$, which is even better than that of the vectorized model. As $\hat k_1^{(t)}$ or $\hat k_2^{(t)}$ has large probability to be exactly $k_1$ or $k_2$ after a few iterations,  the iterative algorithm performs impressively well empirically.

\section{Simulation studies}\label{sec4}
\subsection{Data generation}

In this section, we investigate the finite sample performances of the proposed projection procedure. The observed data matrices are generated according to model (\ref{mod1.1}). In detail, we set $k_1=k_2=3$, draw the entries of $\Rb$ and $\Cb$ independently from uniform distribution $\mathcal{U}(-1,1)$, and let
\begin{equation}\label{equ4.1}
\begin{split}
\text{Vec}(\Fb_t)=&\phi \times \text{Vec}(\Fb_{t-1})+\sqrt{1-\phi^2}\times\bepsilon_t,\quad \bepsilon_t\overset{i.i.d.}{\sim}\mathcal{N}({\bf 0},\Ib_{k_1\times k_2}),\\
\text{Vec}(\Eb_t)=&\psi \times \text{Vec}(\Eb_{t-1})+\sqrt{1-\psi^2}\times\text{Vec}(\Ub_t),\quad \Ub_t\overset{i.i.d.}{\sim}\mathcal{MN}({\bf 0},\Ub_E,\Vb_E),
\end{split}
\end{equation}
where $\Ub_t$ is from matrix-normal distribution, i.e., $\text{Vec}(\Ub_t)\overset{i.i.d.}{\sim}\mathcal{N}({\bf 0},\Vb_E\otimes \Ub_E)$.
$\Ub_E$ and $\Vb_E$ are matrices with ones on the diagonal, while the off-diagonal are $1/p_1$ and $1/p_2$, respectively. Thus, by setting $\phi$ and $\psi$ unequal to zero, the simulated factors are temporally correlated, and the idiosyncratic noises contain both temporal and cross-sectional  correlations.  The pair  of factor numbers is assumed to be known except in subsection \ref{sec:4.5}, where we investigate the empirical performances of  Algorithm 2 to estimate the numbers of factors.  All the simulation results hereafter are based on 500 replications if not specifically mentioned.

\subsection{Verifying the convergence rates for loading spaces}
We first  compare the performances of our \textbf{P}rojected \textbf{E}stimation (PE) method with those of the $\alpha$-PCA  method by \cite{Chen2020StatisticalIF} in terms of estimating the loadings. We consider two settings, where Setting A is for estimating the row factor loading matrix $\Rb$ while Setting B is designed for estimating the column loading matrix $\Cb$.

\vspace{0.5em}

\noindent\textbf{Setting A}: $p_1=20$, $T=p_2\in\{20,50,100,150,200\}$, $\phi=\psi=0.1$.

\vspace{0.5em}

\noindent\textbf{Setting B}: $p_2=20$, $T=p_1\in\{20,50,100,150,200\}$, $\phi=\psi=0.1$.

\vspace{0.5em}

In view of  identifiability, we evaluate the performances by the distance between the estimated loading space and true loading space.  That is,
\[
\mathcal{D}(\hat\Rb,\Rb)=\bigg(1-\frac{1}{k_1}\text{tr}(\hat\Qb{\hat\Qb}^\top\Qb{\Qb}^\top)\bigg)^{1/2},
\]
where $\Qb$ and $\hat\Qb$ are the left singular-vector matrices of the true loading $\Rb$ and  its estimator $\hat\Rb$, respectively. $\mathcal{D}(\hat\Cb,\Cb)$ is defined similarly. Here, we abuse the notations with $\hat\Rb$ and $\hat\Cb$ but it shall cause no misunderstanding. The distance $\mathcal{D}(\hat\Rb,\Rb)$ is always between 0 and 1. When the corresponding matrices lie in the same space, they are equal to 0. If the two spaces are orthogonal, then  they are equal to 1. Once Assumptions A--E are satisfied, the squared distances would converge to 0 with the same rates as in Corollary \ref{cor1}. Thus, $\mathcal{D}(\hat\Rb,\Rb)$ and $\mathcal{D}(\hat\Cb,\Cb)$ are particularly suitable to quantify the estimation accuracy of the loading matrices.

Table \ref{tab1} shows the  averaged estimation errors with standard errors in parentheses under Settings A and B. We take $\alpha=-1,0,1$ for the $\alpha$-PCA as in \cite{Chen2020StatisticalIF}. All the methods benefit from large dimensions, and PE always shows lowest estimation errors and standard errors.
Figure \ref{fig2} plots the averaged log errors of the PE and $\alpha$-PCA with $\alpha=0$, which reflects the different convergence rates of the estimators by PE and $\alpha$-PCA . The left plane shows that the log error of the PE method for estimating $\Rb$ is almost linear to $\log (\sqrt{Tp_2})$ with slope $-1$, which matches the rate in Corollary \ref{cor1}. However, for the $\alpha$-PCA method, the log error first decreases with growing $\log(\sqrt{Tp_2})$ but later tends to be invariant. This result is conceivable as the convergence rate of $\hat\Rb$ by $\alpha$-PCA mainly depends on $p_1$ when  $T$ and $p_2$ are sufficiently large. A similar conclusion can be drawn for the column factor loading matrix $\Cb$ from the right panel of Figure \ref{fig2}. We conclude that the projected method leads to more accurate estimation of the loading spaces compared with $\alpha$-PCA  and the numerical results verify the convergence rates in Corollary \ref{cor1}.

\begin{table*}[hbtp]
	\begin{center}
		\small
		\addtolength{\tabcolsep}{2pt}
		\caption{Averaged estimation errors and standard errors (in parentheses) of $\mathcal{D}(\hat\Rb,\Rb)$ and $\mathcal{D}(\hat\Cb,\Cb)$ for Settings A and B (effects of $T,p_1,p_2$), over 500 replications. ``PE'':  proposed projected method. ``($a$)PCA'': $\alpha$-PCA  with $\alpha=a$.}\label{tab1}
		 \renewcommand{\arraystretch}{1.2}
		\scalebox{0.9}{ 	
				 \begin{tabular*}{18cm}{cccccccc}
				\toprule[1.2pt]
				 Evaluation&$T$&$p_1$&$p_2$&PE&(-1)PCA&(0)PCA&(1)PCA\\\midrule[1pt]
				 \multirow{5}*{$\mathcal{D}(\hat\Rb,\Rb)$}&20& \multirow{5}*{20}&20&{\bf0.0934(0.0154)}&0.1166(0.0293)&0.1138(0.0276)&0.1174(0.0279)
\\
				 &50&&50&{\bf0.0358(0.0052)}&0.0599(0.0208)&0.0595(0.0205)&0.0600(0.0203)
\\
				 &100&&100&{\bf0.0175(0.0026)}&0.0479(0.0199)&0.0478(0.0199)&0.0479(0.0199)
\\
				 &150&&150&{\bf0.0116(0.0016)}&0.0430(0.0185)&0.0430(0.0186)&0.0431(0.0186)
\\
				 &200&&200&{\bf0.0088(0.0012)}&0.0446(0.0236)&0.0445(0.0236)&0.0445(0.0236)\\\midrule[1pt]
				  \multirow{5}*{$\mathcal{D}(\hat\Cb,\Cb)$}&20&20& \multirow{5}*{20}&{\bf0.0928(0.0153)}&0.1153(0.0306)&0.1127(0.0299)&0.1162(0.0305)
\\
				 &50&50&&{\bf0.0359(0.0052)}&0.0598(0.0216)&0.0596(0.0218)&0.0603(0.0220)
\\
				 &100&100&&{\bf0.0173(0.0024)}&0.0460(0.0191)&0.0460(0.0191)&0.0461(0.0191)
\\
				 &150&150&&{\bf0.0117(0.0017)}&0.0423(0.0190)&0.0422(0.0189)&0.0422(0.0188)
\\
				 &200&200&&{\bf0.0087(0.0012)}&0.0437(0.0218)&0.0437(0.0218)&0.0437(0.0218)\\
\bottomrule[1.2pt]		
		\end{tabular*}}		
	\end{center}
\end{table*}

 \begin{figure}[hbpt]
	\begin{subfigure}{.5\textwidth}
		\centering
		 \includegraphics[width=7.5cm,height=7.5cm]{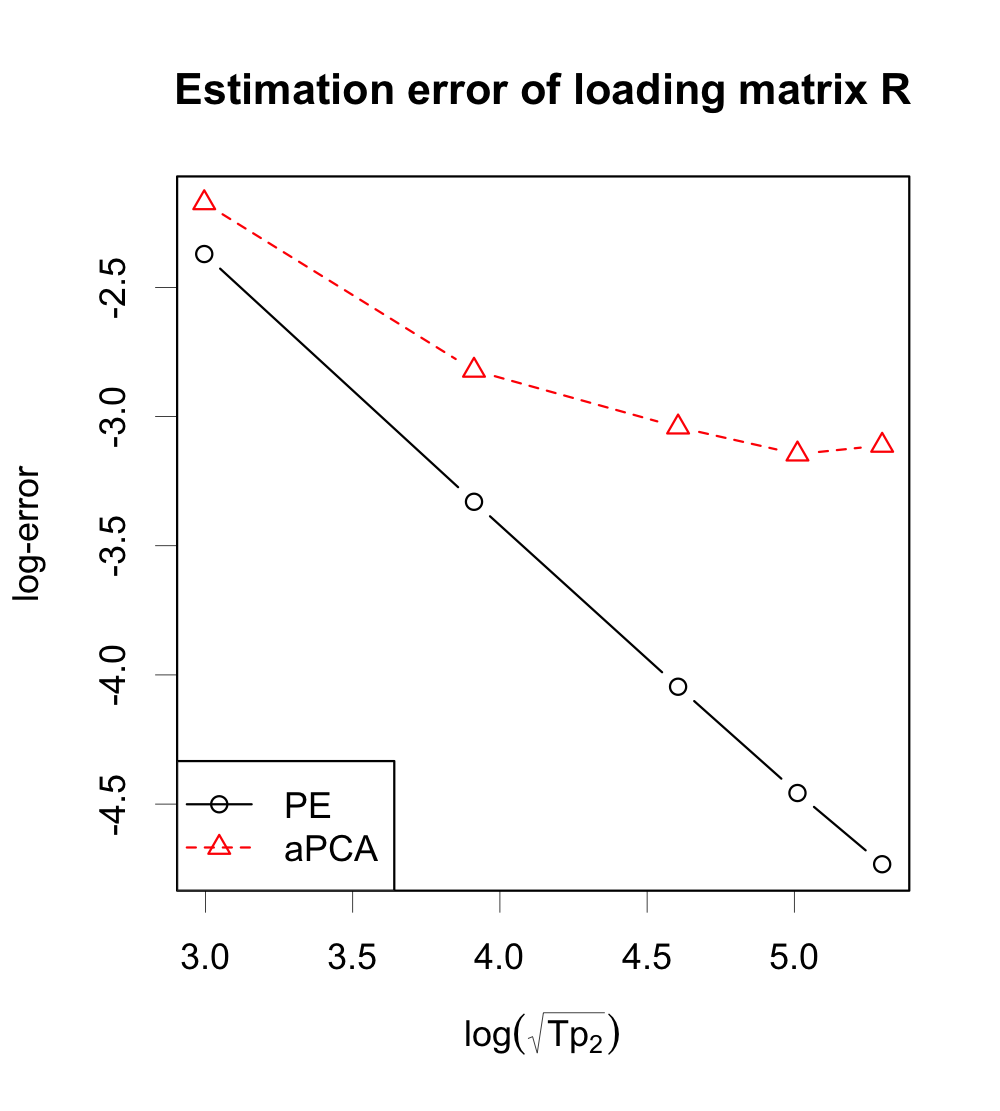}
	\end{subfigure}
	\begin{subfigure}{.5\textwidth}
		\centering
		 \includegraphics[width=7.5cm,height=7.5cm]{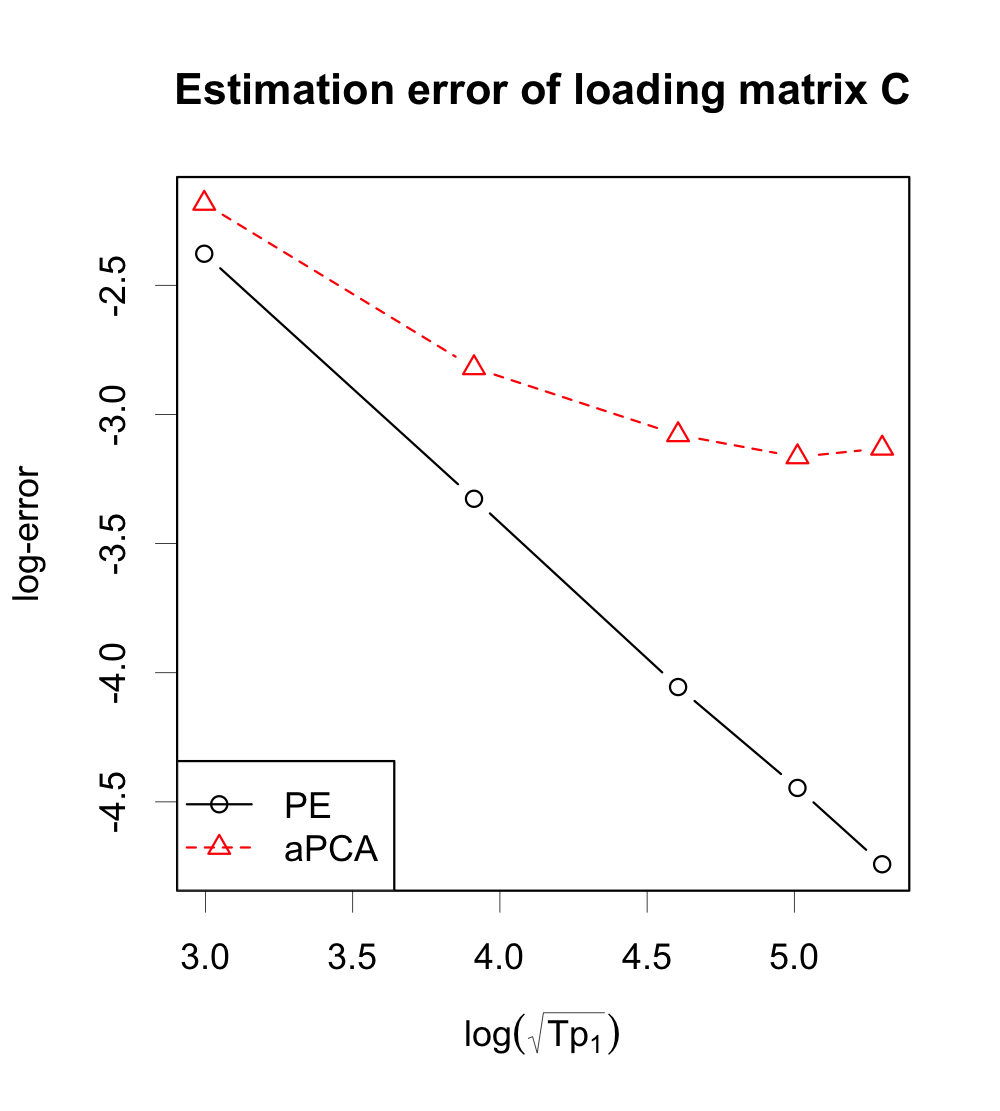}
	\end{subfigure}
	\caption{ Mean log error of estimating loading matrices, over 500 replications. Left: for $\Rb$, $p_1=20$, $T=p_2\in\{20,50,100,150,200\}$. Right: for $\Cb$, $p_2=20$, $T=p_1\in\{20,50,100,150,200\}$. ``PE'': the proposed projected method. ``aPCA'': $\alpha$-PCA  with $\alpha=0$.}\label{fig2}
\end{figure}

\subsection{Verifying the asymptotic normality}

In this section, we  check the asymptotic normality of $\tilde \Rb$ and verify the asymptotic variances in Theorem \ref{thm2} by numerical studies. For data generation, we normalize $\Rb$  as $\sqrt{p_1}$ times its left singular-vector matrix such that the identification condition $\Rb^\top\Rb/p_1=\Ib_{k_1}$ is satisfied. The column loading matrix $\Cb$ is normalized similarly.  Let
\[
\text{Vec}(\Fb_t)\overset{i.i.d.}{\sim}\mathcal{N}({\bf 0},\Db), \quad \text{with } \Db=\text{diag}(1.5,1,0.5,1.5,1,0.5,1.5,1,0.5),
\]
so that the eigenvalues of $\bSigma_1$ in Assumption B are distinct. The errors $\Eb_t$ are generated according to equation (\ref{equ4.1}) with $\psi=0$. Thus, $\Fb_t$ and $\Eb_t$ are both independent across time, which simplifies the calculation of the asymptotic covariance matrix. Actually, under the above setting, as $\min\{T,p_1,p_2\}\rightarrow \infty$, we have
\[
\bSigma_1=\left(\begin{matrix}
4.5&0&0\\
0&3&0\\
0&0&1.5
\end{matrix}\right), \quad \sqrt{Tp_2}(\tilde\bR_{i\cdot}-\bR_{i\cdot}^\top\tilde\Hb_1)\overset{d}{\rightarrow}\mathcal{N}({\bf 0},\bSigma_{\tilde\Rb}),\quad \text{where} \quad  \bSigma_{\tilde\Rb}:=\frac{\text{tr}(\Cb^\top\Vb_{E}\Cb)}{3p_2}\bSigma_1^{-1}.
\]We set $k_1=k_2=3$, $p_1=20$, and $T=p_2=200$.

Figure \ref{fig3} shows the histograms of the first coordinates of  $\sqrt{Tp_2}\bSigma_{\tilde\Rb}^{-1/2}(\tilde\bR_1-\tilde\Hb_1^\top\bR_1)$ and $\sqrt{Tp_2}\bSigma_{\tilde\Rb}^{-1/2}(\hat\bR_1-\hat\Hb_1^\top\bR_1)$ under the above setting with 500 replications. The asymptotic covariance matrices of the initial estimators and the projected estimators are the same theoretically, although the rotational matrices are not identical.  Panels (a) and (b) of Figure \ref{fig3} show that the projected estimator is almost normally distributed, but the estimator by $\alpha$-PCA deviates far from ``normal'' when $p_1=20$. This result is expected because in this case, the condition $Tp_2=o(p_1^2)$ in Theorem \ref{thm4} is not met, but the much looser condition for our projected estimators in Theorem \ref{thm2} is already satisfied. When we increase $p_1$ to 400, both estimators show ``normality" with the same covariance matrix, as demonstrated in panels (c) and (d) of Figure \ref{fig3}.

 \begin{figure}[hbpt]
	\begin{subfigure}{.24\textwidth}
		\centering
		 \includegraphics[width=3.5cm,height=3.5cm]{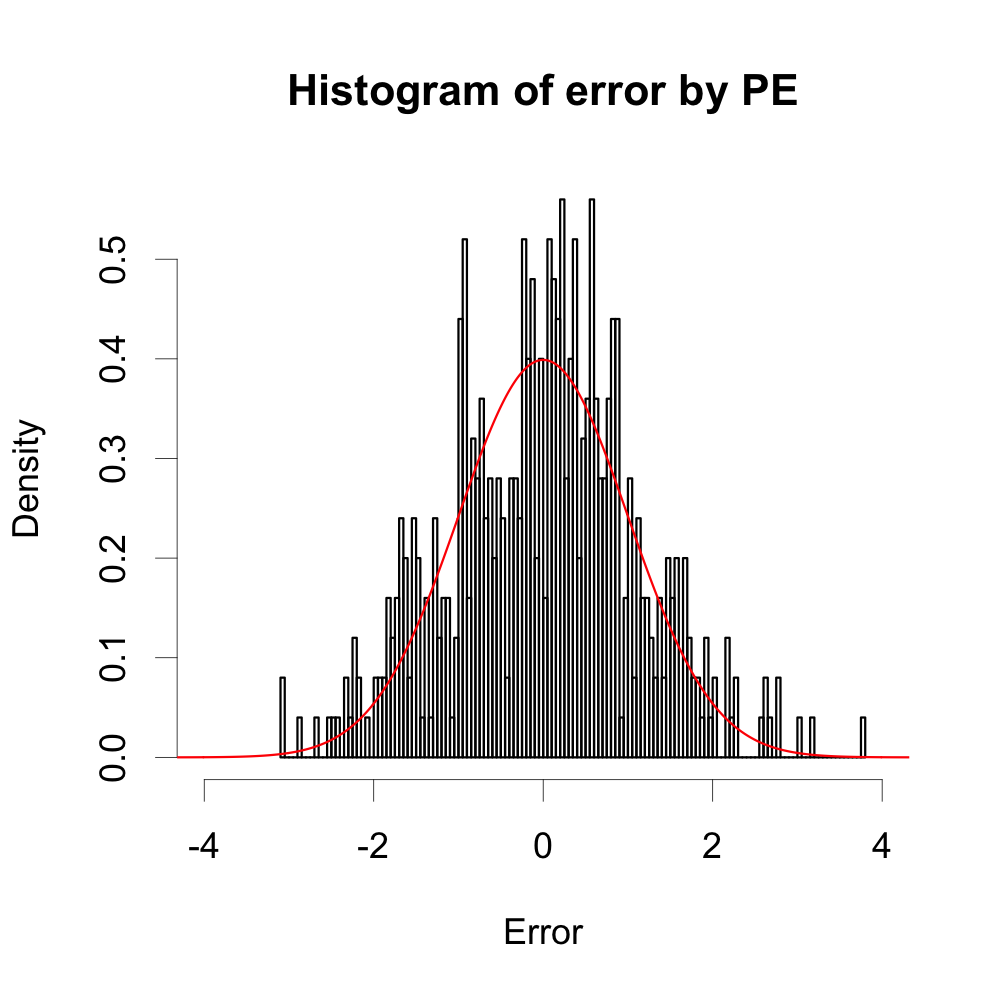}
		\subcaption{}
	\end{subfigure}
	\begin{subfigure}{.24\textwidth}
		\centering
		 \includegraphics[width=3.5cm,height=3.5cm]{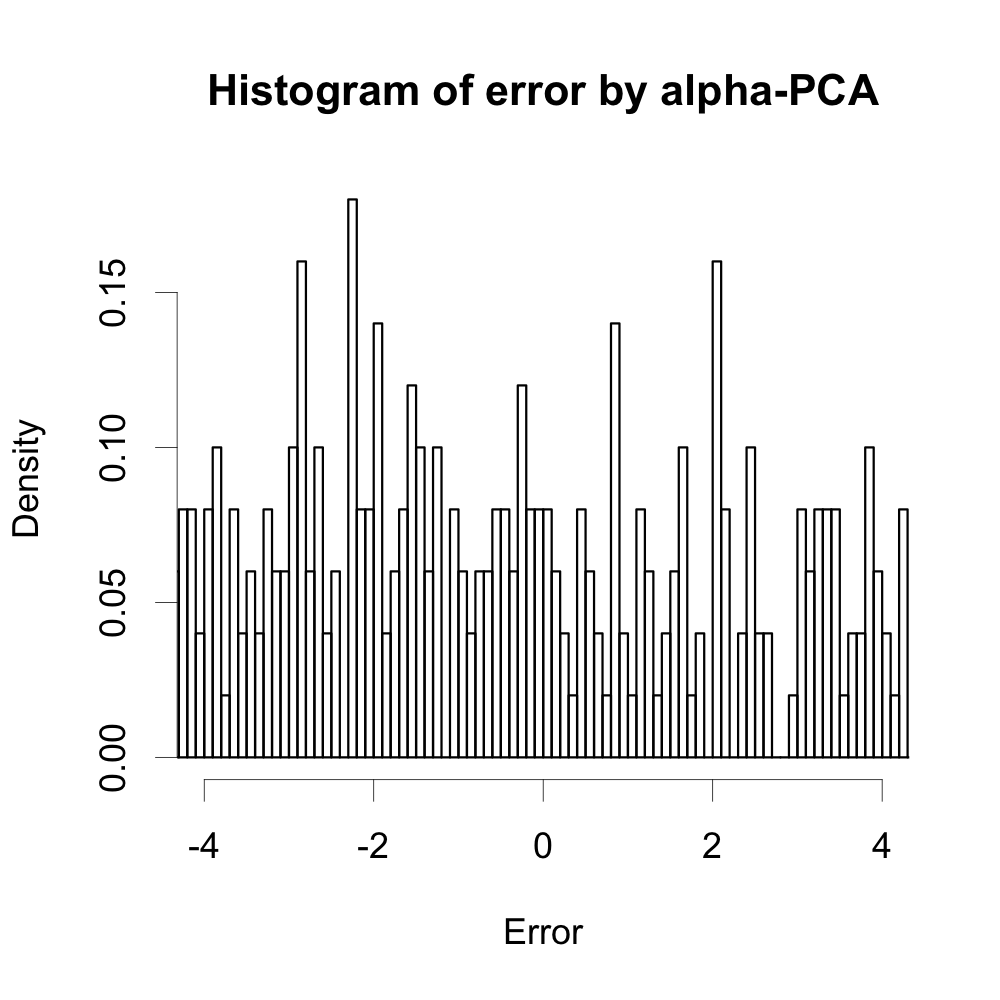}
				\subcaption{}
	\end{subfigure}
	\begin{subfigure}{.24\textwidth}
		\centering
		 \includegraphics[width=3.5cm,height=3.5cm]{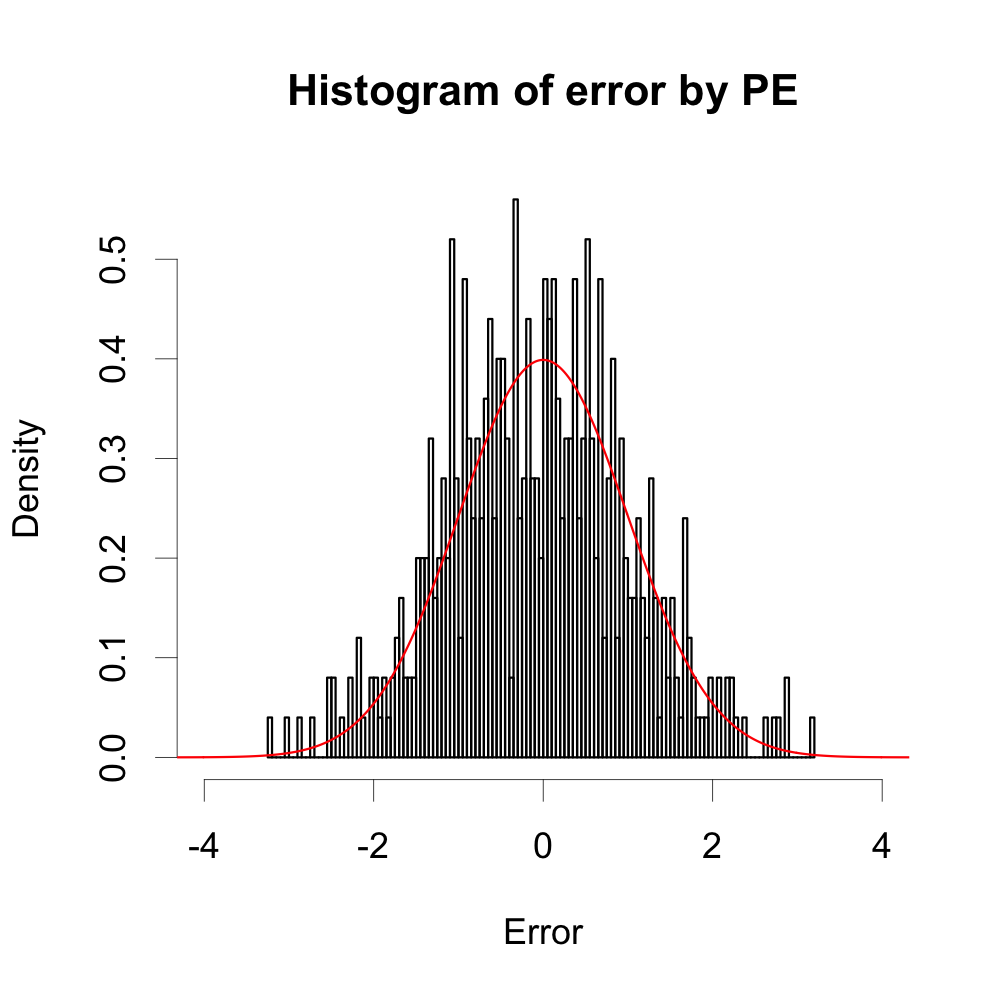}
				\subcaption{}
	\end{subfigure}
	\begin{subfigure}{.24\textwidth}
		\centering
		 \includegraphics[width=3.5cm,height=3.5cm]{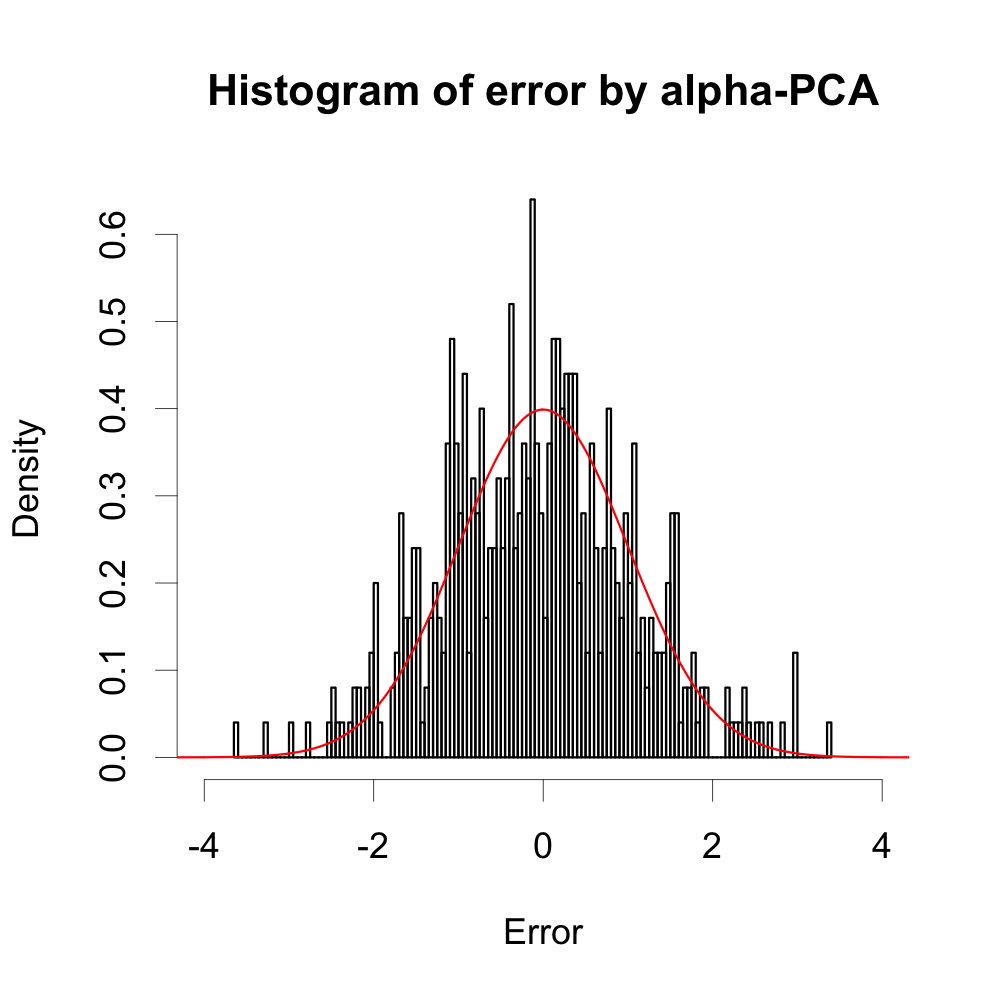}
				\subcaption{}
	\end{subfigure}
	\caption{Histograms of estimation error for $\mathrm{R}_{11}$ after normalization, over 500 replications. $T=p_2=200$. (a). PE, $p_1=20$. (b). $\alpha$-PCA with $\alpha=0$, $p_1=20$. (c). PE, $p_1=400$. (d). $\alpha$-PCA with $\alpha=0$, $p_1=400$. The red real line plots the probability density function of standard normal distribution.}\label{fig3}
\end{figure}

\subsection{Estimation error of common components}
In this subsection, we investigate the empirical performances of the PE and $\alpha$-PCA methods in terms of estimating the common components under Setting A. We evaluate the performance of different methods by the mean squared error, i.e.,
\[
\text{MSE}=\frac{1}{Tp_1p_2}\sum_{t=1}^T\|\hat{\bf S}_t-{\bf S}_t\|_F^2.
\]
We also investigate the effects when we under/overestimate the numbers of factors in this subsection.

Figure \ref{fig4} shows the boxplots of  the MSEs by PE and $\alpha$-PCA with $\alpha=0$ under Setting A over 500 replications. We do not show the results for $\alpha=\pm1$ because the choice of $\alpha$ has a minimal effect on the results under this setting. The left panel of  Figure \ref{fig4} shows the boxplots of MSEs for various $T$ by the PE and $\alpha$-PCA estimates with the true numbers of factors. Both methods perform better as the dimension $T=p_2$ grows, and the PE always leads to slightly lower MSEs and smaller deviations.  The right panel corresponds to the case when we overestimate/underestimate the numbers of factors with $T=p_2=200$. The effect of overestimation is negligible while  underestimating the numbers of factors would result in intolerable MSEs, which is consistent with the loadings for  vector factor models. Detailed numerical results are reported in Table \ref{tab2}.

\begin{figure}[hbpt]
	\begin{subfigure}{.5\textwidth}
		\centering
		 \includegraphics[width=7.5cm,height=7.5cm]{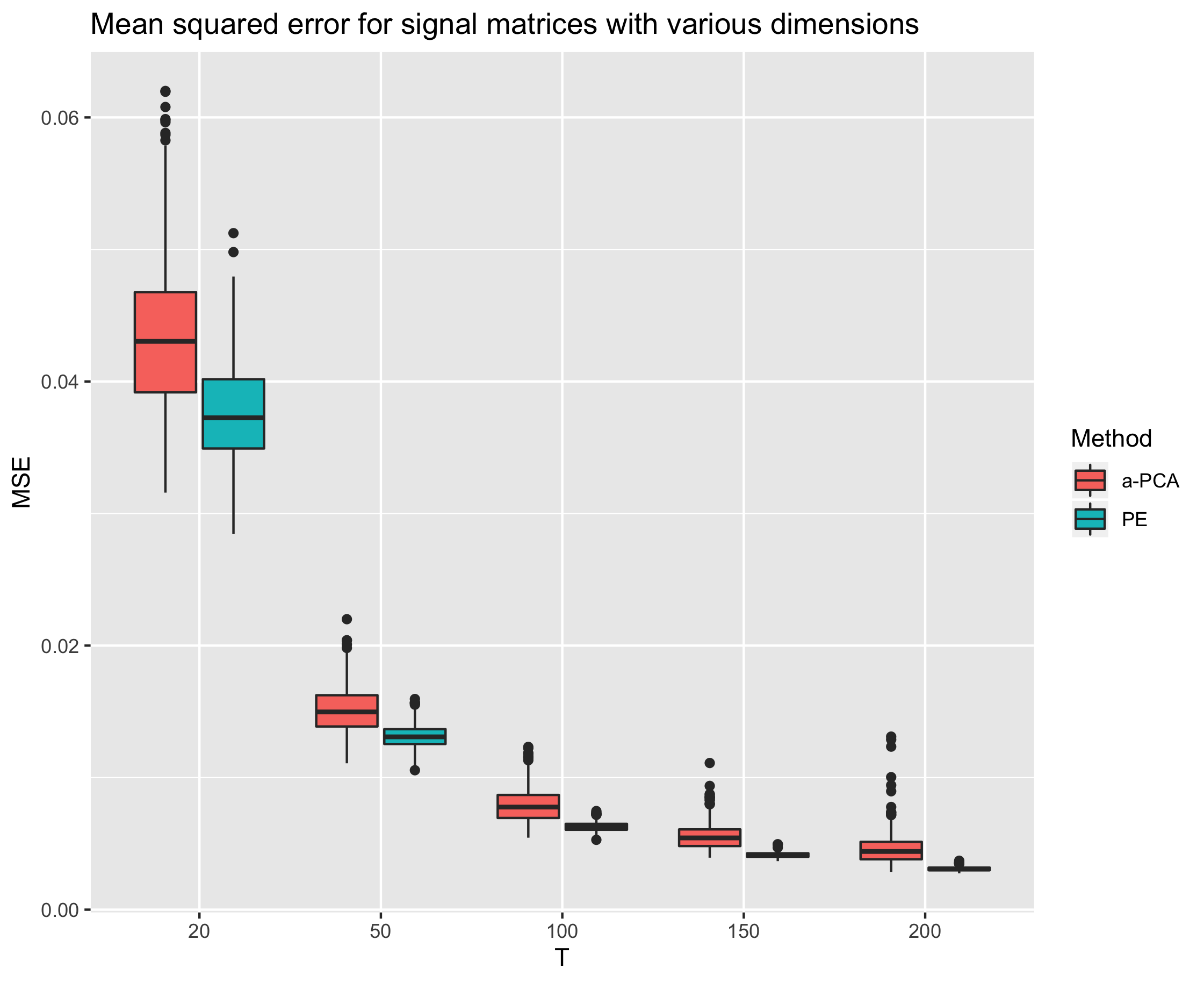}
	\end{subfigure}
	\begin{subfigure}{.5\textwidth}
		\centering
		 \includegraphics[width=7.5cm,height=7.5cm]{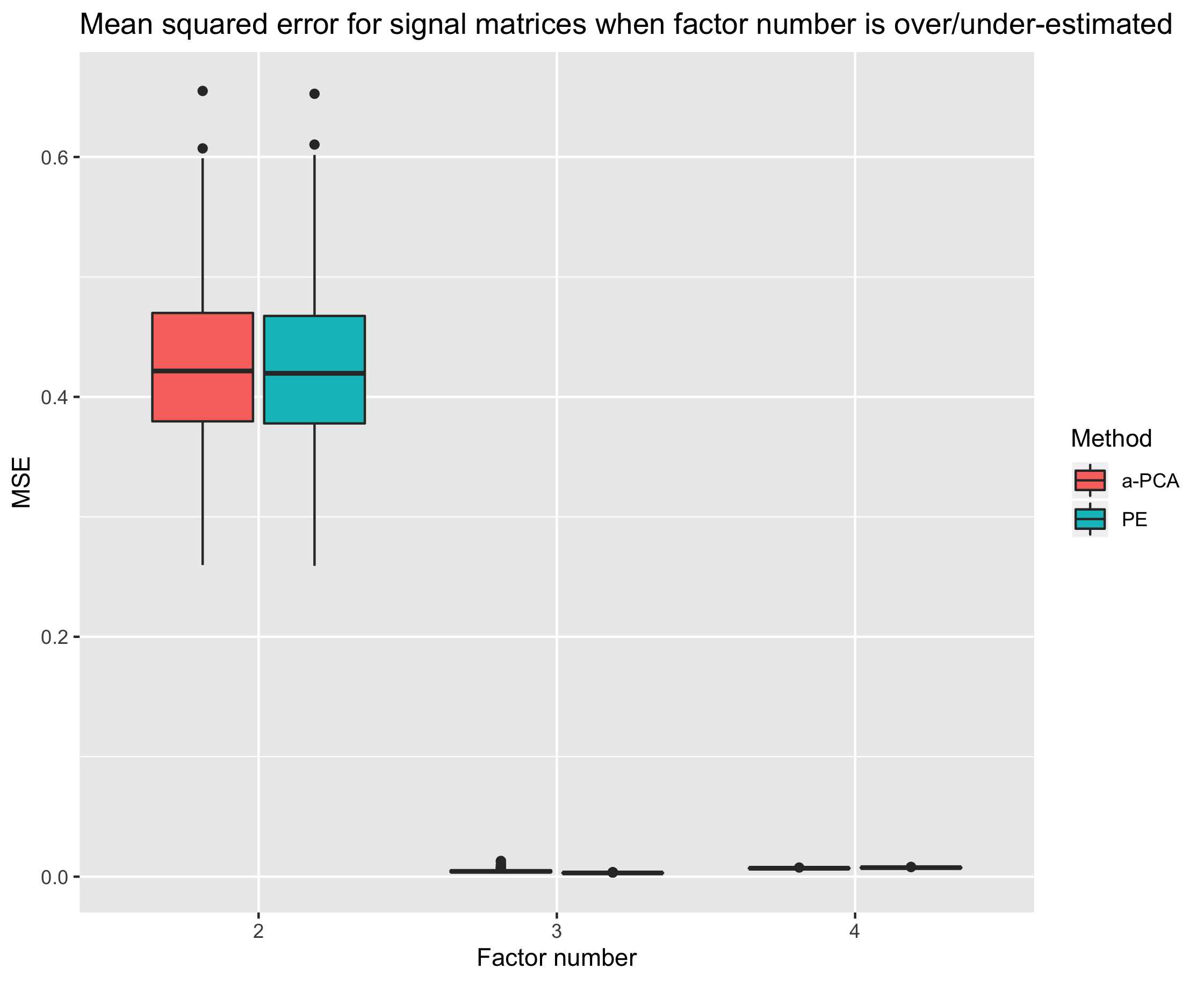}
	\end{subfigure}
	\caption{Boxplots of  MSEs of common components under Setting A over 500 replications. Left: The true numbers of factors are given while $T=p_2$ grows. Right: The effects if we use less ($\hat k_1=\hat k_2=2$) or more ($\hat k_1=\hat k_2=4$) factors in the estimation, $T=p_2=200$. ``PE'': proposed projected method. ``aPCA'': $\alpha$-PCA  with $\alpha=0$.}\label{fig4}
\end{figure}

\begin{table*}[hbtp]
	\begin{center}
		\small
		\addtolength{\tabcolsep}{0pt}
		\caption{Squared errors and standard errors (in parentheses) of MSEs for common components (effects of $T,p_1,p_2$ and under/overestimating the numbers of factors) over 500 replications. ``PE'': the projected method. ``(0)PCA'': $\alpha$-PCA  with $\alpha=0$. The true numbers of factors are $k_1=k_2=3$.  ``$0.000$'' means a very small deviation.}\label{tab2}
		 \renewcommand{\arraystretch}{1.2}
		\scalebox{0.9}{ 	
			 \begin{tabular*}{18.5cm}{ccccccccc}
				\toprule[1.2pt]
				 \multirow{2}*{$T$}&\multirow{2}*{$p_1$}&\multirow{2}*{$p_2$}&\multicolumn{2}{c}{$\hat k_1=\hat k_2=2$}&\multicolumn{2}{c}{$\hat k_1=\hat k_2=3$}&\multicolumn{2}{c}{$\hat k_1=\hat k_2=4$}\\\cline{4-9}
				 &&&PE&(0)PCA&PE&(0)PCA&PE&(0)PCA\\\midrule[1pt]
				20&20&20&{\bf 0.343(0.084)}&0.347(0.084)&{\bf 0.038(0.004)}&0.043(0.006)&0.080(0.007)&{\bf 0.077(0.006)}
\\
				50&20&50&{\bf 0.385(0.069)}&0.388(0.070)&{\bf 0.013(0.001)}&0.015(0.002)&0.031(0.002)&{\bf 0.029(0.001)}
\\
				100&20&100&{\bf 0.406(0.064)}&0.408(0.065)&{\bf 0.006(0.000)}&0.008(0.001)&0.015(0.001)&{\bf 0.014(0.001)}
\\
				150&20&150&{\bf 0.421(0.065)}&0.423(0.065)&{\bf 0.004(0.000)}&0.006(0.001)&0.010(0.000)&{\bf 0.009(0.000)}
\\
				200&20&200&{\bf 0.423(0.065)}&0.425(0.065)&{\bf 0.003(0.000)}&0.005(0.001)&0.008(0.000)&{\bf 0.007(0.000)}\\
				\bottomrule[1.2pt]		
		\end{tabular*}}		
	\end{center}
\end{table*}

\subsection{Estimating the numbers of factors}\label{sec:4.5}
As shown in Table \ref{tab2}, accurate specification of the numbers of factors is critical to the matrix factor model. In this subsection, we  compare the empirical performances of the Vectorized Eigenvalue-Ratio (VER) criterion in \cite{ahn2013eigenvalue}, $\alpha$-PCA based ER method ($\alpha$-PCA-ER) in \cite{Chen2020StatisticalIF},  and the proposed iterative method in Algorithm \ref{alg2} (IterER) in terms of estimating the numbers of factors.

Table \ref{tab3} presents the frequencies of exact estimation and underestimation over 500 replications under Setting A. For the proposed IterER method, we try $c=0,10^{-4},1$ in equation (\ref{equ2}). For the VER criterion in \cite{ahn2013eigenvalue}, we first vectorize the matrix observations and regard the true number of factors as $k_1\times k_2$. We set $k_{\max}=8$ for IterER and $\alpha$-PCA-ER while $k_{\max}=64$ for VER.  We find that the IterER has the highest accuracy and lowest underestimation risk even with small $T$, while the others only work when $T$ is large.  Under this setting, the constant $c$ and $\alpha$ seem to have a minimal effect on the results.


 \begin{table*}[hbtp]
 	\begin{center}
 		\small
 		\addtolength{\tabcolsep}{0pt}
 		\caption{The frequencies of exact estimation and underestimation (in parentheses) of the numbers of factors under Setting A over 500 replications. $p_1=20$, $T=p_2$. ``IterER($c$)'': the iterative algorithm with constant $c$ in the denominator. ``($a$)PCA-ER'': the $\alpha$-PCA-ER in \cite{Chen2020StatisticalIF} with $\alpha=a$. ``VER'': eigenvalue ratio estimation of the vectorized model.}\label{tab3}
 		 \renewcommand{\arraystretch}{1.2}
 		\scalebox{0.85}{ 	
 			 \begin{tabular*}{19.2cm}{cccccccc}
 				\toprule[1.2pt]
 			 $T$&IterER(0)&IterER($10^{-4}$)&IterER(1)&(-1)PCA-ER&(0)PCA-ER&(1)PCA-ER&VER\\\midrule[1pt]
 			20&{\bf 0.996(0.004)}&{\bf0.996(0.004)}&{\bf0.996(0.004)}&0.608(0.388)&0.630(0.368)&0.620(0.380)&0(1)
\\
 			 50&{\bf1(0)}&{\bf1(0)}&{\bf1(0)}&0.880(0.120)&0.884(0.116)&0.882(0.116)&0.766(0.234)
\\
 			 100&{\bf1(0)}&{\bf1(0)}&{\bf1(0)}&0.914(0.086)&0.912(0.088)&0.912(0.088)&0.994(0.006)\\
 			 150&{\bf1(0)}&{\bf1(0)}&{\bf1(0)}&0.908(0.092)&0.910(0.090)&0.906(0.092)&1(0)
\\
 			 200&{\bf1(0)}&{\bf1(0)}&{\bf1(0)}&0.884(0.114)&0.884(0.114)&0.884(0.112)&0.998(0.002)\\
 				\bottomrule[1.2pt]		
 		\end{tabular*}}		
 	\end{center}
 \end{table*}

As one reviewer pointed out, the specification criterion of the numbers of factors deserves more concern when the mean of the data matrix is not zero. If the data are not demeaned suitably, this may enlarge the first eigenvalue of the corresponding matrix and result in underestimation of the numbers of factors. We propose to first demean the vectorized data either by subtracting the sample mean or adopting the double-demeaned strategy by \cite{ahn2013eigenvalue}, and then structure the demeaned vectors into matrices. We further design the following two settings with nonzero mean  to investigate the empirical performances of different methods.

\vspace{0.5em}

\noindent\textbf{Setting D}: (nonzero mean of factors). Data are generated similarly to Setting A except that  $\text{Vec}(\Fb_t)=\phi \times \text{Vec}(\Fb_{t-1})+\sqrt{1-\phi^2}\times\bepsilon_t$, $\bepsilon_t\overset{i.i.d.}{\sim}\mathcal{N}({\bf 1},\Ib_{k_1\times k_2})$.

\vspace{0.5em}

\noindent\textbf{Setting F}: (nonzero mean of entries). Data are generated similarly to Setting A except that $\Xb_t=\bmu+\Rb^\top\Fb_t\Cb+\Eb_t$ with the mean matrix $\bmu=(\mu_{ij})_{p_1\times p_2}$ independently sampled from $\mathcal{N}(0,1)$.

\vspace{0.5em}

Under Setting D, the non-zero mean term is from the factors, so we need not demean the data for $\alpha$-PCA-ER according to \cite{Chen2020StatisticalIF}. For the IterER and the VER criterion, we  try the two aforementioned demean strategies.  Table \ref{tab4} reports the empirical frequencies of exact estimation and underestimation over 500 replications. Both demean strategies work well for the IterER and VER, while the performance of the $\alpha$-PCA-ER  is not satisfactory due to the small value of $p_1$. Note also that taking $\alpha=-1$ is equivalent to subtracting the sample mean for the  $\alpha$-PCA-ER method.

 \begin{table*}[hbtp]
	\begin{center}
		\small
		 \addtolength{\tabcolsep}{0.5pt}
		\caption{The frequencies of exact estimation and underestimation (in parentheses) of the numbers of factors under Setting D: {nonzero mean of factors.} $p_1=20$, $T=p_2$. ``-S'' stands for subtracting the sample mean strategy. ``-D'' stands for the double-demeaned strategy by \cite{ahn2013eigenvalue}.}\label{tab4}
		 \renewcommand{\arraystretch}{1.2}
		\scalebox{0.85}{ 	
			 \begin{tabular*}{19.2cm}{cccccccc}
				\toprule[1.2pt]
				 $T$&IterER(0)-S&IterER(0)-D&(-1)PCA-ER&(0)PCA-ER&(1)PCA-ER&VER-S&VER-D\\\midrule[1pt]
			20&{\bf 0.988(0.012)}&{\bf 0.988(0.012)}&0.636(0.362)&0.218(0.782)&0.060(0.940)&0(0)&0(0)
\\
			50&{\bf 1(0)}&{\bf 1(0)}&0.890(0.108)&0.358(0.642)&0.106(0.894)&0.760(0.240)&0.766(0.234)
\\
			100&{\bf 1(0)}&{\bf 1(0)}&0.904(0.094)&0.400(0.600)&0.102(0.898)&0.994(0.006)&0.994(0.006)
\\
			150&{\bf 1(0)}&{\bf 1(0)}&0.890(0.098)&0.406(0.594)&0.088(0.912)&1(0)&1(0)
\\
			200&{\bf 1(0)}&{\bf 1(0)}&0.916(0.076)&0.366(0.634)&0.086(0.914)&1(0)&1(0)\\
			\bottomrule[1.2pt]		
		\end{tabular*}}		
	\end{center}
\end{table*}

Under Setting F, the non-zero mean term is not from the factors and we always demean the data before estimating the numbers of factors by  all three methods. Table \ref{tab5} reports the numeric results. The parameter $\alpha$  has a minimal effect on the results, hence we only present the results with $\alpha=0$ for the $\alpha$-PCA-ER method. Both demean strategies work well and  the IterER performs  well whether $T$ is small or large.

 \begin{table*}[hbtp]
	\begin{center}
		\small
		\addtolength{\tabcolsep}{4pt}
		\caption{The frequencies of exact estimation and underestimation (in parentheses) of the factor numbers under Setting F: nonzero  mean of individuals. $p_1=20$, $T=p_2$. ``-S'' stands for the subtracting sample mean strategy. ``-D'' stands for the double-demean strategy by \cite{ahn2013eigenvalue}.}\label{tab5}
		 \renewcommand{\arraystretch}{1.2}
		\scalebox{0.85}{ 	
			 \begin{tabular*}{19.2cm}{ccccccc}
				\toprule[1.2pt]
				 $T$&IterER(0)-S&IterER(0)-D&(0)PCA-ER-S&(0)PCA-ER-D&VER-S&VER-D\\\midrule[1pt]
20&{\bf 0.980(0.02)}&{\bf 0.984(0.016)}&0.602(0.398)&0.666(0.334)&0(0)&0(0)
\\
50&{\bf 1(0)}&{\bf 1(0)}&0.878(0.120)&0.890(0.110)&0.754(0.246)&0.758(0.242)
\\
100&{\bf 1(0)}&{\bf 1(0)}&0.924(0.076)&0.924(0.076)&0.988(0.012)&0.988(0.012)\\
150&{\bf 1(0)}&{\bf 1(0)}&0.884(0.112)&0.888(0.110)&0.998(0.002)&0.998(0.002)
\\
200&{\bf 1(0)}&{\bf 1(0)}&0.920(0.076)&0.920(0.078)&0.998(0.002)&0.998(0.002)\\
				\bottomrule[1.2pt]		
		\end{tabular*}}		
	\end{center}
\end{table*}

\subsection{Evaluation of  recursive procedure}\label{sec:iterativealg}
Finally, we consider a recursive projection procedure before ending the simulation studies.  The projected method can be recursively implemented by setting the newly estimated loadings $\tilde\Rb$ and $\tilde\Cb$ as initial projection matrices. A  question arises naturally: can we benefit from more iterative steps?

To this end, we conduct simulaions under Setting A ($p_1=20$, $T=p_2=200$) and Setting B ($p_2=20$, $T=p_1=200$). We start with the (0)-PCA  estimators, denoted as $\hat\Rb^{(1)}$ and $\hat\Cb^{(1)}$. At step $t+1$, we set $\hat\Rb^{(t)}$ and $\hat\Cb^{(t)}$ as the projection matrices to calculate the projected estimators, denoted as $\hat\Rb^{(t+1)}$ and $\hat\Cb^{(t+1)}$. At each step, the estimation errors of the corresponding  loading spaces are recorded.

Figure \ref{fig5} shows the averaged errors at each step over 500 replications. Under Setting A, the red real line shows a significant drop at the second step, corresponding to the reduced  estimation error from $\hat\Rb$ to $\tilde\Rb$. However, the error cannot be further decreased even though more iterative steps are involved.  The blue dashed line is flat because the estimation error of $\Cb$ is dominated by $1/\sqrt{Tp_1}$ at all steps when $p_1$ is small. Similar patterns are observed under Setting B. We conclude that our method performs satisfactorily with a single  projection step.

 \begin{figure}[hbpt]
	\begin{subfigure}{.5\textwidth}
		\centering
		 \includegraphics[width=7.5cm,height=7.5cm]{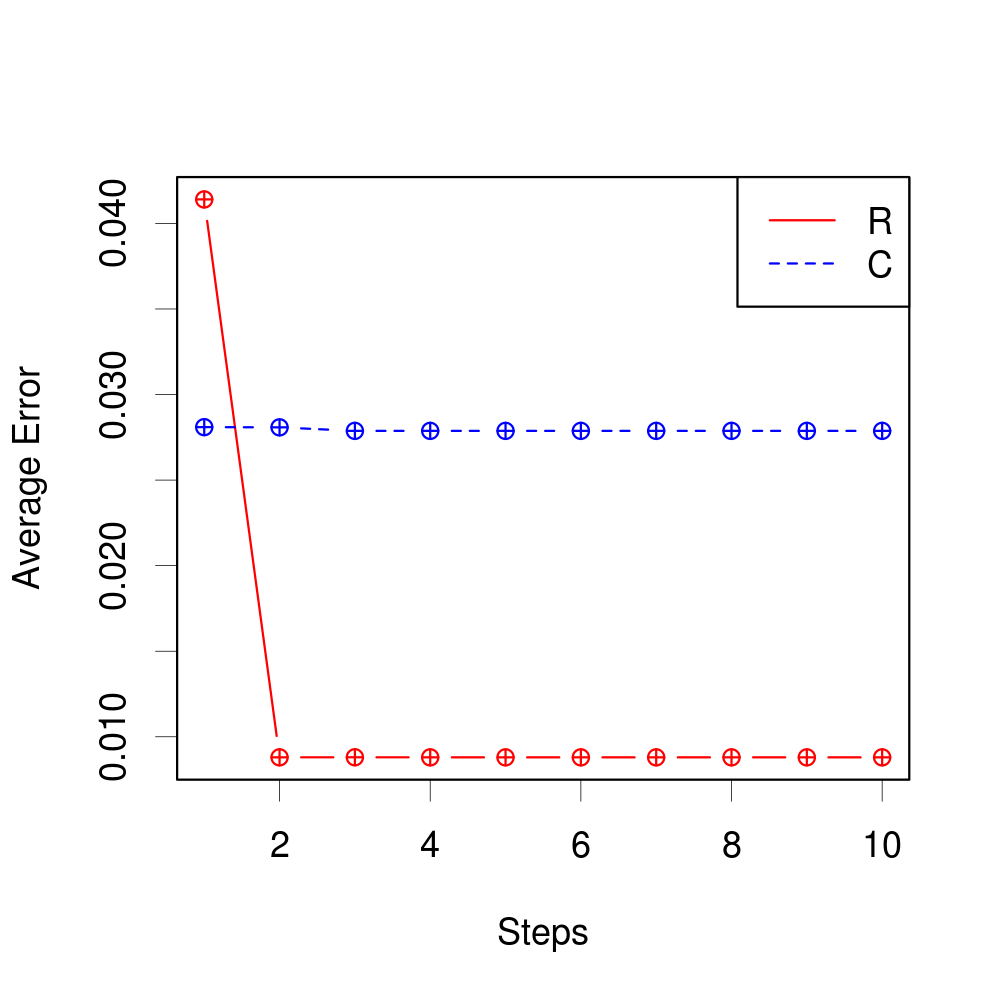}
		\caption{Setting A with $p_1=20$, $T=p_2=200$}
	\end{subfigure}
	\begin{subfigure}{.5\textwidth}
		\centering
		 \includegraphics[width=7.5cm,height=7.5cm]{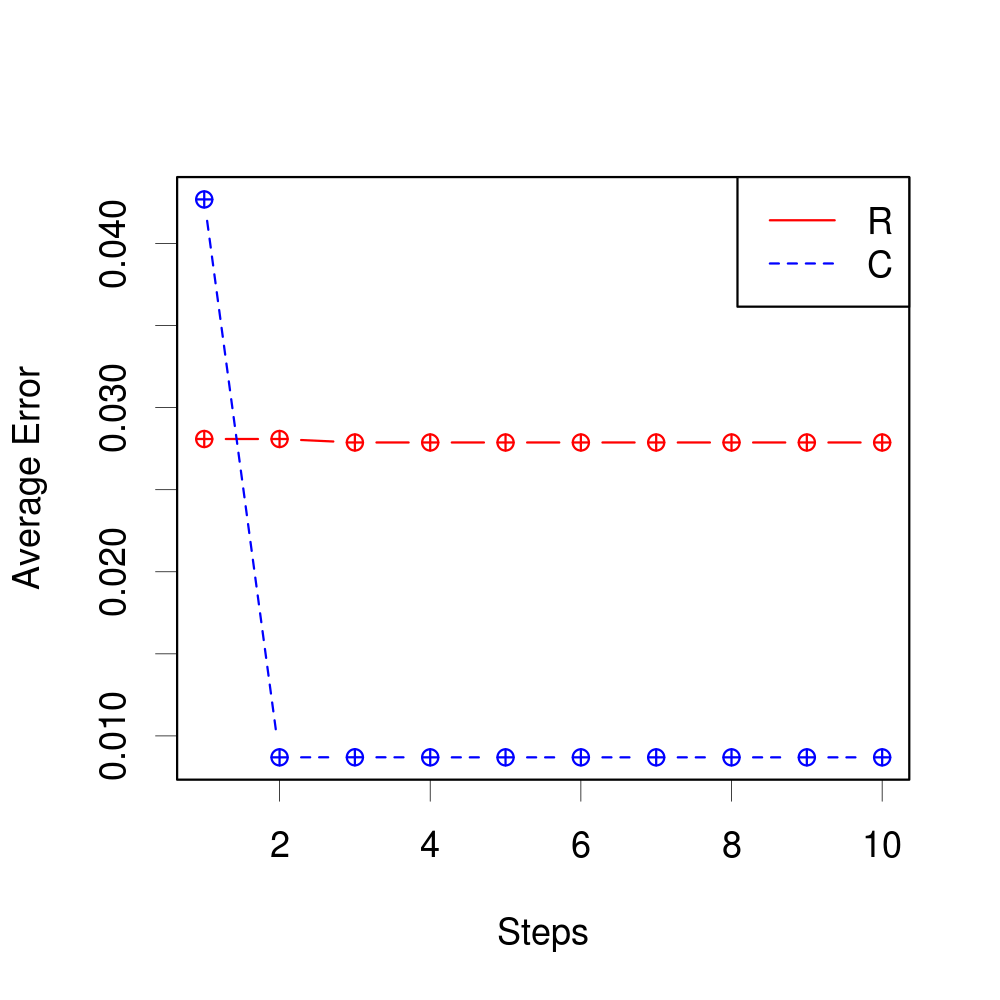}
		\caption{Setting B with $p_2=20$, $T=p_1=200$}
	\end{subfigure}
	\caption{Mean estimation error at each step of the recursive procedure over 500 replications.}\label{fig5}
\end{figure}

\section{Real data analysis}\label{sec5}
\subsection{Fama--French $10\times10$ portfolios}
 For ease of  comparison, in our first real example we use the same dataset as that used by \cite{wang2019factor}. This dataset consists of monthly returns of  100 portfolios, structured into a $10\times10$ matrix according to 10 levels of market capital size (S1-S10) and 10 levels of book-to-equity ratio (BE1-BE10). The monthly returns from January 1964 to December 2019 are collected, covering 672 months. Detailed information can be found on the website \url{http://mba.tuck.dartmouth.edu/pages/faculty/ken.french/data_library.html}.

 Following \cite{wang2019factor}, we adjusted the return series  first  by subtracting the corresponding monthly market excess returns. In the next step, we impute the missing values by linear interpolation for each series. The augmented Dickey--Fuller test rejects the null hypothesis for all the series, indicating stationality.  With the standardized monthly returns, our iterative eigenvalue-ratio method suggests that $k_1=2$ while $k_2=1$. For better illustration, we take $k_1=k_2=2$. The estimated front and back loading matrices after varimax rotation and scaling are reported in Tables \ref{tab6} and \ref{tab7}.

  \begin{table*}[hbpt]
 	\begin{center}
 		\small
 		\addtolength{\tabcolsep}{4pt}
 		\caption{Size loading matrix for Fama--French data set, after varimax rotation and scaling by 30.  ``PE" stands for the projected estimator, ``ACCE" stands for the  approach used by \cite{wang2019factor}, while $\alpha$-PCA represents the method in \cite{Chen2020StatisticalIF} with $\alpha=0$. }\label{tab6}
 		 \renewcommand{\arraystretch}{1}
 		\scalebox{1}{ 		 \begin{tabular*}{16cm}{cc|cccccccccc}
 				\toprule[1.2pt]
 				 Method&Factor&S1&S2&S3&S4&S5&S6&S7&S8&S9&S10\\\hline
 				 \multirow{2}*{PE}&1&\cellcolor {Lavender}-16&\cellcolor {Lavender}-15&\cellcolor {Lavender}-12&\cellcolor {Lavender}-10&\cellcolor {Lavender}-8&-5&-2&-1&4&\cellcolor {Lavender}7\\
 				&2&-5&-1&3&5&\cellcolor {Lavender}8&\cellcolor {Lavender}10&\cellcolor {Lavender}12&\cellcolor {Lavender}13&\cellcolor {Lavender}15&\cellcolor {Lavender}10\\\hline
 				 \multirow{2}*{ACCE}&1&\cellcolor {Lavender}-12&\cellcolor {Lavender}-14&\cellcolor {Lavender}-12&\cellcolor {Lavender}-13&\cellcolor {Lavender}-10&\cellcolor {Lavender}-6&-3&-1&5&\cellcolor {Lavender}9
 				\\
 				 &2&-1&-1&-1&2&5&\cellcolor {Lavender}10&\cellcolor {Lavender}11&\cellcolor {Lavender}18&\cellcolor {Lavender}15&\cellcolor {Lavender}10
 				\\\hline
 	 			 \multirow{2}*{$\alpha$-PCA}&1&\cellcolor {Lavender}14&\cellcolor {Lavender}14&\cellcolor {Lavender}13&\cellcolor {Lavender}11&\cellcolor {Lavender}9&\cellcolor {Lavender}6&4&1&-3&\cellcolor {Lavender}-8
 	\\
 	&2&-4&-2&2&3&6&\cellcolor {Lavender}9&\cellcolor {Lavender}12&\cellcolor {Lavender}13&\cellcolor {Lavender}16&\cellcolor {Lavender}13
 	\\
 				\bottomrule[1.2pt]		
 		\end{tabular*}}		
 	\end{center}
 \end{table*}

 \begin{table*}[hbpt]
 	\begin{center}
 		\small
 		\addtolength{\tabcolsep}{1pt}
 		\caption{Book-to-Equity (BE) loading matrix for Fama--French data set, after varimax rotation and scaling by 30. ``PE" is for the projected estimator, ``ACCE" stands for the  approach in \cite{wang2019factor}, and $\alpha$-PCA is for the method in \cite{Chen2020StatisticalIF} with $\alpha=0$.}\label{tab7}
 		 \renewcommand{\arraystretch}{1}
 		\scalebox{1}{ 		 \begin{tabular*}{16cm}{cc|cccccccccc}
 				\toprule[1.2pt]
 				 Method&Factor&BE1&BE2&BE3&BE4&BE5&BE6&BE7&BE8&BE9&BE10\\\hline
 				 \multirow{2}*{PE}&1&\cellcolor {Lavender}6&1&-4&\cellcolor {Lavender}-7&\cellcolor {Lavender}-10&\cellcolor {Lavender}-11&\cellcolor {Lavender}-12&\cellcolor {Lavender}-12&\cellcolor {Lavender}-12&\cellcolor {Lavender}-10
 				\\
 				&2&\cellcolor {Lavender}20&\cellcolor {Lavender}17&\cellcolor {Lavender}11&\cellcolor {Lavender}8&4&2&0&-1&-1&0\\\hline
 				 \multirow{2}*{ACCE}&1&\cellcolor {Lavender}6&-1&-4&\cellcolor {Lavender}-8&\cellcolor {Lavender}-8&\cellcolor {Lavender}-10&\cellcolor {Lavender}-10&\cellcolor {Lavender}-13&\cellcolor {Lavender}-14&\cellcolor {Lavender}-12
\\
&2&\cellcolor {Lavender}-21&\cellcolor {Lavender}-15&\cellcolor {Lavender}-11&\cellcolor {Lavender}-7&-5&-2&-1&2&3&-1
\\\hline
 				 \multirow{2}*{$\alpha$-PCA}&1&\cellcolor {Lavender}6&2&-4&\cellcolor {Lavender}-7&\cellcolor {Lavender}-10&\cellcolor {Lavender}-11&\cellcolor {Lavender}-12&\cellcolor {Lavender}-13&\cellcolor {Lavender}-12&\cellcolor {Lavender}-11
\\
&2&\cellcolor {Lavender}19&\cellcolor {Lavender}18&\cellcolor {Lavender}12&\cellcolor {Lavender}8&4&2&0&-1&-1&-1
\\ 				
 				\bottomrule[1.2pt]		
 		\end{tabular*}}		
 	\end{center}
 \end{table*}

From these tables, we observe that the PE, $\alpha$-PCA and Auto-Cross-Correlation Estimation (ACCE) method by  \cite{wang2019factor} lead to very similar estimated loadings. From the perspective of size, the small size portfolios load heavily on the first factor while the large size portfolios load mainly on the second factor. From the perspective of book-to-equity, the small BE portfolios load heavily on the second factor while the large BE portfolios load mainly on the first factor. Clearly, the portfolios tend to perform more similarly if they are constructed by public companies with similar size and book-to-equity ratio.

To further compare these methods, we use a similar rolling-validation procedure as in \cite{wang2019factor}. For each year $t$ from 1996 to 2019, we repeatedly use the $n$ (bandwidth) years observations before $t$ to fit the matrix-variate factor model and estimate the two loading matrices. The loadings are then used to estimate the  factors and corresponding residuals of the 12 months in the current year. Specifically, let $\Yb_t^i$ and $\hat\Yb_t^i$ be the observed and estimated price matrix of month $i$ in year $t$, denote $\bar {\Yb}_t$ as the mean price matrix, and further define
\[
\text{MSE}_t=\frac{1}{12\times10\times10}\sum_{i=1}^{12}\|\hat\Yb_t^i-\Yb_t^i\|_F^2, \quad \rho_t=\frac{\sum_{i=1}^{12}\|\hat\Yb_t^i-\Yb_t^i\|_F^2}{\sum_{i=1}^{12}\|\hat\Yb_t^i-\bar\Yb_t\|_F^2},
\]
as the mean squared pricing error and unexplained proportion of total variances, respectively.  During the rolling-validation procedure, the variation of loading space is measured by $v_t:=\mathcal{D}(\hat\Cb_t\otimes \hat\Rb_t,\hat\Cb_{t-1}\otimes \hat\Rb_{t-1})$. The matrix factor model (\ref{mod1.1}) can be written in vector form with $(\Cb\otimes \Rb)$ being the loading matrix.

Table \ref{tab8} reports the means of MSE, $\rho$ and $v$ by PE, ACCE,  $\alpha$-PCA and a conventional PCA estimation applied to the vectorized data. Diversified combinations of bandwidth $n$ and numbers of factors ($k_1=k_2=k$) are compared. On the one hand,	the pricing errors of PE, $\alpha$-PCA, and the vector model are very close especially for large $n$ and $k$, but lower than the ACCE method. On the other hand, in terms of estimating the loading space,  PE  always performs much more stably compared with the other two methods. Financial data are usually heavily-tailed with outliers, so the more robust PE method is  preferred  to control transaction costs and reduce risks.

 \begin{table*}[hbtp]
	\begin{center}
		\small
		\addtolength{\tabcolsep}{0pt}
		\caption{Rolling validation of  Fama--French dataset. $12n$ is the sample size of the training set. $k_1=k_2=k$ is the number of factors. $\overline{MSE}$, $\bar \rho$, $\bar v$ are the mean pricing error, mean unexplained proportion of total variances and mean variation of the estimated loading space. ``PE'' is the projected method. ``ACCE'' is the method in \cite{wang2019factor}. ``$\alpha$-PCA'' is the method in \cite{Chen2020StatisticalIF} with $\alpha=0$. ``Vec'' is the PCA applied to vectorized data. }\label{tab8}
		 \renewcommand{\arraystretch}{1.2}
		\scalebox{0.85}{ 	
			 \begin{tabular*}{19.5cm}{cc|cccc|cccc|cccc}
				\toprule[1.2pt]
				 \multirow{2}{*}{$n$}&\multirow{2}{*}{$k$}&\multicolumn{4}{c|}{$\overline{MSE}$}&\multicolumn{4}{c|}{$\bar{\rho}$}&\multicolumn{4}{c}{$\bar{v}$}\\
				 &&PE&ACCE&$\alpha$-PCA&Vec&PE&ACCE&$\alpha$-PCA&Vec&PE&ACCE&$\alpha$-PCA&Vec\\\midrule[1pt]
5&1&0.869&0.883&{\bf0.863}&0.910&{\bf0.803}&0.827&0.797&0.841&{\bf0.173}&0.303&0.238&0.252
\\
10&1&{\bf0.855}&0.880&0.860&0.939&{\bf0.785}&0.815&0.793&0.864&{\bf0.084}&0.166&0.200&0.180
\\
15&1&{\bf0.853}&0.884&0.861&0.894&{\bf0.783}&0.812&0.793&0.816&{\bf0.063}&0.153&0.232&0.131
\\\hline
5&2&0.596&0.668&0.602&{\bf0.590}&{\bf0.624}&0.674&0.630&0.626&{\bf0.216}&0.473&0.341&0.398
\\
10&2&0.603&0.658&0.613&{\bf0.595}&0.629&0.672&0.638&{\bf0.625}&{\bf0.089}&0.260&0.261&0.212
\\
15&2&0.604&0.639&0.614&{\bf0.590}&0.628&0.653&0.632&{\bf0.617}&{\bf0.059}&0.191&0.175&0.192
\\\hline
5&3&{\bf0.526}&0.566&0.533&0.534&{\bf0.554}&0.593&0.561&0.567&{\bf0.293}&0.512&0.451&0.477
\\
10&3&{\bf0.526}&0.575&0.531&0.532&{\bf0.554}&0.595&0.559&0.564&{\bf0.139}&0.301&0.406&0.301
\\
15&3&{\bf0.522}&0.562&0.526&0.527&0.549&0.584&{\bf0.548}&0.557&{\bf0.090}&0.278&0.345&0.257\\
				\bottomrule[1.2pt]		
		\end{tabular*}}		
	\end{center}
\end{table*}

\subsection{Multinational macroeconomic indices}
In the second real example, we analyze a multinational macroeconomic index dataset collected from OECD using the proposed method. A similar dataset is  studied in \cite{Chen2020StatisticalIF}. It contains 10 macroeconomic indices across 8 countries over 130 quarters from 1988-Q1 to 2020-Q2. The countries are the United States, the United Kingdom, Canada, France, Germany, Norway, Australia and New Zealand. The indices are from 4 major groups, namely consumer price, interest rate, production, and international trade. Logarithm transform and difference operators are applied to each of the series according to \cite{Chen2020StatisticalIF} so that the $\alpha$-mixing assumption is satisfied. Detailed description can be found in our supplementary material. We further standardize each of the transformed series to avoid the effects of non-zero mean or diversified  variances.

The first step is to determine the numbers of row and column factors. The proposed iterative algorithm suggests taking $k_1=1$ and $k_2=5$ for the $ 8\times 10$ matrix-valued observations. For better illustration, we take $k_1=3$ and $k_2=4$ such that the row and column factors can explain nearly 75\% variances of the matrices $\tilde \Mb_1$ and $\tilde\Mb_2$. The  detailed results are reported in Tables \ref{tab9} and \ref{tab10}.

\begin{table*}[hbpt]
	\begin{center}
		\small
		\addtolength{\tabcolsep}{12pt}
		\caption{Row (countries)  loading matrices by PE, ACCE and $\alpha$-PCA ($\alpha=0$) for multinational macroeconomic indices dataset, varimax rotated, and multiplied by 10.  }\label{tab9}
		 \renewcommand{\arraystretch}{1}
		\scalebox{0.75}{ 		 \begin{tabular*}{21.5cm}{cc|cccccccc}
				\toprule[1.2pt]
				 Method&Factor&AUS&NZL&USA&CAN&NOR&DEU&FRA&GBR
				\\\midrule[1pt]
				 \multirow{3}*{PE}&1&0&1&\cellcolor {Lavender}-7&\cellcolor {Lavender}-6&3&-3&-2&-1
\\
				&2&2&-2&1&-1&\cellcolor {Lavender}-7&-2&\cellcolor {Lavender}-5&\cellcolor {Lavender}-5
\\
				&3&\cellcolor {Lavender}8&\cellcolor {Lavender}6&0&-1&0&2&-1&1
				\\\hline
				 \multirow{3}*{ACCE}&1&2&-2&1&-2&\cellcolor {Lavender}-6&0&\cellcolor {Lavender}-6&\cellcolor {Lavender}-5
\\
				&2&\cellcolor {Lavender}7&\cellcolor {Lavender}5&0&0&0&\cellcolor {Lavender}5&0&0
\\
				&3&0&-2&\cellcolor {Lavender}8&\cellcolor {Lavender}4&-2&1&1&2
				\\\hline
				 \multirow{3}*{$\alpha$-PCA}&1&-1&1&\cellcolor {Lavender}-7&\cellcolor {Lavender}-5&3&-3&-2&-1
\\
			&2&	1&-1&0&-1&\cellcolor {Lavender}-7&-2&\cellcolor {Lavender}-5&\cellcolor {Lavender}-4
\\
			&3&	\cellcolor {Lavender}-7&\cellcolor {Lavender}-7&0&1&0&-1&1&-1\\
				\bottomrule[1.2pt]		
		\end{tabular*}}

	\end{center}
\end{table*}

\begin{table*}[hbpt]
	\begin{center}
		\small
		 \addtolength{\tabcolsep}{0.5pt}
		\caption{Column (indices) loading matrices by PE, ACCE, and $\alpha$-PCA ($\alpha=0$) for multinational macroeconomic indices dataset, varimax rotated and multiplied by 10.  }\label{tab10}
		 \renewcommand{\arraystretch}{1}
	
		\scalebox{0.75}{ 		 \begin{tabular*}{21.5cm}{cc|cccccccccc}
				\toprule[1.2pt]
				 Method&Factor&CPI:Tot&CPI:Ener&CPI:NFNE&IR:3-Mon&IR:Long&P:TIEC&P:TM&GDP&IT:Ex&IT:Im
				\\\midrule[1pt]
				 \multirow{4}*{PE}&1&1&-2&3&1&-1&\cellcolor {Lavender}6&\cellcolor {Lavender}7&2&-1&0
\\
				&2&\cellcolor {Lavender}6&\cellcolor {Lavender}7&3&-1&1&0&0&0&0&0
\\
				&3&0&0&-1&\cellcolor {Lavender}-6&\cellcolor {Lavender}-8&0&0&1&0&-1\\
				 &4&1&-2&3&0&0&-1&0&\cellcolor {Lavender}-5&\cellcolor {Lavender}-6&\cellcolor {Lavender}-5\\\hline
				 \multirow{4}*{ACCE}&1&0&0&0&0&1&\cellcolor {Lavender}-7&\cellcolor {Lavender}-7&0&0&0
\\
			&2&	1&0&0&\cellcolor {Lavender}-5&\cellcolor {Lavender}-4&1&-1&-2&\cellcolor {Lavender}-4&\cellcolor {Lavender}-6
\\
			&3&	\cellcolor {Lavender}-4&2&\cellcolor {Lavender}-9&0&0&2&-1&2&0&0
\\
				&4&\cellcolor {Lavender}6&\cellcolor {Lavender}7&0&-1&3&0&0&2&1&-1\\\hline
				 \multirow{4}*{$\alpha$-PCA}&1&0&-1&1&1&-1&\cellcolor {Lavender}7&\cellcolor {Lavender}6&\cellcolor {Lavender}4&0&0
\\
				&2&\cellcolor {Lavender}7&\cellcolor {Lavender}5&\cellcolor {Lavender}5&-1&1&0&1&0&0&0
\\
				&3&0&0&0&\cellcolor {Lavender}-7&\cellcolor {Lavender}-7&1&0&-1&1&0\\
				 &4&0&2&-2&0&0&0&0&2&\cellcolor {Lavender}7&\cellcolor {Lavender}6\\
							 \bottomrule[1.2pt]		
				
				\end{tabular*}}
	\end{center}
\end{table*}

For the row factors, Table \ref{tab9} shows that they are closely related to the geographical location. The neighboring countries tend to load similarly on the factors. From the estimated row loading matrix by the PE method, we observe that the Oceania countries load mainly on the third row factor, the north American countries load mainly on the first row factor, and the European countries load heavily on the second row factor. Therefore, the 7 countries (excluding Germany) naturally divides into 3 groups, which exactly match their geographical locations. The factors discovered by the PE and $\alpha$-PCA are almost the same. However, the factors discovered by the ACCE differ from those of PE and $\alpha$-PCA in the order of the leading geographical factors. Both PE and $\alpha$-PCA associate the first factor with the North American countries while the ACCE does so with the European countries.

For the column factors, the macroeconomic indices  are divided into 4 groups in Table \ref{tab10} which coincides  with the economic interpretation. According to the estimated column loading matrix by the PE method,  the indices corresponding to price load heavily on the second factor, the interest rate indices load  on the third factor, the production indices load mainly on the first factor, while international trade indices load mainly on the fourth factor. We observe that all three methods indicate the same first factor. Again, the PE and $\alpha$-PCA have the same second to fourth factors, but they are different from those of ACCE.

Next, a rolling-validation procedure is applied to each of the  mentioned methods as well as a vectorized PCA approach to further compare their performances. In view of the small sample size in this example, for each quarter $t$ from 2008-Q1 to 2020-Q2, we repeatedly estimate the matrix or vector valued factor models based on 80 observations before $t$. The estimated loadings are then used to calculate the mean squared error and variation of loading space at $t$.

  \begin{table*}[hbtp]
 	\begin{center}
 		\small
 		\addtolength{\tabcolsep}{10pt}
 		\caption{Rolling validation of the multinational macroeconomic index dataset. $k_1=k_2=k$.}\label{tab11}
 		 \renewcommand{\arraystretch}{1.2}
 		\scalebox{0.85}{ 	
 			 \begin{tabular*}{19.5cm}{c|cccc|cccc}
 				\toprule[1.2pt]
 				 \multirow{2}{*}{$k$}&\multicolumn{4}{c|}{$\overline{MSE}$}&\multicolumn{4}{c}{$\bar{v}$}\\
 				 &PE&ACCE&$\alpha$-PCA&Vec&PE&ACCE&$\alpha$-PCA&Vec\\\midrule[1pt]
2&0.749&0.835&0.727&{\bf 0.652}&0.058&{\bf 0.053}&0.327&0.061
\\
3&0.594&0.623&0.604&{\bf 0.486}&{\bf 0.041}&0.047&0.235&0.115
\\
4&0.486&0.519&0.468&{\bf 0.353}&{\bf 0.067}&0.106&0.219&0.125
\\
5&0.351&0.407&0.351&{\bf 0.227}&{\bf 0.045}&0.083&0.182&0.140\\
 				\bottomrule[1.2pt]		
 		\end{tabular*}}		
 	\end{center}
 \end{table*}

The averaged MSE and variation are reported in Table \ref{tab11} with various specified numbers of factors. The results are similar to those of the Fama--French dataset.  The projected estimator seems to be more stable than the other approaches. The PCA method for vectorized data leads to the smallest validation error in this example possibly due to deeper complexity in the sense of more parameters. Furthermore its averaged variation of the loading space is larger than the PE method. The validation errors for the PE and $\alpha$-PCA methods are comparable, but lower than those of ACCE.

  \begin{table*}[hbtp]
	\begin{center}
		\small
		\addtolength{\tabcolsep}{10pt}
		\caption{Mean absolute prediction error for the Consumer Price Index of Canada by different models.  }\label{tab12}
		 \renewcommand{\arraystretch}{1.8}
		\scalebox{0.68}{ 	
			 \begin{tabular*}{22.5cm}{cccccccccc}
				\toprule[1.2pt]
				 $(k_1,k_2)$&\textbf{Model 1}&\textbf{Model 2}&\textbf{Model 3}&\textbf{Model 4: PE}&\textbf{Model 4: ACCE}&\textbf{Model 4: $\alpha$-PCA}\\\midrule[1pt]
(3,3)&0.852&0.8171&0.6552&0.5988&{\bf0.5940}&0.6185
\\
(3,4)&0.852&0.7968&0.6658&{\bf0.6231}&0.6334&0.6594
\\
(4,4)&0.852&0.7968&0.6195&{\bf0.6039}&0.7182&0.6101\\
				\bottomrule[1.2pt]		
		\end{tabular*}}		
	\end{center}
\end{table*}

At last, we evaluate the practical utility of different methods by a rolling prediction procedure. At each quarter $t$, denote $y_t$ as the total CPI (CPI:Tot) in Canada, $\bx_t$ as the vector of all the other 9 indices in Canada, and $\Zb_t$ as the $8\times 9$ matrix observation of all the other 9 indices in all countries. We predict $y_{t+1}$ using the following Auto-Regression (AR) model (Model 1) and Factor-Augmented-Auto-Regression (FAAR) models (Models 2--4).
\begin{description}
	\item[Model 1]  $y_{t+1}=a+by_t+\epsilon_{t+1}$,
	\item[Model 2]  $y_{t+1}=a+by_t+\bbeta^\top\bbf_{1t}+\epsilon_{t+1}$, where $\bbf_{1t}$'s are estimated  from the vector factor model with observations $\{\bx_t\}$.
	\item[Model 3] $y_{t+1}=a+by_t+\bbeta^\top\bbf_{2t}+\epsilon_{t+1}$, where $\bbf_{2t}$'s are estimated from the vector factor model with observations $\{\text{Vec}(\Zb_t)\}$.
	\item[Model 4]  $y_{t+1}=a+by_t+\bbeta^\top\text{Vec}(\Fb_t)+\epsilon_{t+1}$, where $\Fb_t$'s are estimated from the matrix factor model with observations $\{\Zb_t\}$, by the PE, ACCE, and $\alpha$-PCA, respectively.
\end{description}
Model 1 is a simple auto-regression model. In Model 2, we add  common factors into the auto-regression model, which summarize the information of the other macroeconomic indices of Canada. In Model 3 and Model 4, the indices of the other countries are also taken into account and further summarized into several key factors by different methods. The coefficients $a$, $b$ and $\bbeta$ are estimated separately for different models with ordinary least squares.

For each quarter $t$ from 2008-Q1 to 2020-Q2, we always use the 80 neighboring observations before $t$ to train the models and predict $y_{t+1}$. The mean absolute prediction errors are reported in Table \ref{tab12} with different combinations of $k_1$ and $k_2$. In Model 2 and Model 3, the factor numbers are set as $k_2$ and $k_1\times k_2$, respectively. Comparing the prediction errors of Model 1  and Model 2, we conclude that taking the other macroeconomic indices of Canada into account  facilitates the prediction of the CPI, as what's expected. The prediction accuracy is further improved by a large margin by considering cross-country information, by comparing the prediction errors of Model 2 with those of Model 3 and Model 4. Moreover, for Model 4, the PE method tends to have smaller prediction errors than the ACCE and $\alpha$-PCA.

\section{Conclusions and discussions}\label{sec6}
The current study focuses on the estimation of matrix factor models. We start with the column or row sample covariances instead of the auto-cross covariances for the estimation of the front and back loading matrices. A projected approach is proposed to improve the estimation accuracy. Statistical convergence rates and asymptotic distributions of the estimated loadings are provided under mild conditions. An iterative approach is introduced to determine the numbers of factors. Thorough numerical studies and real examples show the advantages of the projected method over existing methods. The matrix factor models can be further extended to analyze high-order tensor data, such as video streaming, which are widely applied in recommender systems. This subject will be addressed in a future study. We are also interested in incorporating the matrix factor structure into estimating large-dimensional covariance matrices or detecting structure breaks.

\section{Acknowledgements}
He's work is supported by the National Key R$\&$D Program
of China (Grant No. 2018YFA0703900), the grant of National Science Foundation of China (Grant No. 11801316),  Natural Science Foundation
of Shandong Province (Grant No. ZR2019QA002) and the Fundamental Research Funds of Shandong University. Kong's work is partially supported by NSF China (Grant Nos. 71971118 and 11831008) and the WRJH-QNBJ Project and Qinglan Project of Jiangsu Province. Zhang's work is partially supported by NSF China (Grant No. 11971116).

\section{Supplementary Material}
The technical proofs of the main results  and details of the datasets  are included in the Supplementary Material.
  \bibliographystyle{Chicago}
  \bibliography{Ref}

\newpage
\begin{center}
\begin{large}
	{\bf Supplementary Material for ``Projected Estimation for Large-dimensional Matrix  Factor Models" }\\
\end{large}
	Long Yu\\
		School of Management, Fudan University, China\\
		Yong He\\
		Shandong University, China\\
		Xin-bing Kong\\
		Nanjing Audit University, China\\
		and\\
		Xinsheng Zhang\\
		School of Management, Fudan University,  China
\end{center}

\begin{appendices}
The supplementary material provides all the detailed proofs of the theorems and description of real data in the main text. It's structured as follows. Section \ref{seca} is for the consistency of the projected estimators, corresponding to Theorem \ref{thm1}. In Section \ref{secb}, we prove that the initial estimators satisfy the sufficient conditions (\ref{c1}) and (\ref{c2}), and show the convergence rates, corresponding to Theorem \ref{thm3}. Sections \ref{secc} and \ref{secd} are the proofs of Theorems \ref{thm4} and \ref{thm2}, i.e., the limiting distributions of the initial and projected estimators, respectively. Section \ref{sece} proves the consistency of factor and signal matrices in Theorem \ref{thm5}. Section \ref{secf} proves the consistency of the estimated factor numbers in Theorem \ref{thm6}.  Section \ref{secg} provides detailed description of the two data sets corresponding to our real applications.

	As $k_1,k_2$ are both fixed constants, without loss of generality, we  assume $k_1=k_2=1$ in some parts of the proof as long as it simplifies the notations. In that case, $\Rb$ and $\Cb$ are vectors rather than matrices and we use $R_i$ and $C_j$ to denote the corresponding entries. Similarly,  the common factors are denoted as $F_t$ instead of the bold $\Fb_t$. The roman characters $\mathcal{\uppercase\expandafter{\romannumeral1}},\mathcal{\uppercase\expandafter{\romannumeral2}},\ldots$ are used repeatedly in the proof but represent different terms in separate  sections.
	
	\section{Proof of Theorem \ref{thm1}: consistency of the projected estimators}\label{seca}
	\begin{proof}
		We focus on $\tilde\Cb$ and the results of $\tilde\Rb$ can be obtained by parallel procedure.  Note that  by definition, we have
	\begin{equation}\label{equa1}
	\begin{split}
	 \tilde\Mb_2=&\frac{1}{Tp_1^2p_2}\sum_{t=1}^{T}\Xb_t^\top\hat\Rb\hat\Rb^\top\Xb_t\\
	 =&\frac{1}{Tp_1^2p_2}\sum_{t=1}^{T}(\Cb\Fb_t^\top\Rb^\top+\Eb_t^\top)\hat\Rb\hat\Rb^\top(\Rb\Fb_t\Cb^\top+\Eb_t)\\
	 =&\frac{1}{Tp_1^2p_2}\sum_{t=1}^{T}\Cb\Fb_t^\top\Rb^\top\hat\Rb\hat\Rb^\top\Rb\Fb_t\Cb^\top+\frac{1}{Tp_1^2p_2}\sum_{t=1}^{T}\Eb_t^\top\hat\Rb\hat\Rb^\top\Rb\Fb_t\Cb^\top\\
	 &+\frac{1}{Tp_1^2p_2}\sum_{t=1}^{T}\Cb\Fb_t^\top\Rb^\top\hat\Rb\hat\Rb^\top\Eb_t+\frac{1}{Tp_1^2p_2}\sum_{t=1}^{T}\Eb_t^\top\hat\Rb\hat\Rb^\top\Eb_t\\
	 :=&\mathcal{\uppercase\expandafter{\romannumeral1}}+\mathcal{\uppercase\expandafter{\romannumeral2}}+\mathcal{\uppercase\expandafter{\romannumeral3}}+\mathcal{\uppercase\expandafter{\romannumeral4}}.
	\end{split}
	\end{equation}
	Denote $\tilde\bLambda_2$ as the diagonal matrix composed of the leading $k_2$ eigenvalues of $\tilde\Mb_2$, and
	\[
	 \tilde\Hb_2=\frac{1}{Tp_1^2p_2}\sum_{t=1}^{T}\Fb_t^\top\Rb^\top\hat\Rb\hat\Rb^\top\Rb\Fb_t\Cb^\top\tilde\Cb\tilde\bLambda_2^{-1},
	\]
	then by the definition of $\tilde\Cb$ we have $\tilde\Cb\tilde\bLambda_2=\tilde\Mb_2\tilde\Cb$, and
	\begin{equation}\label{equa2}
	 \tilde\Cb-\Cb\tilde\Hb_2=(\mathcal{\uppercase\expandafter{\romannumeral2}}+\mathcal{\uppercase\expandafter{\romannumeral3}}+\mathcal{\uppercase\expandafter{\romannumeral4}})\tilde\Cb\tilde\bLambda_2^{-1}.
	\end{equation}
	We will show that the diagonal entries of $\tilde\bLambda_2$ converge to some distinct positive constants in Lemma \ref{lema2}, then $\|\tilde\Hb_2\|=O_p(1)$. Further in Lemma \ref{lema3} we will have
	\[
	 \frac{1}{p_2}\|\mathcal{\uppercase\expandafter{\romannumeral2}}\tilde\Cb\|_F^2=O_p\bigg(\frac{1}{Tp_1}+w_2\bigg),\quad 		 \frac{1}{p_2}\|\ \mathcal{\uppercase\expandafter{\romannumeral3}}\tilde\Cb\|_F^2=O_p\bigg(\frac{1}{Tp_1}+w_2\bigg),
	\]
	while Lemma \ref{lema4} shows that
	\[
	\frac{1}{p_2}	 \|\mathcal{\uppercase\expandafter{\romannumeral4}\tilde\Cb}\|_F^2\lesssim\frac{1}{Tp_1p_2}+\frac{1}{p_1^2p_2^2}+w_1^2\times\bigg(\frac{1}{p_2^2}+\frac{1}{Tp_2}\bigg)+o_p(1)\times \frac{1}{p_2}\|\tilde\Cb-\Cb\tilde\Hb_2\|_F^2.
	\]
	Hence, the consistency in Theorem \ref{thm1} follows directly. It remains to show that $\tilde\Hb_2^\top\tilde\Hb_2\overset{p}{\rightarrow}\Ib_{k_2}$. This is easy because
	\[
	 \bigg\|\frac{1}{p_2}\Cb^\top(\tilde\Cb-\Cb\tilde\Hb_2)\bigg\|_F\le \bigg(\frac{\|\Cb\|_F^2}{p_2}\frac{\|\tilde\Cb-\Cb\tilde\Hb_2\|_F^2}{p_2}\bigg)^{1/2}= o_p(1),\quad\bigg\|\frac{1}{p_2}\tilde\Cb^\top(\tilde\Cb-\Cb\tilde\Hb_2)\bigg\|_F=o_p(1).
	\]
Note that $p_2^{-1}\tilde\Cb^\top\tilde\Cb=\Ib_{k_2}$ while $p_2^{-1}\Cb^\top\Cb\rightarrow\Ib_{k_2}$, then
\[
\Ib_{k_2}=\frac{1}{p_2}\tilde\Cb^\top\Cb\tilde\Hb_2+{\bf o_p(1)}=\tilde\Hb_2^\top\tilde\Hb_2+{\bf o_p(1)}.
\]
The row-wise consistency holds with the proof of asymptotic normality in Section \ref{secd}, which concludes Theorem \ref{thm1}.
\end{proof}

Before we move to the detailed investigation of the mentioned lemmas, the next Lemma  \ref{lema1} provides some technical bounds.
\begin{lemma}\label{lema1}
	Under Assumptions A to E, as $\min\{T,p_1,p_2\}\rightarrow\infty$, we have \\
(1). $\sum_{t=1}^{T}\mathbb{E}\|\Eb_t^\top\Rb\|_F^2=O(Tp_1p_2)$, $\sum_{t=1}^{T}\mathbb{E}\|\Eb_t\Cb\|_F^2=O(Tp_1p_2)$,\\
			(2). $\mathbb{E}\|\sum_{t=1}^T\Fb_t\Cb^\top\Eb_t^\top\|_F^2\le O(Tp_1p_2)$, $\mathbb{E}\|\sum_{t=1}^T\Fb_t^\top\Rb\Eb_t\|_F^2\le O(Tp_1p_2)$,\\ $\mathbb{E}\|\sum_{t=1}^T\Fb_t\Cb^\top\Eb_t^\top\Rb\|_F^2\le O(Tp_1p_2)$, $\mathbb{E}\|\sum_{t=1}^T\Fb_t^\top\Rb^\top\Eb_t\Cb\|_F^2\le O(Tp_1p_2)$,\\
	(3) it holds that for any $i\le p_1$,
	\[
\mathbb{E}\bigg\|\sum_{t=1}^{T}\Eb_t\be_{t,i\cdot}\bigg\|^2=O(Tp_1p_2+T^2p_2^2),\quad \mathbb{E}\bigg\|\sum_{t=1}^{T}\Rb^\top\Eb_t\be_{t,i\cdot}\bigg\|^2=O(Tp_1p_2+T^2p_2^2),\quad
	\]
\end{lemma}
\begin{proof} We assume $k_1=k_2=1$ in the proof. \\
	(1). For any $t,j$, by Assumption D.2
	\[
	\mathbb{E}(\be_{t,\cdot j}^\top\Rb)^2=\sum_{i_1,i_2}\mathbb{E}(R_{i_1}R_{i_2}e_{t,i_1j}e_{t,i_2j})\le \bar r^2\sum_{i_1,i_2}\bigg|\mathbb{E}(e_{t,i_1j}e_{t,i_2j})\bigg|\le c\bar r^2p_1.
	\]
	Hence, it's easy that $\sum_t\mathbb{E}\|\Eb_t^\top\Rb\|_F^2=O(Tp_1p_2)$. $\sum_t\mathbb{E}\|\Eb_t\Cb\|_F^2=O(Tp_1p_2)$ holds similarly. \\
	(2). The results hold directly by Assumption E.1.\\
	(3). On one hand, use Assumptions D then we  have
	\[
	\begin{split}
	 &\mathbb{E}\bigg\|\sum_{t=1}^{T}\Eb_t\be_{t,i\cdot}\bigg\|^2=\sum_{i_1}\mathbb{E}\bigg(\sum_t\sum_j(e_{t,ij}e_{t,i_1j}-\mathbb{E}e_{t,ij}e_{t,i_1j})\bigg)^2+\sum_{i_1}\bigg(\sum_t\sum_j\mathbb{E}e_{t,ij}e_{t,i_1j}\bigg)^2\\
	 \le&\sum_{i_1,t,j_1}\sum_{s,j_2}\text{Cov}(e_{t,ij_1}e_{t,i_1j_1},e_{s,ij_2}e_{s,i_1j_2})+\sum_{i_1,t,s,j_1,j_2}|\mathbb{E}e_{t,ij_1}e_{t,i_1j_1}||\mathbb{E}e_{s,ij_2}e_{s,i_1j_2}|\\
	 \le&cTp_1p_2+c\sum_{t,s,j_1,j_2}\bigg(\sum_{i_1}|\mathbb{E}e_{t,ij_1}e_{t,i_1j_1}|\bigg)=O(Tp_1p_2+T^2p_2^2).
	\end{split}
	\]
	On the other hand, assume $k_1=1$ so that
		\[
	\begin{split}
	 &\mathbb{E}\bigg\|\sum_{t=1}^{T}\Rb^\top\Eb_t\be_{t,i\cdot}\bigg\|^2\\
	 =&\mathbb{E}\bigg(\sum_t\sum_{i_1}\sum_j(R_{i_1}e_{t,ij}e_{t,i_1j}-\mathbb{E}R_{i_1}e_{t,ij}e_{t,i_1j})\bigg)^2+\bigg(\sum_t\sum_{i_1}\sum_j\mathbb{E}R_{i_1}e_{t,ij}e_{t,i_1j}\bigg)^2\\
	 \le&\sum_{t,i_1,j_1}\sum_{s,i_2,j_2}|R_{i_1}R_{i_2}|\bigg|\text{Cov}(e_{t,ij_1}e_{t,i_1j_1},e_{s,ij_2}e_{s,i_2j_2})\bigg|+\bigg(\sum_t\sum_{i_1}\sum_j|R_{i_1}|\bigg|\mathbb{E}e_{t,ij}e_{t,i_1j}\bigg|\bigg)^2\\
	\le&O(Tp_1p_2+T^2p_2^2),
	\end{split}
	\]
which concludes the lemma.
\end{proof}

\begin{lemma}\label{lema2}
	Under Assumptions A-F and the sufficient condition, as $\min\{T,p_1,p_2\}\rightarrow\infty$, for $j\le k_2$ we have
	\[
	 \lambda_j(\tilde\Mb_2)=\lambda_j(\bSigma_2)+o_p(1),
	\]
	where $\tilde \Mb_2=(Tp_2)^{-1}\sum_{t=1}^{T}\hat\Zb_t\hat\Zb_t^\top$ with $\hat\Zb_t=p_1^{-1}\Xb_t^\top\hat\Rb$.
\end{lemma}
\begin{proof}
	Recall that  by equation (\ref{equa1}),  $\tilde\Mb_2=\mathcal{\uppercase\expandafter{\romannumeral1}}+\mathcal{\uppercase\expandafter{\romannumeral2}}+\mathcal{\uppercase\expandafter{\romannumeral3}}+\mathcal{\uppercase\expandafter{\romannumeral4}}$.
	Firstly, without loss of generality we assume $k_1=k_2=1$, then
	\begin{small}
	\[
	\begin{split}
	 \mathcal{\uppercase\expandafter{\romannumeral1}}=&\frac{1}{Tp_1^2p_2}\Cb\bigg(\sum_{t=1}^{T}\Fb_t^\top\Rb^\top\hat\Rb\hat\Rb^\top\Rb\Fb_t\bigg)\Cb^\top\\
	=& \frac{1}{Tp_1^2p_2}\Cb\bigg(\sum_{t=1}^{T}\Fb_t^\top\Rb^\top(\Rb\hat\Hb_1+\hat\Rb-\Rb\hat\Hb_1)(\Rb\hat\Hb_1+\hat\Rb-\Rb\hat\Hb_1)^\top\Rb\Fb_t\bigg)\Cb^\top\\
	= &\frac{1}{Tp_1^2p_2}\Cb\bigg(\sum_{t=1}^{T}\Fb_t^\top\Rb^\top\Rb\hat\Hb_1\hat\Hb_1^\top\Rb^\top\Rb\Fb_t\bigg)\Cb^\top+\bigg(\frac{1}{Tp_2}\Cb\sum_{t=1}^{T}\Fb_t^\top\Fb_t\Cb^\top\bigg)\times
	\\
	 &\bigg(\frac{1}{p_1^2}\Rb^\top\Rb\hat\Hb_1(\hat\Rb-\Rb\hat\Hb_1)^\top\Rb+\frac{1}{p_1^2}\Rb^\top(\hat\Rb-\Rb\hat\Hb_1)\hat\Hb_1^\top\Rb^\top\Rb+\frac{1}{p_1^2}\Rb^\top(\hat\Rb-\Rb\hat\Hb_1)(\hat\Rb-\Rb\hat\Hb_1)^\top\Rb\bigg).
	\end{split}
	\]
	\end{small}
Note that $p_1^{-1}\Rb^\top\Rb\rightarrow\Ib_{k_1}$ and $\hat\Hb_1\hat\Hb_1^\top\overset{p}{\rightarrow}\Ib_{k_1}$, then by sufficient condition (\ref{c1}) (a) and Weyl's inequality,  $\lambda_j(\mathcal{\uppercase\expandafter{\romannumeral1}})=\lambda_j(p_2^{-1}\Cb\bSigma_2\Cb^\top)+o_p(1)$ for $j\le k_2$.  The leading $k_2$ eigenvalues of $p_2^{-1}\Cb\bSigma_2\Cb^\top$ are the same as those of $\bSigma_2$ as $p_2\rightarrow\infty$. Hence, $\lambda_j(\mathcal{\uppercase\expandafter{\romannumeral1}})=\lambda_j(\bSigma_2)+o_p(1)$ for $j\le k_2$. Since $\text{rank}(\mathcal{\uppercase\expandafter{\romannumeral1}})\le k_2$, $\lambda_j(\mathcal{\uppercase\expandafter{\romannumeral1}})=0$ for $j>k_2$.
	Secondly, let $k_1=k_2=1$, then
	\[
	\begin{split}
	 \|\mathcal{\uppercase\expandafter{\romannumeral2}}\|\lesssim&\bigg\|\frac{1}{Tp_1p_2}\sum_t\Eb_t^\top\hat\Rb\Fb_t\Cb^\top\bigg\|_F\\
	\lesssim& \bigg\|\frac{1}{Tp_1p_2}\sum_t\Eb_t^\top\Rb\Fb_t\Cb^\top\bigg\|_F+\bigg\|\frac{1}{Tp_1p_2}\sum_t\Eb_t^\top(\hat\Rb-\Rb\hat\Hb_1)\Fb_t\bigg\|_F\|\Cb\|_F
	\\
	 =&O_p\bigg(\frac{1}{\sqrt{Tp_1}}+\sqrt{w_2}\bigg).
	\end{split}
	\]
 Similarly,
 \[
 \|\mathcal{\uppercase\expandafter{\romannumeral3}}\|=O_p\bigg(\frac{1}{\sqrt{Tp_1}}+\sqrt{w_2}\bigg).
 \]
 Lastly, it's not hard that
 \[
 \begin{split}
 \|\mathcal{\uppercase\expandafter{\romannumeral4}}\|\le&\frac{1}{Tp_1^2p_2}\sum_t\|\Eb_t^\top\hat\Rb\|_F^2\lesssim\frac{1}{Tp_1^2p_2}\sum_t\|\Eb_t^\top\Rb\|_F^2+\frac{1}{Tp_1^2p_2}\sum_t\|\Eb_t\|_F^2\|\hat\Rb-\Rb\hat\Hb_1\|_F^2\\
 =&O_p(p_1^{-1}+w_1)=o_p(1).
 \end{split}
 \]
 The lemma then holds by Weyl's theorem.
\end{proof}

\begin{lemma}\label{lema3}
	Under Assumptions A-E and the sufficient condition, as $\min\{T,p_1,p_2\}\rightarrow\infty$,  it holds that
	\[
	 \frac{1}{p_2}\|\mathcal{\uppercase\expandafter{\romannumeral2}}\tilde\Cb\|_F^2=O_p\bigg(\frac{1}{Tp_1}+w_2\bigg),\quad \frac{1}{p_2}\|\mathcal{\uppercase\expandafter{\romannumeral3}}\tilde\Cb\|_F^2=O_p\bigg(\frac{1}{Tp_1}+w_2\bigg).
	\]
\end{lemma}
\begin{proof}
	Firstly, with $k_1=k_2=1$, we have
	\[
	 \mathcal{\uppercase\expandafter{\romannumeral2}}\tilde\Cb=\frac{\hat\Hb_1\hat\Rb^\top\Rb}{p_1}\frac{\Cb^\top\tilde\Cb}{p_2}\frac{1}{Tp_1}\sum_{t=1}^{T}\Eb_t^\top(\hat\Rb\hat\Hb_1^{-1}-\Rb+\Rb)\Fb_t.
	\]
	By Lemma \ref{lema1} (2),
	\[
	 \frac{1}{p_2}\bigg\|\frac{1}{Tp_1}\sum_{t=1}^{T}\Eb_t^\top\Rb\Fb_t\bigg\|_F^2=O_p\bigg(\frac{1}{Tp_1}\bigg).
	\]
	Hence, by sufficient condition (\ref{c1}) (b) and  triangular inequality,
	\[
	 \frac{1}{p_2}\|\mathcal{\uppercase\expandafter{\romannumeral2}}\tilde\Cb\|_F^2=O_p\bigg(\frac{1}{Tp_1}+w_2\bigg).
	\]
	The proof of $\mathcal{\uppercase\expandafter{\romannumeral3}}$ is similar and omitted here.
\end{proof}

\begin{lemma}\label{lema4}
	Under Assumptions A-F and the sufficient condition, as $\min\{T,p_1,p_2\}\rightarrow\infty$, it holds that
	\[
	\frac{1}{p_2}	 \|\mathcal{\uppercase\expandafter{\romannumeral4}\tilde\Cb}\|_F^2\lesssim\frac{1}{Tp_1p_2}+\frac{1}{p_1^2p_2^2}+w_1^2\times\bigg(\frac{1}{p_2^2}+\frac{1}{Tp_2}\bigg)+o_p(1)\times \frac{1}{p_2}\|\tilde\Cb-\Cb\tilde\Hb_2\|_F^2.
	\]
\end{lemma}
\begin{proof}
	Some simple calculations lead to
	\[
	\begin{split}
	 \mathcal{\uppercase\expandafter{\romannumeral4}}\tilde\Cb=&\frac{1}{Tp_1^2p_2}\sum_{t=1}^{T}\Eb_t^\top\hat\Rb\hat\Rb^\top\Eb_t\tilde\Cb\\
	 =&\frac{1}{Tp_1^2p_2}\bigg(\sum_{t=1}^{T}\Eb_t^\top\Rb\hat\Hb_1\hat\Hb_1^\top\Rb^\top\Eb_t\Cb\tilde\Hb_2+\sum_{t=1}^{T}\Eb_t^\top\Rb\hat\Hb_1\hat\Hb_1^\top\Rb^\top\Eb_t(\tilde\Cb-\Cb\tilde\Hb_2)\\
	 &+\sum_{t=1}^{T}\Eb_t^\top(\hat\Rb-\Rb\hat\Hb_1)\hat\Hb_1^\top\Rb^\top\Eb_t\tilde\Cb+\sum_{t=1}^{T}\Eb_t^\top\hat\Rb(\hat\Rb-\Rb\hat\Hb_1)^\top\Eb_t\tilde\Cb\bigg)\\
	 :=&\mathcal{\uppercase\expandafter{\romannumeral5}}+\mathcal{\uppercase\expandafter{\romannumeral6}}+\mathcal{\uppercase\expandafter{\romannumeral7}}+\mathcal{\uppercase\expandafter{\romannumeral8}}.
	\end{split}
	\]
	For $\mathcal{\uppercase\expandafter{\romannumeral5}}$, note that
	\[
	 \mathbb{E}\bigg\|\sum_{t=1}^{T}\Eb_t^\top\Rb\Rb^\top\Eb_t\Cb\bigg\|_F^2\le \|\Rb\|_F^2\sum_{i,j}\mathbb{E}\bigg\|\sum_t\Rb^\top\be_{t,\cdot j}\be_{t,i\cdot}^\top\Cb\bigg\|^2,
	\]
	while for any $i,j$,
\begin{equation}\label{equa3}
\begin{split}
	 &\mathbb{E}\bigg\|\sum_t\Rb^\top\be_{t,\cdot j}\be_{t,i\cdot}^\top\Cb\bigg\|^2\\
	\lesssim& \mathbb{E}\bigg\|\sum_t(\Rb^\top\be_{t,\cdot j}\be_{t,i\cdot}^\top\Cb-\mathbb{E}\Rb^\top\be_{t,\cdot j}\be_{t,i\cdot}^\top\Cb)\bigg\|^2+\bigg\|\sum_t\mathbb{E}\Rb^\top\be_{t,\cdot j}\be_{t,i\cdot}^\top\Cb\bigg\|^2\\
	\lesssim& \sum_{t,i_1,j_1}\sum_{s,i_2,j_2}\bigg|\text{Cov}(e_{t,i_1j}e_{t,ij_1},e_{s,i_2j}e_{s,ij_2})\bigg|+\bigg(\sum_t\sum_{i_1}\sum_{j_1}\bigg|\mathbb{E}e_{t,i_1j}e_{t,ij_1}\bigg|\bigg)^2\\
	\le &O(Tp_1p_2+T^2).\quad\text{ \Big(by Assumptions D.3(2) and D.2(2)\Big)}
\end{split}
\end{equation}
	Hence,
	\[
	 \|\mathcal{\uppercase\expandafter{\romannumeral5}}\|_F^2=O_p\bigg(\frac{1}{Tp_1}+\frac{1}{p_1^2p_2}\bigg).
	\]
	For $\mathcal{\uppercase\expandafter{\romannumeral6}}$, we have
	\[
	 \|\mathcal{\uppercase\expandafter{\romannumeral6}}\|_F^2\lesssim \frac{c}{T^2p_1^4p_2^2}\bigg(\sum_t\|\Eb_t^\top\Rb\|_F^2\bigg)\times \bigg(\sum_t\|\Eb_t^\top\Rb\|_F^2\bigg)\|\tilde\Cb-\Cb\tilde\Hb_2\|_F^2\le\frac{c}{p_1^2}\|\tilde\Cb-\Cb\tilde\Hb_2\|_F^2.
	\]
	Next, for $\mathcal{\uppercase\expandafter{\romannumeral7}}$ we have
	\[
	\begin{split}
	 &\|\mathcal{\uppercase\expandafter{\romannumeral7}}\|_F^2\\
	 \le& \frac{c}{T^2p_1^4p_2^2}\bigg(\bigg\|\sum_t\Eb_t^\top(\hat\Rb-\Rb\hat\Hb_1)\Rb^\top\Eb_t\Cb\bigg\|_F^2+\sum_t\|\Eb_t^\top(\hat\Rb-\Rb\hat\Hb_1)\|_F^2\sum_t\|\Rb^\top\Eb_t\|_F^2\|\tilde\Cb-\Cb\tilde\Hb_2\|_F^2\bigg)\\
	 \lesssim&\frac{1}{T^2p_1^4p_2^2}\bigg\|\sum_t\Eb_t^\top(\hat\Rb-\Rb\hat\Hb_1)\Rb^\top\Eb_t\Cb\bigg\|_F^2+\frac{w_1}{p_1}\times\|\tilde\Cb-\Cb\tilde\Hb_2\|_F^2.
	\end{split}
	\]
	Assume $k_1=k_2=1$, then for the first term
	\[
	\begin{split}
	 \bigg\|\sum_t\Eb_t^\top(\hat\Rb-\Rb\hat\Hb_1)\Rb^\top\Eb_t\Cb\bigg\|_F^2\le&\sum_{i,j}\bigg\|\sum_{t}e_{t,ij}\Rb^\top\Eb_t\Cb\bigg\|_F^2\|\hat\Rb-\Rb\hat\Hb_1\|_F^2,
	\end{split}
	\]
	while for any $i,j$,
	\[
	\begin{split}
	 &\mathbb{E}\bigg\|\sum_{t}e_{t,ij}\Rb^\top\Eb_t\Cb\bigg\|_F^2\\
	 \le&\sum_{t,i_1,j_1}\sum_{s,i_2,j_2}\bigg|\text{Cov}(e_{t,ij}e_{t,i_1j_1},e_{s,ij}e_{s,i_2j_2})\bigg|+\bigg\|\sum_{t}\mathbb{E}e_{t,ij}\Rb^\top\Eb_t\Cb\bigg\|_F^2\\
	\lesssim& Tp_1p_2+\bigg(\sum_t\sum_{i_1,j_1}\bigg|\mathbb{E}e_{t,ij}e_{t,i_1j_1}\bigg|\bigg)^2\quad \Big(\text{by Assumption D.3(2)}\Big)\\
	=&O(Tp_1p_2+T^2).\quad \Big(\text{by Assumption D.2(1)}\Big)
	\end{split}
	\]
	Hence,
	\[
	 \|\mathcal{\uppercase\expandafter{\romannumeral7}}\|_F^2=w_1\times \bigg(\frac{1}{Tp_1}+\frac{1}{p_1^2p_2}\bigg)+\frac{w_1}{p_1}\times\|\tilde\Cb-\Cb\tilde\Hb_2\|_F^2.
	\]
	For  $\mathcal{\uppercase\expandafter{\romannumeral8}}$,
	\[
	\begin{split}
	 \|\mathcal{\uppercase\expandafter{\romannumeral8}}\|_F^2\le& \frac{c}{T^2p_1^4p_2^2}\bigg(\sum_t\|\Eb_t^\top\hat\Rb\|_F^2\bigg)\times \sum_t\|(\hat\Rb-\Rb\hat\Hb_1)^\top\Eb_t\|_F^2\|\tilde\Cb-\Cb\tilde\Hb_2\|_F^2\\
	 &+\frac{c}{T^2p_1^4p_2^2}\bigg\|\sum_t\Eb_t^\top\hat\Rb(\hat\Rb-\Rb\hat\Hb_1)^\top \Eb_t\Cb\bigg\|_F^2\\
	 \lesssim&\bigg(\frac{1}{p_1}+w_1\bigg)\times w_1\times \|\tilde\Cb-\Cb\tilde\Hb_2\|_F^2+\frac{1}{T^2p_1^4p_2^2}\bigg\|\sum_t\Eb_t^\top\hat\Rb(\hat\Rb-\Rb\hat\Hb_1)^\top \Eb_t\Cb\bigg\|_F^2.
	\end{split}
	\]
	Again let's assume $k_1=k_2=1$, then
	\[
	\begin{split}
	 &\frac{1}{T^2p_1^4p_2^2}\bigg\|\sum_t\Eb_t^\top\hat\Rb(\hat\Rb-\Rb\hat\Hb_1)^\top \Eb_t\Cb\bigg\|_F^2\le\frac{\|\hat\Rb-\Rb\hat\Hb_1\|_F^2}{T^2p_1^4p_2^2}\sum_{i=1}^{p_1}\sum_{j=1}^{p_2}\bigg\|\sum_{t}\Cb^\top\be_{t,i\cdot}\be_{t,\cdot j}^\top\hat\Rb\bigg\|_F^2\\
	 \lesssim&\frac{w_1}{T^2p_1^3p_2^2}\sum_{i=1}^{p_1}\sum_{j=1}^{p_2}\bigg\|\sum_{t}\Cb^\top\be_{t,i\cdot}\be_{t,\cdot j}^\top\Rb\bigg\|_F^2+\frac{w_1^2}{T^2p_1^2p_2^2}\sum_{i=1}^{p_1}\sum_{j=1}^{p_2}\bigg\|\sum_{t}\Cb^\top\be_{t,i\cdot}\be_{t,\cdot j}^\top\bigg\|_F^2.
	\end{split}
	\]
	By equation (\ref{equa3}) we have for any $i,j$, $\mathbb{E}\|\sum_{t}\Cb^\top\be_{t,i\cdot}\be_{t,\cdot j}^\top\Rb\|_F^2= O(Tp_1p_2+T^2)$, while
	\[
	\begin{split}
	 &\mathbb{E}\bigg\|\sum_{t}\Cb^\top\be_{t,i\cdot}\be_{t,\cdot j}^\top\bigg\|_F^2\le\sum_{i_1=1}^{p_1}\mathbb{E}
	 \bigg(\sum_t\sum_{j_1}C_{j_1}e_{t,ij_1}e_{t,i_1j}\bigg)^2\\
	 \le&\sum_{i_1}\sum_{t,j_1}\sum_{s,j_2}\bigg|\text{Cov}(e_{t,i_1j}e_{t,ij_1},e_{s,i_1j}e_{s,ij_2})\bigg|+\sum_{i_1}
	 \bigg(\sum_t\sum_{j_1}\bigg|\mathbb{E}e_{t,ij_1}e_{t,i_1j}\bigg|\bigg)^2\\
	\le &O(Tp_1p_2+T^2p_1).\quad \Big(\text{by Assumptions D.3(2) and D.2(2)}\Big)
	\end{split}
	\]
	Consequently,
	\[
	 \|\mathcal{\uppercase\expandafter{\romannumeral8}}\|_F^2\lesssim w_1\times\bigg(\frac{1}{Tp_1}+\frac{1}{p_1^2p_2}\bigg)+w_1^2\times\bigg(\frac{1}{p_2}+\frac{1}{T}\bigg)+o_p(1)\times \|\tilde\Cb-\Cb\tilde\Hb_2\|_F^2,
	\]
	Combine the above results so that
	\[
	\frac{1}{p_2}	 \|\mathcal{\uppercase\expandafter{\romannumeral4}\tilde\Cb}\|_F^2\lesssim\frac{1}{Tp_1p_2}+\frac{1}{p_1^2p_2^2}+w_1^2\times\bigg(\frac{1}{p_2^2}+\frac{1}{Tp_2}\bigg)+o_p(1)\times \frac{1}{p_2}\|\tilde\Cb-\Cb\tilde\Hb_2\|_F^2.
	\]
	and the lemma holds.
\end{proof}

\section{Proof of Theorem \ref{thm3}: verifying the sufficient condition for initial estimator}\label{secb}
\begin{proof}
We only need to prove the results for $\hat\Rb$ because  $\hat\Cb$ is estimated by a parallel procedure.   Expand $\Xb_t=\Rb\Fb_t\Cb^\top$ in $\hat\Mb_1$, then
\begin{equation}\label{equb1}
\begin{split}
\hat\Mb_1
=&\frac{1}{Tp_1p_2}\sum_{t=1}^{T}(\Rb\Fb_t\Cb^\top+\Eb_t)(\Rb\Fb_t\Cb^\top+\Eb_t)^\top\\
=&\frac{1}{Tp_1p_2}\bigg(\sum_{t=1}^{T}\Rb\Fb_t\Cb^\top\Cb\Fb_t^\top\Rb^\top+\sum_{t=1}^{T}\Rb\Fb_t\Cb^\top\Eb_t^\top+\sum_{t=1}^{T}\Eb_t\Cb\Fb_t^\top\Rb^\top+\sum_{t=1}^{T}\Eb_t\Eb_t^\top\bigg)\\
:=&\mathcal{\uppercase\expandafter{\romannumeral1}}+\mathcal{\uppercase\expandafter{\romannumeral2}}+\mathcal{\uppercase\expandafter{\romannumeral3}}+\mathcal{\uppercase\expandafter{\romannumeral4}}.
\end{split}
\end{equation}

Define $\hat\bLambda_1$ as the $k_1\times k_1$ diagonal matrix with $\hat\Lambda_{1,jj}=\lambda_j(\hat\Mb_1)$, then
\[
\hat\Rb\hat\bLambda_1=\hat\Mb_1\hat\Rb.
\]
Further define $\hat\Hb_1=(Tp_1p_2)^{-1}\sum_{t=1}^{T}\Fb_t\Cb^\top\Cb\Fb_t^\top\Rb^\top\hat\Rb\hat\bLambda_1^{-1}$, then
\begin{equation}\label{equb2}
\hat \Rb-\Rb\hat\Hb_1=(\mathcal{\uppercase\expandafter{\romannumeral2}}+\mathcal{\uppercase\expandafter{\romannumeral3}}+\mathcal{\uppercase\expandafter{\romannumeral4}})\hat\Rb\hat\bLambda_1^{-1}.
\end{equation}
We will show that the diagonal entries of $\hat\bLambda_1$ converge to some positive constants in Lemma \ref{lemb1}. Hence,  by the fact that $T^{-1}\sum_t\Fb_t\Fb_t^\top\overset{p}{\rightarrow}\bSigma_1$, $p_2^{-1}\Cb^\top\Cb\rightarrow\Ib_{k_2}$, $\|\Rb\|_F^2\asymp\|\hat\Rb\|_F^2\asymp p_1$,  we have $\|\hat\Hb_1\|=O_p(1)$. Next, detailed calculations of  $\mathcal{\uppercase\expandafter{\romannumeral2}},\mathcal{\uppercase\expandafter{\romannumeral3}}$ and $\mathcal{\uppercase\expandafter{\romannumeral4}}$ in Lemma \ref{lemb2} lead to
\[
		\begin{split}
\frac{1}{p_1}\|\mathcal{\uppercase\expandafter{\romannumeral2}}\hat\Rb\|_F^2=& O_p\bigg(\frac{1}{Tp_1p_2}\bigg)+o_p(1)\times \frac{1}{p_1}\|\hat\Rb-\Rb\hat\Hb_1\|_F^2,\\
\frac{1}{p_1}\|\mathcal{\uppercase\expandafter{\romannumeral3}}\hat\Rb\|_F^2=&O_p\bigg(\frac{1}{Tp_2}\bigg),\\
\frac{1}{p_1}\|\mathcal{\uppercase\expandafter{\romannumeral4}}\hat\Rb\|_F^2=& O_p\bigg(\frac{1}{p_1^2}+\frac{1}{Tp_1p_2}\bigg)+o_p(1)\times \frac{1}{p_1}\|\hat\Rb-\Rb\hat\Hb_1\|_F^2.		
\end{split}
\]
 Hence, the convergence rate of $w_1$  follows. The rate of $w_2$ is verified separately in Lemma \ref{lemb3} using similar technique.

To complete the proof, it remains to show that  $\hat\Hb_1^\top\hat\Hb_1\overset{p}{\rightarrow}\Ib_{k_1}$. By Cauchy-Schwartz inequality,
\[
\bigg\|\frac{1}{p_1}\Rb^\top(\hat\Rb-\Rb\hat\Hb_1)\bigg\|_F^2\le\frac{\|\Rb\|_F^2}{p_1}\frac{\|\hat\Rb-\Rb\hat\Hb_1\|_F^2}{p_1}=o_p(1)\quad\text{and}\quad  \bigg\|\frac{1}{p_1}\hat\Rb^\top(\hat\Rb-\Rb\hat\Hb_1)\bigg\|_F^2=o_p(1).
\]
Note that $p_1^{-1}\hat\Rb^\top\hat\Rb=\Ib_r$ and $p_1^{-1}\Rb^\top\Rb\rightarrow\Ib_{k_1}$, then
\[
\Ib_{k_1}=\frac{1}{p_1}\hat\Rb^\top\Rb\hat\Hb_1+{\bf o_p(1)}=\hat\Hb_1^\top\hat\Hb_1+{\bf o_p(1)},
\]
which concludes Theorem \ref{thm3}.
\end{proof}

Now we show the detailed proofs of the mentioned lemmas.

	\begin{lemma}\label{lemb1}
		Under Assumptions A, B and C, as $\min\{T,p_1,p_2\}\rightarrow\infty$,
		\[
		\lambda_j(\hat\Mb_1)=\left\{
		\begin{aligned}
		&\lambda_j(\bSigma_1)+o_p(1),& j\le k_1,\\
		 &O_p\bigg(\frac{1}{\sqrt{Tp_2}}+\frac{1}{p_1}\bigg),&j>k_1.\\
		\end{aligned}
		\right..
		\]
	\end{lemma}	
	\begin{proof}
		Recall that
		\[
		 \hat\Mb_1=\mathcal{\uppercase\expandafter{\romannumeral1}}+\mathcal{\uppercase\expandafter{\romannumeral2}}+\mathcal{\uppercase\expandafter{\romannumeral3}}+\mathcal{\uppercase\expandafter{\romannumeral4}}.
		\]
		We will study the spectral norms of these four terms and show that $\mathcal{\uppercase\expandafter{\romannumeral1}}$ is the main term. Firstly, by Assumptions B and C, we have
		\[
		 \frac{1}{Tp_2}\sum_{t}\Fb_t\Cb^\top\Cb\Fb_t^\top\overset{p}{\rightarrow}\bSigma_1,
		\]
		while the leading $k_1$ eigenvalues of $p_1^{-1}\Rb\bSigma_1\Rb^\top$ are asymptotically equal to those of $\bSigma_1$. Hence,  $\lambda_j(\mathcal{\uppercase\expandafter{\romannumeral1}})=\lambda_j(\bSigma_1)+o_p(1)$ for $j\le k_1$ while $\lambda_j(\mathcal{\uppercase\expandafter{\romannumeral1}})=0$ for $j> k_1$ because $\text{rank}(\mathcal{\uppercase\expandafter{\romannumeral1}})\le k_1$. Secondly by Cauchy-Schwartz inequality and Lemma \ref{lema1} (2),
		\[
		 \|\mathcal{\uppercase\expandafter{\romannumeral2}}\|\le\frac{1}{\sqrt{Tp_2}}\bigg\|\frac{1}{\sqrt{Tp_1p_2}}\sum_{t=1}^{T}\Fb_t\Cb^\top\Eb_t^\top\bigg\|\le O\bigg(\frac{1}{\sqrt{Tp_2}}\bigg).
		\]
		 Similarly, $\|\mathcal{\uppercase\expandafter{\romannumeral3}}\|\le O_p(1/\sqrt{Tp_2})$. Lastly for $\mathcal{\uppercase\expandafter{\romannumeral4}}$, denote $\Ub_{\Eb}=(Tp_1p_2)^{-1}\sum_t\mathbb{E}(\Eb_t\Eb_t^\top)$, then
		 \[
		 \begin{split}
		 &\mathbb{E}\bigg\|\mathcal{\uppercase\expandafter{\romannumeral4}}-\Ub_{\Eb}\bigg\|_F^2=\frac{1}{T^2p_1^2p_2^2}\sum_{i_1,i_2}\mathbb{E}\bigg(\sum_{t,j}e_{t,i_1j}e_{t,i_2j}-\mathbb{E}e_{t,i_1j}e_{t,i_2j}\bigg)^2\\
		 \le &\frac{1}{T^2p_1^2p_2^2}\sum_{i_1,i_2}\sum_{t,j_1}\sum_{s,j_2}\bigg|\text{Cov}(e_{t,i_1j_1}e_{t,i_2j_1},e_{s,i_1j_2}e_{s,i_2j_2})\bigg|\\
		 \le& \frac{c}{Tp_2}.\quad \text{\Big(by Assumption D.3(1)\Big)}
		 \end{split}
		 \]
		 Meanwhile, by Assumption D.2(1), for any $i\le p_1$,
		 \[
		 \sum_{i_1}\bigg|\sum_t\sum_j\mathbb{E}(e_{t,ij}e_{t,i_1j})\bigg|\le \sum_{t,j}\sum_{i_1}\bigg|\mathbb{E}(e_{t,ij}e_{t,i_1j})\bigg|\le cTp_2.
		 \]
		 Hence, $\|\Ub_{\Eb}\|_1=\|\Ub_{\Eb}\|_{\infty}\le O(p_1^{-1})$, which further implies $\|\Ub_{\Eb}\|\le O(p_1^{-1})$. Therefore,
		 \[
		 \|\mathcal{\uppercase\expandafter{\romannumeral4}}\|=O_p\bigg(\frac{1}{\sqrt{Tp_2}}+\frac{1}{p_1}\bigg).
		 \]
	The lemma holds with Weyl's inequality.
	\end{proof}

	\begin{lemma}\label{lemb2}
		Under Assumptions A, B and C, as $\min\{T,p_1,p_2\}\rightarrow\infty$, it holds that
		\[
		\begin{split}
			 \frac{1}{p_1}\|\mathcal{\uppercase\expandafter{\romannumeral2}}\hat\Rb\|_F^2=& O_p\bigg(\frac{1}{Tp_1p_2}\bigg)+o_p(1)\times \frac{1}{p_1}\|\hat\Rb-\Rb\hat\Hb_1\|_F^2,\\
			 \frac{1}{p_1}\|\mathcal{\uppercase\expandafter{\romannumeral3}}\hat\Rb\|_F^2=&O_p\bigg(\frac{1}{Tp_2}\bigg),\\
			 \frac{1}{p_1}\|\mathcal{\uppercase\expandafter{\romannumeral4}}\hat\Rb\|_F^2=& O_p\bigg(\frac{1}{p_1^2}+\frac{1}{Tp_1p_2}\bigg)+o_p(1)\times \frac{1}{p_1}\|\hat\Rb-\Rb\hat\Hb_1\|_F^2,			
		\end{split}
		\]
		where $\mathcal{\uppercase\expandafter{\romannumeral2}},\mathcal{\uppercase\expandafter{\romannumeral3}}$ and $\mathcal{\uppercase\expandafter{\romannumeral4}}$ are defined in equation (\ref{equb1})
	\end{lemma}	
	\begin{proof}
		Firstly, by equation (\ref{equb1}),
		\begin{equation}\label{equb3}
		\begin{split}
		 &\frac{1}{p_1}\|\mathcal{\uppercase\expandafter{\romannumeral2}}\hat\Rb\|_F^2\le\frac{\|\Rb\|_F^2}{p_1}\bigg\|\frac{1}{Tp_1p_2}\sum_{t=1}^{T}\Fb_t\Cb^\top\Eb_t^\top\hat\Rb\bigg\|_F^2\\
		\lesssim& \bigg\|\frac{1}{Tp_1p_2}\sum_{t=1}^{T}\Fb_t\Cb^\top\Eb_t^\top\Rb\bigg\|_F^2\|\hat\Hb_1\|_F^2+\bigg\|\frac{1}{Tp_1p_2}\sum_{t=1}^{T}\Fb_t\Cb^\top\Eb_t^\top\bigg\|_F^2\|\hat\Rb-\Rb\hat\Hb_1\|_F^2.
		\end{split}
		\end{equation}
		It's easy that $\|\hat\Hb_1\|_F^2=O_p(1)$, further by Lemma \ref{lema1} (2), we have
		\[
		 \frac{1}{p_1}\|\mathcal{\uppercase\expandafter{\romannumeral2}}\hat\Rb\|_F^2= O_p\bigg(\frac{1}{Tp_1p_2}\bigg)+o_p(1)\times \frac{1}{p_1}\|\hat\Rb-\Rb\hat\Hb_1\|_F^2.
		\]
		
	Secondly, it's not hard that
	\[
	 \frac{1}{p_1}\|\mathcal{\uppercase\expandafter{\romannumeral3}}\hat\Rb\|_F^2\le \frac{1}{p_1}\bigg\|\frac{1}{Tp_1p_2}\sum_{t=1}^{T}\Eb_t\Cb\Fb_t^\top\bigg\|_F^2\|\hat\Rb\|_F^2\|\Rb\|_F^2=O_p\bigg(\frac{1}{Tp_2}\bigg).
	\]
	
	Thirdly, use Lemma \ref{lema1} (3) so that
			 \begin{equation}\label{equb4}
	\begin{split}
	 \frac{1}{p_1}\|\mathcal{\uppercase\expandafter{\romannumeral4}}\hat\Rb\|_F^2=&\frac{1}{p_1}\sum_{i=1}^{p_1}\bigg\|\frac{1}{Tp_1p_2}\sum_{t=1}^{T}\Rb^\top\Eb_t\be_{t,i\cdot}+\frac{1}{Tp_1p_2}\sum_{t=1}^{T}(\hat\Rb-\Rb\hat\Hb_1)^\top\Eb_t\be_{t,i\cdot}\bigg\|^2\\
	=& O_p\bigg(\frac{1}{p_1^2}+\frac{1}{Tp_1p_2}\bigg)+o_p(1)\times \frac{1}{p_1}\|\hat\Rb-\Rb\hat\Hb_1\|_F^2.
	\end{split}
	\end{equation}
	The lemma follows.
	\end{proof}
	
	\begin{lemma}\label{lemb3}
	Under Assumptions A, B and C, as $\min\{T,p_1,p_2\rightarrow\infty\}$, we have
	\[
\frac{1}{p_2}\bigg\|\frac{1}{Tp_1}\sum_{s=1}^{T}\Eb_s^\top(\hat\Rb-\Rb\hat\Hb_1)\Fb_s\bigg\|^2=O_p\bigg(\frac{1}{Tp_1^2}+\frac{1}{T^2p_2^2}\bigg).
	\]
\end{lemma}
\begin{proof}
By equation (\ref{equb2}),
\[
\begin{split}
\sum_{s=1}^{T}\Eb_s^\top(\hat\Rb-\Rb\hat\Hb_1)\Fb_s=\sum_{s=1}^{T}\Eb_s^\top(\mathcal{\uppercase\expandafter{\romannumeral2}}+\mathcal{\uppercase\expandafter{\romannumeral3}}+\mathcal{\uppercase\expandafter{\romannumeral4}})\hat\Rb\hat\bLambda_1^{-1}\Fb_s.
\end{split}
\]
Assume $k_1=k_2=1$, then $\hat\bLambda_1^{-1}$ can be ignored. Firstly,
\[
\begin{split}
\bigg\| \sum_{s=1}^T\Eb_s^\top\mathcal{\uppercase\expandafter{\romannumeral2}}\hat\Rb\Fb_s\bigg\|_F^2=&\bigg\|\frac{1}{Tp_1p_2} \sum_{s=1}^T\Eb_s^\top\bigg(\Rb\sum_{t=1}^{T}\Fb_t\Cb^\top\Eb_t^\top\bigg)\hat\Rb\Fb_s\bigg\|_F^2\\
\lesssim &\bigg\|\frac{1}{Tp_1p_2}\sum_{t=1}^{T}\Fb_t\Cb^\top\Eb_t^\top\hat\Rb\bigg\|_F^2\times\bigg\|\sum_{s=1}^T\Eb_s^\top\Rb\Fb_s\bigg\|_F^2.
\end{split}
\]
Using equation (\ref{equb3}) and Lemma \ref{lema1} (2), it's easy to verify
\[
\frac{1}{p_2}\bigg\|\frac{1}{Tp_1}\sum_{s=1}^T\Eb_s^\top\mathcal{\uppercase\expandafter{\romannumeral2}}\hat\Rb\Fb_s\bigg\|_F^2=O_p\bigg(\frac{1}{T^2p_1^2p_2}+\frac{1}{T^3p_1p_2^2}\bigg).
\]

Secondly, if $k_1=k_2=1$,
\[
\begin{split}
&\mathbb{E}\bigg\|\sum_{s=1}^T\Eb_s^\top\mathcal{\uppercase\expandafter{\romannumeral3}}\hat\Rb\Fb_s\bigg\|_F^2=\mathbb{E}\bigg\|\frac{1}{Tp_1p_2} \sum_{s=1}^T\Eb_s^\top\bigg(\sum_{t=1}^{T}\Eb_t\Cb\Fb_t^\top\Rb^\top\bigg)\hat\Rb\Fb_s\bigg\|_F^2\\
=&  \frac{1}{T^2p_2^2}\sum_{j=1}^{p_2}\mathbb{E}\bigg(\sum_{t=1}^{T}\sum_{s=1}^T\sum_{i=1}^{p_1}\sum_{j_1=1}^{p_2}C_{j_1}F_tF_se_{s,ij}e_{t,ij_1}\bigg)^2\\
\lesssim&\frac{1}{p_2^2}\sum_{j=1}^{p_2}\mathbb{E}\bigg(\sum_{i=1}^{p_1}\sum_{j_1=1}^{p_2}C_{j_1}(\zeta_{ij}\zeta_{ij_1}-\mathbb{E}\zeta_{ij}\zeta_{ij_1})\bigg)^2+\frac{1}{p_2^2}\sum_{j=1}^{p_2}\bigg(\sum_{i=1}^{p_1}\sum_{j_1=1}^{p_2}\Big|\mathbb{E}\zeta_{ij}\zeta_{ij_1}\Big|\bigg)^2\\
=&O\bigg(p_1+\frac{p_1^2}{p_2}\bigg). \quad \Big( \text{by Assumption E.2}\Big)
\end{split}
\]
Therefore,
\[
\frac{1}{p_2}\bigg\|\frac{1}{Tp_1}\sum_{s=1}^T\Eb_s^\top\mathcal{\uppercase\expandafter{\romannumeral3}}\hat\Rb\Fb_s\bigg\|_F^2=O_p\bigg(\frac{1}{T^2p_1p_2}+\frac{1}{T^2p_2^2}\bigg).
\]

	Thirdly, by Assumption E.1, Lemma \ref{lema1} (3) and the decomposition in  equation (\ref{equb4}),
		\[
\begin{split}
\frac{1}{p_2}\bigg\|\frac{1}{Tp_1}\sum_{s=1}^T\Eb_s^\top\mathcal{\uppercase\expandafter{\romannumeral4}}\hat\Rb\Fb_s\bigg\|_F^2\le\frac{1}{T^2p_1^2p_2}\bigg\|\sum_s\Eb_s^\top\otimes\Fb_s\bigg\|_F^2\times \bigg\|\mathcal{\uppercase\expandafter{\romannumeral4}}\hat\Rb\bigg\|_F^2=O_p\bigg(\frac{1}{Tp_1^2}+\frac{1}{T^3p_2^2}\bigg).
\end{split}
\]

As a result,
\[
\frac{1}{p_2}\bigg\|\frac{1}{Tp_1}\sum_{s=1}^{T}\Eb_s^\top(\hat\Rb-\Rb\hat\Hb_1)\Fb_s\bigg\|^2=O_p\bigg(\frac{1}{Tp_1^2}+\frac{1}{T^2p_2^2}\bigg),
\]
and the lemma holds.
\end{proof}

\section{Proof of Theorem \ref{thm4}: asymptotic distribution for initial estimator}\label{secc}
\begin{proof}
		Based on equation (\ref{equb2}),
	\[
	 \hat\bR_i-\hat\Hb_1^\top\bR_i=\frac{1}{Tp_1p_2}\sum_{t=1}^{T}\hat\bLambda_1^{-1}(\hat\Rb^\top\Eb_t\Cb\Fb_t^\top\bR_i+\hat\Rb^\top\Rb\Fb_t\Cb^\top\be_{t,i\cdot}+\hat\Rb^\top\Eb_t\be_{t,i\cdot}).
	\]
	Use the decomposition in equation (\ref{equb3}), then
	\[
	 \bigg\|\frac{1}{Tp_1p_2}\sum_{t=1}^{T}\hat\bLambda_1^{-1}\hat\Rb^\top\Eb_t\Cb\Fb_t^\top\bR_i\bigg\|\lesssim  \bigg\|\frac{1}{Tp_1p_2}\sum_{t=1}^{T}\hat\Rb^\top\Eb_t\Cb\Fb_t^\top\bigg\|=O_p\bigg(\frac{1}{\sqrt{Tp_1p_2}}+\frac{1}{Tp_2}\bigg).
	\]
	By Assumption D,
	\[
	 \bigg\|\frac{1}{Tp_1p_2}\sum_{t=1}^{T}\hat\bLambda_1^{-1}\hat\Rb^\top\Rb\Fb_t\Cb^\top\be_{t,i\cdot}\bigg\|\le\|\hat\bLambda_1^{-1}\|\bigg\|\frac{1}{p_1}\hat\Rb^\top\Rb\bigg\|_F \bigg\|\frac{1}{Tp_2}\sum_{t=1}^{T}\Fb_t\Cb^\top\be_{t,i\cdot}\bigg\|=O_p\bigg(\frac{1}{\sqrt{Tp_2}}\bigg).
	\]	
	On the other hand, by Lemma \ref{lema1} (3),  we have
	\[
	\begin{split}
	\bigg\| \frac{1}{Tp_1p_2}\sum_{t=1}^{T}\hat\bLambda_1^{-1}\hat\Rb^\top\Eb_t\be_{t,i\cdot}\bigg\|\lesssim&\bigg\| \frac{1}{Tp_1p_2}\sum_{t=1}^{T}(\hat\Rb-\Rb\hat\Hb_1)^\top\Eb_t\be_{t,i\cdot}\bigg\|+\bigg\| \frac{1}{Tp_1p_2}\sum_{t=1}^{T}\Rb^\top\Eb_t\be_{t,i\cdot}\bigg\|\\
	 =&O_p\bigg(\frac{1}{p_1}+\frac{1}{Tp_2}\bigg).
	\end{split}
	\]
	Therefore,  when  $\min\{T,p_1,p_2\}\rightarrow\infty$ and $Tp_2=o(p_1^2)$, we have
	\[
	 \sqrt{Tp_2}(\hat\bR_i-\hat\Hb_1\bR_i)=\hat\bLambda_1^{-1}\frac{\hat\Rb^\top\Rb}{p_1} \frac{1}{\sqrt{Tp_2}}\sum_{t=1}^{T}\Fb\Cb^\top\be_{t,i\cdot}+{\bf o_p(1)}.
	\]
	 By Lemma \ref{lemb1}, $\hat\bLambda_1\stackrel{p}{\rightarrow}\bLambda_1$.  By Theorem \ref{thm1}, $\hat\Hb_1=p_1^{-1}\Rb^\top\hat\Rb+{\bf o_p(1)}$.  Assume that $\bSigma_1$ has the spectral decomposition $\bSigma_1=\bGamma_1\bLambda_1\bGamma_1^\top$, then by the definition of $\hat\Hb_1$ we have
	\[
	 \hat\Hb_1=\bGamma_1\bLambda_1\bGamma_1^\top\hat\Hb_1\bLambda_1^{-1}+{\bf o_p(1)}.
	\]
	Rearrange the last equation so that
	\[
	 \bGamma_1^\top\hat\Hb_1\bLambda_1=\bLambda_1\bGamma_1^\top\hat\Hb_1+{\bf o_p(1)}.
	\]
	Because $\bLambda_1$ is diagonal with distinct entries, $\bGamma_1^\top\hat\Hb_1$ must be asymptotically diagonal. Further by
	\[
	 \hat\Hb_1^\top\hat\Hb_1=(\bGamma_1^\top\hat\Hb_1)^\top\bGamma_1^\top\hat\Hb_1=\Ib_{k_1}+{\bf o_p(1)},
	\]
	the diagonal entries of $\bGamma_1^\top\hat\Hb_1$ must be asymptotically $1$ or $-1$. Without loss of generality we can make $\bGamma_1^\top\hat\Hb_1=\Ib_{k_1}+{\bf o_p(1)}$ by choosing the column signs of $\hat\Rb$. Therefore
	\[
	 \frac{\hat\Rb^\top\Rb}{p_1}=\hat\Hb_1^\top+{\bf o_p(1)}=\bGamma_1^\top+{\bf o_p(1)}.
	\]
	Consequently, by the Slutsky's theorem, when  $\min\{T,p_1,p_2\}\rightarrow\infty$ and $Tp_2=o(p_1^2)$,
	\[
	 \sqrt{Tp_2}(\hat\bR_i-\hat\Hb_1^\top\bR_i)\stackrel{d}{\rightarrow}\mathcal{N}({\bf 0},\bLambda_1^{-1}\bGamma_1^\top\Vb_{1i}\bGamma_1\bLambda_1^{-1}).
	\]
	Parallel procedures yield the results for the back loadings $\bC_j$, which concludes the theorem.
\end{proof}

\section{Proof of Theorem \ref{thm2}: asymptotic distribution for projected estimator}\label{secd}
\begin{proof}
	By equation (\ref{equa2}), for any $j\le p_2$
	\[
	\begin{split}
	 \tilde\bC_j-\tilde\Hb_2^\top\bC_j=&\frac{1}{Tp_1^2p_2}\tilde\bLambda_2^{-1}\tilde\Cb^\top\sum_{t=1}^{T}\bigg(\Cb\Fb_t^\top\Rb^\top\hat\Rb\hat\Rb^\top\be_{t,\cdot j}+\Eb_t^\top\hat\Rb\hat\Rb^\top\Rb\Fb_t\bC_j+\Eb_t^\top\hat\Rb\hat\Rb^\top\be_{t,\cdot j}\bigg)\\
	 :=&\mathcal{\uppercase\expandafter{\romannumeral1}}+\mathcal{\uppercase\expandafter{\romannumeral2}}+\mathcal{\uppercase\expandafter{\romannumeral3}}.
	\end{split}
	\]
	Firstly,
	\[
	\begin{split}
	 \mathcal{\uppercase\expandafter{\romannumeral1}}=&\tilde\bLambda_2^{-1}\frac{\tilde\Cb^\top\Cb}{p_2}\frac{1}{Tp_1}\sum_{t=1}^{T}\Fb_t^\top\frac{\Rb^\top\hat\Rb}{p_1}(\hat\Rb-\Rb\hat\Hb_1+\Rb\hat\Hb_1)^\top\be_{t,\cdot j}\\
	 =&\tilde\bLambda_2^{-1}\frac{\tilde\Cb^\top\Cb}{p_2}\frac{1}{Tp_1}\sum_{t=1}^{T}\Fb_t^\top[\Ib_{k_1}+{\bf o_p(1)}]\Rb^\top\be_{t,\cdot j}+\tilde\bLambda_2^{-1}\frac{\tilde\Cb^\top\Cb}{p_2}\frac{\Rb^\top\hat\Rb}{p_1}\frac{1}{Tp_1}\sum_{t=1}^{T}\Fb_t^\top(\hat\Rb-\Rb\hat\Hb_1)^\top\be_{t,\cdot j}.
	\end{split}
	\]
	Similarly to the proof of Lemma \ref{lemb3},
	\[
	 \frac{1}{Tp_1}\sum_{t=1}^{T}\Fb_t^\top(\hat\Rb-\Rb\hat\Hb_1)^\top\be_{t,\cdot j}={\bf{O_p}}\bigg(\frac{1}{p_1\sqrt{T}}+\frac{1}{Tp_2}\bigg).
	\]
	Hence, for the first term we have
	\[
		 \mathcal{\uppercase\expandafter{\romannumeral1}}=\tilde\bLambda_2^{-1}\frac{\tilde\Cb^\top\Cb}{p_2}\frac{1}{Tp_1}\sum_{t=1}^{T}\Fb_t^\top\Rb^\top\be_{t,\cdot j}+{\bf o_p}\bigg(\frac{1}{\sqrt{Tp_1}}\bigg)+{\bf O_p}\bigg(\frac{1}{Tp_2}\bigg).
	\]
	Secondly, assume $k_1=k_2=1$, then a similar technique as in equation (\ref{lemb3}) leads to
	\[
	\begin{split}
	 \|\mathcal{\uppercase\expandafter{\romannumeral2}}\|\lesssim&\bigg\|\frac{1}{Tp_1p_2}\sum_{t=1}^{T}\tilde\Cb^\top\Eb_t^\top\hat\Rb\Fb_t\bigg\|
	 \lesssim\frac{1}{\sqrt{Tp_1p_2}}+\frac{1}{Tp_2}+\|\tilde\Cb-\Cb\tilde\Hb_2\|_F\bigg\|\frac{1}{Tp_1p_2}\sum_{t=1}^{T}\Eb_t^\top\hat\Rb\Fb_t\bigg\|.
	\end{split}
	\]
	The proof of Lemma \ref{lema3} shows that
	\[
	 \bigg\|\frac{1}{Tp_1p_2}\sum_{t=1}^{T}\Eb_t^\top\hat\Rb\Fb_t\bigg\|=O_p\bigg(\frac{1}{Tp_1p_2}\bigg).
	\]
	Hence, combined with Corollary \ref{cor1},  we have
	\[
		 \|\mathcal{\uppercase\expandafter{\romannumeral2}}\|=o_p\bigg(\frac{1}{\sqrt{Tp_1}}\bigg)+O_p\bigg(\frac{1}{Tp_2}\bigg).
	\]
	Thirdly, for $\mathcal{\uppercase\expandafter{\romannumeral3}}$, we can use the same techniques as bounding $p_2^{-1}\|\mathcal{\uppercase\expandafter{\romannumeral4}}\tilde\Cb\|_F^2$ in Lemma \ref{lema4} to show that
	\[
	 \|\mathcal{\uppercase\expandafter{\romannumeral3}}\|=o_p\bigg(\frac{1}{\sqrt{Tp_1}}\bigg)+O_p\bigg(\frac{1}{p_1p_2}+\frac{1}{Tp_2}\bigg).
	\]
	Therefore, when $\min\{T,p_1,p_2\}\rightarrow\infty$ and $Tp_1=o(\min\{p_1^2p_2^2,T^2p_2^2\})$, we have
	\[
	 \sqrt{Tp_1}(\tilde\bC_j-\tilde\Hb_2^\top\bC_j)=\tilde\bLambda_2^{-1}\frac{\tilde\Cb^\top\Cb}{p_2}\frac{1}{\sqrt{Tp_1}}\sum_{t=1}^{T}\Fb_t^\top\Rb^\top\be_{t,\cdot j}+{\bf o_p(1)}.
	\]
 By Lemma \ref{lema2}, $\tilde\bLambda_2\stackrel{p}{\rightarrow}\bLambda_2$. 	 Assume $\bSigma_2$ has the spectral decomposition $\bSigma_2=\bGamma_2\bLambda_2\bGamma_2^\top$, then  by the definition of $\tilde\Hb_2$, directly we have
	\[
	 \bGamma_2^\top\tilde\Hb_2\bLambda_2=\bLambda_2\bGamma_2^\top\tilde\Hb_2+{\bf o_p(1)}.
	\]
	Similarly to the treatment of $\hat\Hb_1$ in Theorem \ref{thm4},  $\bGamma_2^\top\tilde\Hb_2=\Ib_{k_2}+{\bf o_p(1)}$. Therefore,  $p_2^{-1}\tilde\Cb^\top\Cb=\bGamma_2^\top+{\bf o_p(1)}$. Further by Slutsky's theorem,
	\[
	 \sqrt{Tp_1}(\tilde\bC_j-\tilde\Hb_2^\top\bC_j)\stackrel{d}{\rightarrow}\mathcal{N}({\bf 0}, \bLambda_2^{-1}\bGamma_2^\top\Vb_{2j}\bGamma_2\bLambda_2^{-1}).
	\]
	Parallel steps can be applied to verify the results of $\tilde\bR_i$, which concludes the theorem.
\end{proof}

\section{Proof of Theorem \ref{thm5}: factor and signal matrices} \label{sece}
By definition,
\[
\begin{split}
\tilde\Fb_t=&\frac{1}{p_1p_2}\tilde\Rb^\top\Xb\tilde\Cb=\frac{1}{p_1p_2}\tilde\Rb^\top\Rb\Fb_t\Cb^\top\tilde\Cb+\frac{1}{p_1p_2}\tilde\Rb^\top\Eb_t\tilde\Cb\\
=&\frac{1}{p_1p_2}\tilde\Rb^\top(\Rb-\tilde\Rb\tilde\Hb_1^{-1}+\tilde\Rb\tilde\Hb_1^{-1})\Fb_t(\Cb-\tilde\Cb\tilde\Hb_2^{-1}+\tilde\Cb\tilde\Hb_2^{-1})^\top\tilde\Cb\\
&+\frac{1}{p_1p_2}(\tilde\Rb-\Rb\tilde\Hb_1+\Rb\tilde\Hb_1)^\top\Eb_t(\tilde\Cb-\Cb\tilde\Hb_2+\Cb\tilde\Hb_2).
\end{split}
\]
Note that $\tilde\Rb^\top\tilde\Rb=p_1\Ib_{k_1}$ and $\tilde\Cb^\top\tilde\Cb=p_2\Ib_{k_2}$, then
\[
\begin{split}
\tilde\Fb_t-\tilde\Hb_1^{-1}\Fb_t\tilde\Hb_2^{-1}=&\frac{1}{p_1p_2}\tilde\Rb^\top(\Rb-\tilde\Rb\tilde\Hb_1^{-1})\Fb_t(\Cb-\tilde\Cb\tilde\Hb_2^{-1})^\top\tilde\Cb\\
&+\frac{1}{p_1}\tilde\Rb^\top(\Rb-\tilde\Rb\tilde\Hb_1^{-1})\Fb_t(\tilde\Hb_2^{-1})^\top+\frac{1}{p_2}\tilde\Hb_1^{-1}\Fb_t(\Cb-\tilde\Cb\tilde\Hb_2^{-1})^\top\tilde\Cb\\
&+\frac{1}{p_1p_2}(\tilde\Rb-\Rb\tilde\Hb_1)^\top\Eb_t(\tilde\Cb-\Cb\tilde\Hb_2)\\
&+\frac{1}{p_1p_2}(\tilde\Hb_1)^\top\Rb^\top\Eb_t(\tilde\Cb-\Cb\tilde\Hb_2)+\frac{1}{p_1p_2}(\tilde\Rb-\Rb\tilde\Hb_1)^\top\Eb_t\Cb\tilde\Hb_2\\
&+\frac{1}{p_1p_2}(\tilde\Hb_1)^\top\Rb^\top\Eb_t\Cb\tilde\Hb_2.
\end{split}
\]
In Lemma \ref{leme1}, we prove that
	\[
\begin{split}
\frac{1}{p_1}\Rb^\top(\tilde\Rb-\Rb\tilde\Hb_1)=&O_p\bigg(\frac{1}{Tp_1}+\frac{1}{Tp_2}+\frac{1}{p_1p_2}+\frac{1}{\sqrt{Tp_1p_2}}\bigg),\\ \frac{1}{p_2}\Cb^\top(\tilde\Cb-\Cb\tilde\Hb_2)=&O_p\bigg(\frac{1}{Tp_1}+\frac{1}{Tp_2}+\frac{1}{p_1p_2}+\frac{1}{\sqrt{Tp_1p_2}}\bigg),
\end{split}
\]
Therefore, by Cauchy-Schwartz inequality, Corollary \ref{cor1}, and the bounds for $\|\Rb^\top\Eb_t\|_F^2,\|\Eb_t\Cb\|_F^2$ and $\|\Rb^\top\Eb_t\Cb\|_F^2$,  we have
\[
\|\tilde\Fb_t-\tilde\Hb_1^{-1}\Fb_t\tilde\Hb_2^{-1}\|\le O_p\bigg(\frac{1}{\sqrt{T}\times \min\{p_1,p_2\}}+\frac{1}{\sqrt{p_1p_2}}\bigg).
\]
Next, for any $t,i,j$,
\[
\begin{split}
\tilde S_{t,ij}-S_{t,ij}=&\tilde\Rb_{i\cdot}^\top\tilde\Fb_t\tilde\Cb_{j\cdot}-\Rb_{i\cdot}^\top\Fb_t\Cb_{j\cdot}\\
=&(\tilde\Rb_{i\cdot}-\tilde\Hb_1^\top\Rb_{i\cdot}+\tilde\Hb_1^\top\Rb_{i\cdot})^\top\tilde\Fb_t(\tilde\Cb_{j\cdot}-\tilde\Hb_2^\top\Cb_{j\cdot}+\tilde\Hb_2^\top\Cb_{j\cdot})-\Rb_{i\cdot}^\top\Fb_t\Cb_{j\cdot}\\
=&(\tilde\Rb_{i\cdot}-\tilde\Hb_1^\top\Rb_{i\cdot})^\top\tilde\Fb_t(\tilde\Cb_{j\cdot}-\tilde\Hb_2^\top\Cb_{j\cdot})+\Rb_{i\cdot}^\top\tilde\Hb_1\tilde\Fb_t(\tilde\Cb_{j\cdot}-\tilde\Hb_2^\top\Cb_{j\cdot})\\
&+(\tilde\Rb_{i\cdot}-\tilde\Hb_1^\top\Rb_{i\cdot})^\top\tilde\Fb_t\tilde\Hb_2^\top\Cb_{j\cdot}+\Rb_{i\cdot}^\top(\tilde\Hb_1\tilde\Fb_t\tilde\Hb_2-\Fb_t)\Cb_{j\cdot}.
\end{split}
\]
Then, by Cauchy-Schwartz inequality, Corollary \ref{cor1} and the consistency of the factor matrix, we have
\[
|\tilde S_{t,ij}-S_{t,ij}|=O_p\bigg(\frac{1}{\min\{\sqrt{p_1p_2},\sqrt{Tp_1},\sqrt{Tp_2}\}}\bigg),
\]
which concludes the theorem.

\begin{lemma}\label{leme1}
	Under Assumptions A-F, take $\hat\Rb$ and $\hat\Cb$ as projection matrices, then we have
	\[
	\begin{split}
	 \frac{1}{p_1}\Rb^\top(\tilde\Rb-\Rb\tilde\Hb_1)=&O_p\bigg(\frac{1}{Tp_1}+\frac{1}{Tp_2}+\frac{1}{p_1p_2}+\frac{1}{\sqrt{Tp_1p_2}}\bigg),\\ \frac{1}{p_2}\Cb^\top(\tilde\Cb-\Cb\tilde\Hb_2)=&O_p\bigg(\frac{1}{Tp_1}+\frac{1}{Tp_2}+\frac{1}{p_1p_2}+\frac{1}{\sqrt{Tp_1p_2}}\bigg),
	\end{split}
	\]
	as $\min\{T,p_1,p_2\}\rightarrow\infty$.
\end{lemma}
\begin{proof}
	We only prove the result for $\tilde\Cb$. By equation (\ref{equa2})
	\[
	 \frac{1}{p_2}\Cb^\top(\tilde\Cb-\Cb\tilde\Hb_2)=\frac{1}{p_2}\Cb^\top(\mathcal{\uppercase\expandafter{\romannumeral2}}+\mathcal{\uppercase\expandafter{\romannumeral3}}+\mathcal{\uppercase\expandafter{\romannumeral4}})\tilde\Cb\tilde\bLambda_2^{-1}.
	\]
	Ignore $\tilde\bLambda_2$, firstly,
	\[
	\begin{split}
	 \bigg\|\frac{1}{p_2}\Cb^\top\mathcal{\uppercase\expandafter{\romannumeral2}}\tilde\Cb\bigg\|=&\bigg\|\frac{1}{Tp_1^2p_2^2}\sum_t\Cb^\top\Eb_t^\top\hat\Rb\hat\Rb^\top\Rb\Fb_t\Cb^\top\tilde\Cb\bigg\|\asymp \bigg\|\frac{1}{Tp_1p_2}\sum_t\Cb^\top\Eb_t^\top\hat\Rb\Fb_t\bigg\|\\
	 =&O_p\bigg(\frac{1}{\sqrt{Tp_1p_2}}+\frac{1}{Tp_2}\bigg),
	\end{split}
	\]
	where the last equality is by a similar decomposition as equation (\ref{equb3}).
	Secondly,
	\[
	\begin{split}
	 \bigg\|\frac{1}{p_2}\Cb^\top\mathcal{\uppercase\expandafter{\romannumeral3}}\tilde\Cb\bigg\|=&\bigg\|\frac{1}{Tp_1^2p_2^2}\sum_t\Cb^\top\Cb\Fb_t^\top\Rb^\top\hat\Rb\hat\Rb^\top\Eb_t\tilde\Cb\bigg\|\asymp \bigg\|\frac{1}{Tp_1p_2}\sum_t\Fb_t^\top\hat\Rb^\top\Eb_t\tilde\Cb\bigg\|\\
	 =&O_p\bigg(\frac{1}{\sqrt{Tp_1p_2}}+\frac{1}{Tp_2}\bigg)\times \bigg(1+\sqrt{\frac{p_2}{Tp_1}}\bigg).
	\end{split}
	\]
	Thirdly,
	\[
	\begin{split}
	 \bigg\|\frac{1}{p_2}\Cb^\top\mathcal{\uppercase\expandafter{\romannumeral4}}\tilde\Cb\bigg\|=&\bigg\|\frac{1}{Tp_1^2p_2^2}\sum_t\Cb^\top\Eb_t^\top\hat\Rb\hat\Rb^\top\Eb_t\tilde\Cb\bigg\|\\
	\le& \frac{1}{p_1p_2}\sqrt{\frac{1}{Tp_1p_2}\sum_t\|\Cb^\top\Eb_t^\top\hat\Rb\|_F^2\times \frac{1}{Tp_1p_2}\sum_t\|\hat\Rb^\top\Eb_t\tilde\Cb\|_F^2}\\
	=&O_p\bigg(\frac{1}{p_1p_2}\times \sqrt{(1+p_1w_1)\times (1+p_1w_1)(1+\|\tilde\Cb-\Cb\tilde\Hb_2\|_F^2)}\bigg)\\
	 =&O_p\bigg(\frac{1}{p_1p_2}+\frac{1}{Tp_2^2}\bigg)+o_p\bigg(\frac{1}{\sqrt{Tp_1p_2}}\bigg).
	\end{split}
	\]
	Combine the above results so that
	\[
	 \frac{1}{p_2}\Cb^\top(\tilde\Cb-\Cb\tilde\Hb_2)=O_p\bigg(\frac{1}{Tp_1}+\frac{1}{Tp_2}+\frac{1}{p_1p_2}+\frac{1}{\sqrt{Tp_1p_2}}\bigg),
	\]
	which concludes the lemma.
\end{proof}

\section{Proof of Theorem \ref{thm6}: determination of the factor numbers} \label{secf}
\begin{proof}
We only prove the first part of this theorem, i.e.,  $\Pr(\hat k_2^{(t)}=k_2)\rightarrow 1$ when  $\hat k_1^{(t)}=j\in[k_1,k_{\max}]$. Another part of this theorem can be proved similarly.  Firstly, if $\hat k_1^{(t)}>k_1$, we assume that $\hat k_1^{(t)}=k_1+1$ in our proof without loss of generality  because $k_{\max}$ is a constant. Under such case,  by definition $\hat\Rb^{(t)}=(\hat\Rb,\hat\bgamma)$, where $\hat\Rb$ is the first-stage estimator with true $k_1$, and $\hat\bgamma/\sqrt{p_1}$ is the $(k_1+1)$-th eigenvector of $\hat\Mb_1$.  $\hat\Rb^\top \hat\bgamma={\bf 0}$. Then,
\[
\begin{split}
\tilde\Mb_2^{(t)}=&\frac{1}{Tp_1^2p_2}\sum_{t=1}^{T}\Xb_t^\top\hat\Rb^{(t)}(\hat\Rb^{(t)})^\top\Xb_t=\frac{1}{Tp_1^2p_2}\sum_{t=1}^{T}\Xb_t^\top(\hat\Rb,\hat\bgamma)(\hat\Rb,\hat\bgamma)^\top\Xb_t\\
=&\tilde\Mb_2+\frac{1}{Tp_1^2p_2}\sum_{t=1}^{T}\Xb_t^\top\hat\bgamma\hat\bgamma^\top\Xb_t.\\
\end{split}
\]

By Lemma \ref{lema2}, we actually have
\[
\lambda_j(\tilde\Mb_2)=\left\{
\begin{aligned}
&\lambda_j(\bSigma_2)+o_p(1),&j\le k_2,\\
&O_p\bigg(\frac{1}{\sqrt{Tp_1}}+\frac{1}{Tp_2}+\frac{1}{p_1}\bigg),&j>k_2.
\end{aligned}
\right.
\]
Note that by the definition of eigenvector, $\lambda_{k_1+1}(\hat\Mb_1)\hat\bgamma=\hat\Mb_1\hat\bgamma$. Hence,
\[
\bigg\|\frac{1}{Tp_1^2p_2}\sum_{t=1}^{T}\Xb_t^\top\hat\bgamma\hat\bgamma^\top\Xb_t\bigg\|\le\frac{1}{Tp_1^2p_2}\sum_{t=1}^{T}\hat\bgamma^\top\Xb_t\Xb_t^\top\hat\bgamma=\frac{1}{p_1}\hat\bgamma^\top\hat\Mb_1\hat\bgamma=\lambda_{k_1+1}(\hat\Mb_1)=O_p\bigg(\frac{1}{\sqrt{Tp_2}}+\frac{1}{p_1}\bigg),
\]
where the last equality is from Lemma \ref{lemb1}. Therefore, by Weyl's theorem,
\[
\lambda_j(\tilde\Mb_2^{(t)})=\left\{
\begin{aligned}
&\lambda_j(\bSigma_2)+o_p(1),&j\le k_2,\\
&O_p\bigg(\frac{1}{\sqrt{Tp_1}}+\frac{1}{\sqrt{Tp_2}}+\frac{1}{p_1}\bigg),&j>k_2.
\end{aligned}
\right.
\]
Now we can calculate the eigenvalue ratios. Let $\delta=\max\{(Tp_1)^{-0.5},(Tp_2)^{-0.5},(p_1)^{-1}\}$, then
\[
\begin{split}
\max_{j\le k_2-1}\frac{\lambda_j(\tilde\Mb_2^{(t)})}{\lambda_{j+1}(\tilde\Mb_2^{(t)})+c\delta}=&O_p(1),\\
\max_{j\ge k_2+1}\frac{\lambda_j(\tilde\Mb_2^{(t)})}{\lambda_{j+1}(\tilde\Mb_2^{(t)})+c\delta}\le&O_p(1),\\
\frac{\lambda_j(\tilde\Mb_2^{(t)})}{\lambda_{j+1}(\tilde\Mb_2^{(t)})+c\delta}\bigg|_{j=k_2}\ge&c\delta^{-1}\rightarrow\infty,\\
\end{split}
\]
which concludes the consistency.

As one reviewer has pointed out, the  maximized eigenvalue ratio is of rate $O_p(\min\{\sqrt{T},p\})$ in conventional vector factor models. Hence, if we pile down the matrix observations into  vectors, the expected ``optimal'' eigenvalue ratio is of rate $O_p(\sqrt{T},p_1p_2)$, which is larger than $\delta^{-1}$ when $p_1$ is small. A larger ratio implies a better separation of the spiked eigenvalues, which leads to better estimation of the number of factors.  We will explain the reason why the maximized eigenvalue ratio of the iterative method is not ``optimal''.

In brief, it is mainly because there are two parameter to estimate, $k_1$ and $k_2$. The estimation of $k_1$ brings new error to $\hat k_2$. Actually, if $k_1$ is known, the maximized eigenvalue ratio in estimating $k_2$ will be larger than the typical rate $O_p(\sqrt{T},p_1p_2)$. To see this, we will prove that the convergence rates for $j>k_2$ in Lemma \ref{lema2} can be improved when $k_1$ is given and the projection matrix is $\hat\Rb$.  Specifically, sharper bound is available for $\|\mathcal{\uppercase\expandafter{\romannumeral4}}\|$ in Lemma \ref{lema2}. Note that
\[
\begin{split}
\mathcal{\uppercase\expandafter{\romannumeral4}}=&\frac{1}{Tp_1^2p_2}\sum_t\Eb_t^\top\hat\Rb\hat\Rb^\top\Eb_t\\
=&\frac{1}{Tp_1^2p_2}\sum_t\Eb_t^\top\Rb\hat\Hb_1\hat\Hb_1^\top\Rb^\top\Eb_t+\frac{1}{Tp_1^2p_2}\sum_t\Eb_t^\top(\hat\Rb-\Rb\hat\Hb_1)\hat\Hb_1^\top\Rb^\top\Eb_t\\
&+\frac{1}{Tp_1^2p_2}\sum_t\Eb_t^\top\Rb\hat\Hb_1(\hat\Rb-\Rb\hat\Hb_1)^\top\Eb_t+\frac{1}{Tp_1^2p_2}\sum_t\Eb_t^\top(\hat\Rb-\Rb\hat\Hb_1)(\hat\Rb-\Rb\hat\Hb_1)^\top\Eb_t
\end{split}
\]
As before, we assume $k_1=1$ to avoid fixed-dimensional matrix-multiplication, then
\[
\begin{split}
(1).&\bigg\|\frac{1}{Tp_1^2p_2}\sum_t\Eb_t^\top\Rb\hat\Hb_1\hat\Hb_1^\top\Rb^\top\Eb_t\bigg\|^2\lesssim \bigg\|\frac{1}{Tp_1^2p_2}\sum_t\Eb_t^\top\Rb\Rb^\top\Eb_t\bigg\|^2\\
\lesssim&\bigg\|\frac{1}{Tp_1^2p_2}\sum_t\Big(\Eb_t^\top\Rb\Rb^\top\Eb_t-\mathbb{E}\Eb_t^\top\Rb\Rb^\top\Eb_t\Big)\bigg\|_F^2+\bigg\|\frac{1}{Tp_1^2p_2}\sum_t\mathbb{E}\Eb_t^\top\Rb\Rb^\top\Eb_t\bigg\|^2\\
\le&\frac{\|\Rb\|_F^2}{T^2p_1^4p_2^2}\sum_{i=1}^{p_1}\sum_{j=1}^{p_2}\bigg\|\sum_t\Big(e_{t,ij}\Rb^\top\Eb_t-\mathbb{E}e_{t,ij}\Rb^\top\Eb_t\Big)\bigg\|_F^2+\bigg\|\frac{1}{Tp_1^2p_2}\sum_t\mathbb{E}\Eb_t^\top\Rb\Rb^\top\Eb_t\bigg\|_1^2\\
\lesssim&\frac{1}{T^2p_1^3p_2^2}\sum_{ij}\sum_{t,i_1,j_1}\sum_{s,i_2}\bigg|\text{Cov}(e_{t,ij}e_{t,i_1j_1},e_{s,ij}e_{t,i_2j_1})\bigg|+\bigg(\frac{1}{Tp_1^2p_2}\sum_t\max_j\sum_{j_1}\bigg|\mathbb{E}\sum_{i_1,i_2}e_{t,i_1j_1}e_{t,i_2j}\bigg|\bigg)^2\\
\lesssim & \frac{1}{Tp_1}+\bigg(\frac{1}{Tp_1^2p_2}\sum_t\max_j\sum_{i_2}\sum_{i_1,j_1}\bigg|\mathbb{E}e_{t,i_1j_1}e_{t,i_2j}\bigg|\bigg)^2 \quad \Big(\text{ by Assumption D.3(2)}\Big)\\
=&O_p\bigg(\frac{1}{Tp_1}+\frac{1}{p_1^2p_2^2}\bigg). \quad \Big(\text{ by Assumption D.2(1)}\Big)
\end{split}
\]
\[
\begin{split}
(2).&\bigg\|\frac{1}{Tp_1^2p_2}\sum_t\Eb_t^\top(\hat\Rb-\Rb\hat\Hb_1)\hat\Hb_1^\top\Rb^\top\Eb_t\bigg\|^2\lesssim \|\hat\Rb-\Rb\hat\Hb_1\|_F^2\sum_{i}\bigg\|\frac{1}{Tp_1^2p_2}\sum_t\be_{t,i\cdot}\Rb^\top\Eb_t\bigg\|^2\\
\lesssim&p_1w_1\bigg(\sum_{i,j}\bigg\|\frac{1}{Tp_1^2p_2}\sum_t(e_{t,ij}\Rb^\top\Eb_t-\mathbb{E}e_{t,ij}\Rb^\top\Eb_t)\bigg\|^2+\sum_{i}\bigg\|\frac{1}{Tp_1^2p_2}\sum_t\mathbb{E}\be_{t,i\cdot}\Rb^\top\Eb_t\bigg\|^2\bigg)\\
=& O_p\bigg(\frac{w_1}{Tp_1}+\frac{w_1}{p_1^2p_2^2}\bigg).\quad \Big(\text{ similarly to (1)}\Big)
\end{split}
\]
\[
\begin{split}
(3).&\bigg\|\frac{1}{Tp_1^2p_2}\sum_t\Eb_t^\top(\hat\Rb-\Rb\hat\Hb_1)(\hat\Rb-\Rb\hat\Hb_1)^\top\Eb_t\bigg\|^2\lesssim \|\hat\Rb-\Rb\hat\Hb_1\|_F^4\sum_{i_1,i_2}\bigg\|\frac{1}{Tp_1^2p_2}\sum_t\be_{t,i_1\cdot}\be_{t,i_2\cdot}^\top\bigg\|^2\\
\lesssim&w_1^2\sum_{i_1,i_2}\bigg(\bigg\|\frac{1}{Tp_1p_2}\sum_t(\be_{t,i_1\cdot}\be_{t,i_2\cdot}^\top-\mathbb{E}\be_{t,i_1\cdot}\be_{t,i_2\cdot}^\top)\bigg\|_F^2+\bigg\|\frac{1}{Tp_1p_2}\sum_t\mathbb{E}\be_{t,i_1\cdot}\be_{t,i_2\cdot}^\top\bigg\|^2\bigg)\\
=& O_p\bigg(\frac{w_1^2}{T}+\frac{w_1^2}{p_2^2}\bigg).\quad \Big(\text{ similarly to (1)}\Big)
\end{split}
\]
Note that $w_1\lesssim (Tp_2)^{-1}+p_1^{-2}$, then $(1)$, $(2)$ and $(3)$ imply
\[
\| \mathcal{\uppercase\expandafter{\romannumeral4}}\|\le O_p\bigg(\frac{1}{\sqrt{Tp_1}}+\frac{1}{p_1p_2}+\frac{1}{Tp_2}\bigg)
\]
Combined with the rates of $\|\mathcal{\uppercase\expandafter{\romannumeral2}}\|$ and $\|\mathcal{\uppercase\expandafter{\romannumeral3}}\|$, we have
\begin{equation}\label{eque1}
\lambda_j(\tilde\Mb_2)=\left\{
\begin{aligned}
&\lambda_j(\bSigma_2)+o_p(1),&j\le k_2,\\
&O_p\bigg(\frac{1}{\sqrt{Tp_1}}+\frac{1}{Tp_2}+\frac{1}{p_1p_2}\bigg),&j>k_2.
\end{aligned}
\right.
\end{equation}
We find that the rate for $j>k_2$ is exactly squared root of those in Corollary \ref{cor1}. Hence, Corollary \ref{cor1} can  also be deduced by an application of the Davis-Kahan's $\sin(\Theta)$  theorem under (\ref{eque1}).

 Note that if $k_1$ is given, we use $\lambda_j(\tilde \Mb_2)$ to estimate $k_2$ rather than the iterative procedure. The asymptotic negligible  term added to the denominator should also be modified as $c\times \max\{(Tp_1)^{-0.5}, (Tp_2)^{-1},(p_1p_2)^{-1}\}$ accordingly. As a result, the maximized eigenvalue ratio is of rate $\min\{\sqrt{Tp_1}, Tp_2,p_1p_2\}$, which is no smaller than the typical rate $\min\{\sqrt{T},p_1p_2\}$ for vectorized models. This can be another advantage by assuming a matrix factor model, because it simplifies the loading structure.
\end{proof}

\section{Description of real data sets}\label{secg}
Figures \ref{fig:a1} and \ref{fig:a2} plot the series  of the Fama-French 100 portfolio data set and the multinational macroeconomic indices data set after preprocessing and standardization. Table \ref{tab:a1}  shows the countries and corresponding short names in the macroeconomic data set. Table \ref{tab:a2}  shows the indices, labels in OECD data base, preprocessing transformations,  and variable definitions in the macroeconomic data set.

\begin{figure}[hbpt]
	\centering
	 \includegraphics[width=15cm,height=15cm]{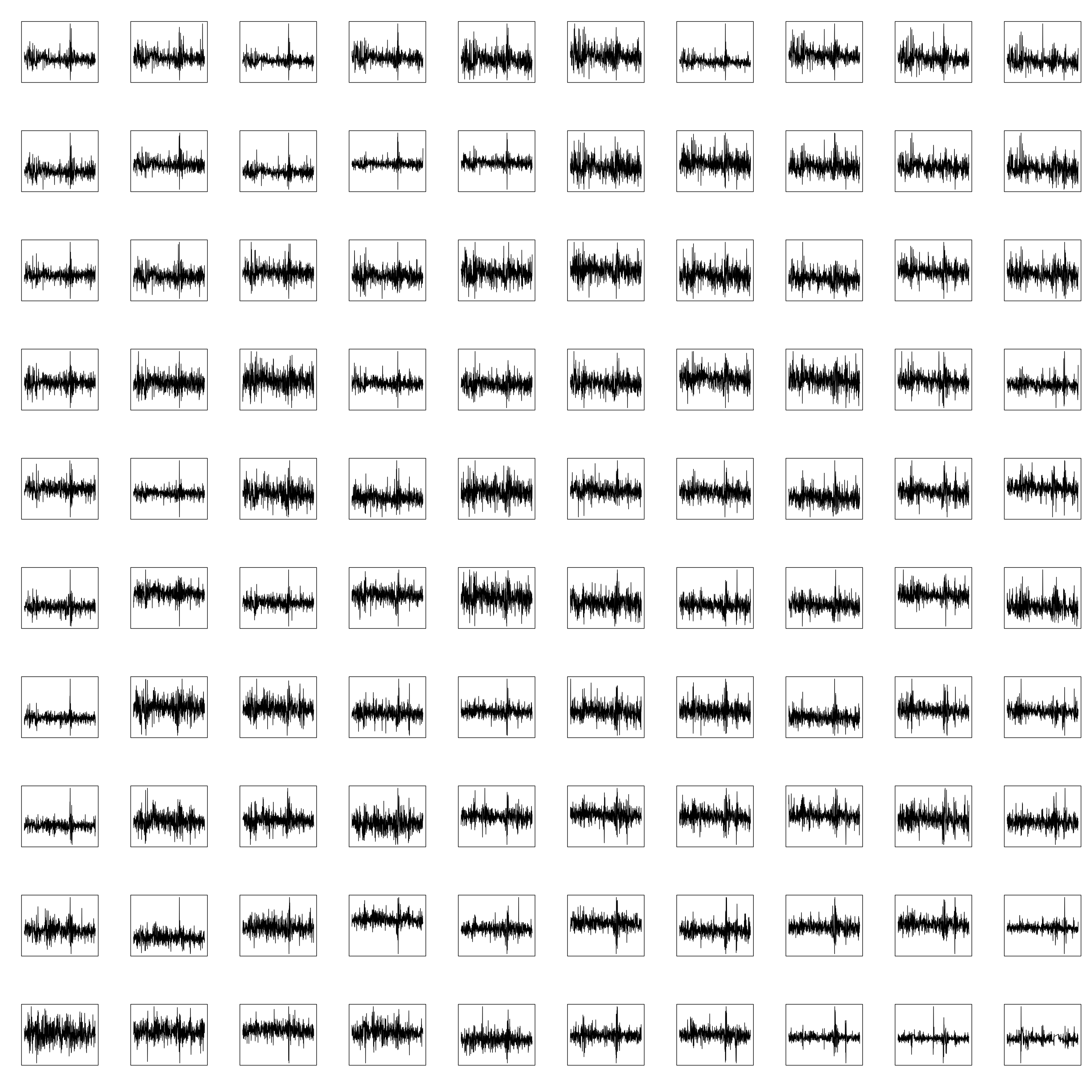}
	\caption{Fama-French 100 portfolio data set after preprocessing and standardization.}
	\label{fig:a1}
\end{figure}

\begin{figure}[hbpt]
	\centering
	 \includegraphics[width=15cm,height=12cm]{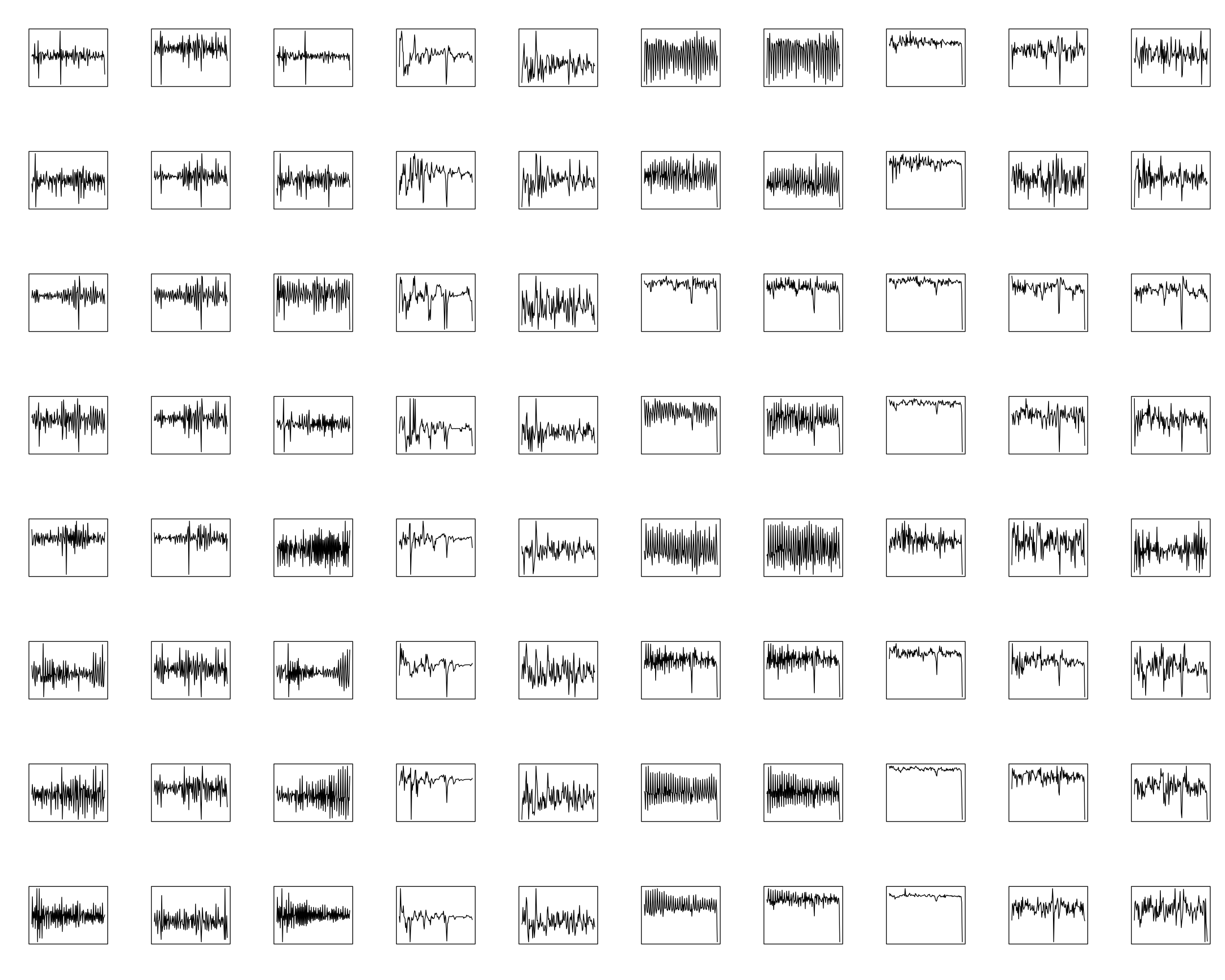}
	\caption{Multinational macroeconomic indices  data set after preprocessing and standardization.}
	\label{fig:a2}
\end{figure}

\begin{table*}[hbtp]
	\begin{center}
		\small
		\addtolength{\tabcolsep}{8pt}
		\caption{Countries in the macroeconomic data set.}\label{tab:a1}
		 \renewcommand{\arraystretch}{1.2}
		\scalebox{1}{ 	
			\begin{tabular*}{9cm}{cc}
				\toprule[1.2pt]
				Country& Code\\\midrule[1pt]
				Australia& AUS\\
				New Zealand& NZL\\
				United States of America& USA\\
				Canada& CAN\\
				Norway& NOR\\
				Germany& DEU\\
				France& FRA\\
				United Kingdom&GBR\\
				\bottomrule[1.2pt]		
		\end{tabular*}}		
	\end{center}
\end{table*}

\begin{table*}[hbtp]
	\begin{center}
		\small
		\addtolength{\tabcolsep}{0pt}
		\caption{Indices in the macroeconomic data set. All indices except interest rates are measured by taking the year 2015 as 100. }\label{tab:a2}
		 \renewcommand{\arraystretch}{1.2}
		\scalebox{0.85}{ 	
			 \begin{tabular*}{19cm}{llcl}
				\toprule[1.2pt]
				Short name& Label&Transformation& Definition\\\midrule[1pt]
				 CPI:Tot&CPALTT01&$\Delta^2\ln$&Consumer Price Index: Total\\
				 CPI:Ener&CPGREN01&$\Delta^2\ln$&Consumer Price Index: Energy\\		
				 CPI:NFNE&CPGRLE01&$\Delta^2\ln$&Consumer Price Index: All items no food no energy\\	
				 IR:3-Mon&IR3TIB01&$\Delta$&Interest Rates: 3-month or 90-day rates and yields,  interbank\\
				 IR:Long&IRLTLT01&$\Delta$&Interest Rates: Long-term government bond yields, 10-year\\
				 P:TIEC&PRINTO01&$\Delta\ln$&Production: Total industry excluding construction\\
				  P:TM&PRMNTO01&$\Delta\ln$&Production: Total industry excluding construction\\
				   GDP& LORSGPOR& $\Delta\ln$& GDP: Original series\\
				    IT:Ex&XTEXVA01&$\Delta\ln$		 &International Trade: Total Exports Value (goods)\\
				     IT:Im&		 XTIMVA01& $\Delta\ln$		 &International Trade: Total Imports Value (goods)
				\\\bottomrule[1.2pt]		
		\end{tabular*}}		
	\end{center}
\end{table*}

\end{appendices}
\end{document}